\documentclass[journal]{IEEEtran}
\pdfoutput=1
\usepackage[T1]{fontenc}
\usepackage{amsmath,amsfonts}
\usepackage{algorithmic}
\usepackage[boxed,ruled,lined]{algorithm2e}
\usepackage{bigstrut}
\usepackage{wasysym}
\usepackage{enumitem}
\usepackage{array}
\usepackage[caption=false,font=normalsize,labelfont=sf,textfont=sf]{subfig}
\usepackage{textcomp}
\usepackage{stfloats}
\usepackage{url}
\usepackage{verbatim}
\usepackage{graphicx}
\usepackage{cite}
\usepackage{lscape}
\usepackage[utf8]{inputenc} 
\usepackage[T1]{fontenc}    

\usepackage{bm}
\usepackage{amssymb}
\usepackage{xcolor}
\usepackage{amsthm}
\usepackage{xcolor}
\usepackage{threeparttable}
\usepackage{multicol}
\usepackage{multirow}
\usepackage{rotating} 
\theoremstyle{plain}

\usepackage{amsfonts}       
\usepackage{nicefrac}       
\usepackage{microtype}      
\usepackage{xcolor}         
\usepackage{tikz}
\usetikzlibrary{snakes}
\usetikzlibrary{trees}
\usepackage{wrapfig}
\usepackage{makecell}
\usepackage{algorithmic}
\usepackage[boxed,ruled,lined]{algorithm2e}
\usepackage{bigstrut}
\usepackage{wasysym}
\usepackage{color, colortbl}
\usepackage{xr-hyper}
\usepackage{float}
\usepackage{mathtools}
\usepackage{booktabs}
\usepackage[normalem]{ulem}
\useunder{\uline}{\ul}{}
\usepackage{titlesec}
\usepackage{setspace}
\usepackage{hyperref}       

\titlespacing\subsection{0pt}{5pt plus 2pt minus 2pt}{3pt plus 2pt minus 2pt}
\def\x{\bm{x}}
\def\w{\bm{w}}

\long\def\comment#1{}
\def\ie{$i.e.$}
\def\eg{$e.g.$}

\def\wrt{$w.r.t.$}
\def\etal{$et~al.$}

\def\blue#1{\textcolor{black}{#1}}

\def\bluetable{\color{black}}
\def\bluetwo#1{\textcolor{black}{#1}}

\begin{document}

\newtheorem{theorem}{\textbf{Theorem}}
\newtheorem{definition}{Definition} 
\newtheorem{lemma}{\textbf{Lemma}} 
\newtheorem{corollary}{\textbf{Corollary}}
\newtheorem{example}{\textbf{Example}}
\newtheorem{proposition}{\textbf{Proposition}}
\newtheorem{remark}{\textbf{Remark}}

\title{Seeking Flat Minima over Diverse Surrogates for Improved Adversarial Transferability: A Theoretical Framework and Algorithmic Instantiation}

\author{
Meixi Zheng, Kehan Wu, Yanbo Fan, Rui Huang, \IEEEmembership{Member,~IEEE}, Baoyuan Wu\textsuperscript{\dag}, \IEEEmembership{Senior Member,~IEEE}
\\

\IEEEcompsocitemizethanks{
\IEEEcompsocthanksitem 
Meixi Zheng, Kehan Wu and Baoyuan Wu are with School of Artificial Intelligence, The Chinese University of Hong Kong, Shenzhen, Guangdong, 518172, P.R. China.
Rui Huang is with School of Science and Engineering, The Chinese University of Hong Kong, Shenzhen, Guangdong, 518172, P.R. China.
Yanbo Fan is with Nanjing University (Suzhou Campus), China.
\IEEEcompsocthanksitem 
Email: meixizheng1@link.cuhk.edu.cn, kehanwu1@link.cuhk.edu.cn, fanyanbo0124@gmail.com, ruihuang@cuhk.edu.cn, wubaoyuan@cuhk.edu.cn.
\IEEEcompsocthanksitem 
\textsuperscript{\dag}Corresponding author: Baoyuan Wu (wubaoyuan@cuhk.edu.cn).
}
}

\markboth{Journal of \LaTeX\ Class Files,~Vol.~xx, No.~xx, August~xxxx}%
{Shell \MakeLowercase{\textit{et al.}}: A Sample Article Using IEEEtran.cls for IEEE Journals}


\maketitle

\begin{abstract}
The transfer-based black-box adversarial attack setting poses the challenge of crafting an adversarial example (AE) on known surrogate models that \bluetwo{remains} effective against unseen target models. Due to the practical importance of this task, numerous methods have been proposed to address this challenge. However, most previous methods are heuristically designed and intuitively justified, lacking a theoretical foundation.
To bridge this gap, we derive a novel transferability bound that offers provable guarantees for adversarial transferability. Our theoretical analysis has the advantages of \textit{(i)} deepening our understanding of previous methods by building a general attack framework and \textit{(ii)} providing guidance for designing an effective attack algorithm.
Our theoretical results demonstrate that optimizing AEs toward flat minima over the surrogate model set, while controlling the surrogate-target model shift measured by the adversarial model discrepancy, yields a comprehensive guarantee for AE transferability.
The results further lead to a general transfer-based attack framework, within which we observe that previous methods consider only partial factors contributing to the transferability. 
Algorithmically, inspired by our theoretical results, we first elaborately construct the surrogate model set in which models exhibit diverse adversarial vulnerabilities with respect to AEs to narrow the instantiated adversarial model discrepancy. Then, a \textit{model-Diversity-compatible Reverse Adversarial Perturbation} (DRAP) is generated to effectively promote the flatness of AEs over diverse surrogate models to improve transferability. 
Extensive experiments on NIPS2017 and CIFAR-10 datasets against various target models demonstrate the effectiveness of our proposed attack.
The code is publicly available\footnote{Code: \url{https://github.com/SCLBD/blackboxbench}}.
\end{abstract}

\begin{IEEEkeywords}
Black-box adversarial attack, adversarial transferability, flatness, model discrepancy.
\end{IEEEkeywords}

\section{Introduction}

\IEEEPARstart{D}{eep} neural networks (DNNs) are vulnerable to adversarial examples (AEs), where attackers add imperceptible perturbations to benign examples but make a model produce erroneous predictions \cite{cgattack, mcg, cisa, prgf, attacksurvey, diffattack}. Under the black-box setting, attackers have no information regarding possible future target models, and the adversarial transferability matters since it allows attackers to attack target models by alternatively generating AEs from the surrogate models. However, as the attacker \bluetwo{cannot} access information of the target model, a potentially unmatched surrogate model may lead to rather limited attack capability of the transferred AE against the target model.

Previous works \cite{nisi} attributed the unsatisfactory transferability to the overfitting of AEs to the surrogate models. In turn, it is essential that such AEs be optimized using methods that ensure that crafted perturbations do in fact transfer beyond the surrogate models. A plethora of transfer-based black-box adversarial attack methods have been proposed \cite{difgsm, nisi, mifgsm, vt, cwa, tsea, lpm}. There \bluetwo{has been substantial} advances in improving AEs' transferability from optimization, feature, input-transformation, and model perspectives. Despite the progress made, the adversarial transferability suffers from a lack of a general theoretical understanding. As a result, the literature relies heavily on empirical heuristics, without theoretical guarantees. \textit{Can we build the theoretical foundation to deepen our understanding of transfer-based attacks?}

To tackle this problem, in this paper we present a novel theoretical analysis of transfer-based attacks towards generalizing previous works and explicitly guiding algorithm design by deriving a transferability bound. 
We start by formalizing the task of interest as crafting an AE that successfully attacks on the target model distribution. This idea consists of minimizing a \textit{target adversarial risk}, which corresponds to the expected error of an AE over the target model distribution. We then decompose it into a \textit{surrogate adversarial risk} and a \textit{transferability gap}. 
The surrogate adversarial risk measures the expected error over the surrogate model distribution and could be upper bound estimated by its empirical version and the loss landscape sharpness at the AE (cf. Theorem \ref{theorem_pac}). 
The transferability gap accounts for the discrepancy between surrogate and target model distributions and could be upper bounded in terms of a novel discrepancy, the adversarial model discrepancy, which is based on a variational representation that lower bounds $\phi$-divergences \cite{birrell} and is tailored to capture ``adversarially significant'' distribution differences (cf. Theorem \ref{theorem1}). 
Combining the two bounds, we derive a transferability bound on target adversarial risk which provides a theoretical guarantee on the adversarial transferability (cf. Theorem \ref{theorem_main}), which demonstrates that the transferability can be expected if one seeks a flat minimum of empirical surrogate adversarial risk, where the surrogate distribution is built to control the adversarial model discrepancy.
This bound further implies that the adversarial transferability of AEs has a positive correlation with three key factors simultaneously: (1) the white-box attack performance of the AE, (2) the regularization for surrogate models, and (3) the $\phi$-divergences between the surrogate and target model distributions, resulting in an attack framework generalizing previously popular attacks as special cases (cf. Equation \ref{equ_framework}).
By comparing our bound with these methods through the lens of this framework, we find they typically control only one or two key factors of this framework, neither of which is desirable nor sufficient to achieve satisfactory transferability. However, our bound considers the surrogate adversarial risk and transferability gap jointly and properly, providing a more comprehensive guarantee on the transferability of AEs.

From an algorithmic perspective, \blue{we design a novel transfer-based adversarial attack inspired by our theoretical results. In particular, by instantiating the derived bound with total variation (TV), Kullback-Leibler (KL) and $\chi^2$ divergences (cf. Corollaries \ref{lemma_tv},\ref{corollary:KL},\ref{corollary_chi2}),} we first propose to diversify the adversarial vulnerabilities in surrogate models by accounting for both between-distribution diversity and within-distribution diversity, thus controlling the surrogate-target shift. We then propose to inject a model-\textbf{D}iversity-compatible \textbf{R}everse \textbf{A}dversarial \textbf{P}erturbation (DRAP) into the attack procedure to effectively optimize the loss landscape flatness at the AE over a set of diverse surrogate models. 
We conduct extensive experiments to evaluate DRAP on NIPS2017 and CIFAR-10 datasets, covering untargeted and targeted attacks against both standard and adversarially trained models and show that (1) compared with \blue{23} state-of-the-art baseline attacks, DRAP achieves significant improvements in attack success rates; (2) DRAP is scalable to be combined with previous methods to further boost transferability; (3) \bluetwo{b}oth optimization signals \bluetwo{in DRAP}, seeking flat minima and improving diversity, contribute to the transferability, corroborating our theoretical findings.

\textbf{Contributions} This work is an extension of our previous conference paper \cite{rap}, compared to which the most significant updates and contributions are three-fold:
\begin{itemize}[leftmargin=15pt,itemsep=1pt,topsep=0pt]
    \item Theoretically, we prove that the difference between target adversarial risk and empirical surrogate adversarial risk is upper bounded by a sharpness penalty and a model discrepancy penalty. This result provides a theoretical foundation for the assumed relationship between flatness and transferability in RAP and further points out considering loss landscape flatness and model diversity in adversarial vulnerability simultaneously is exactly when this paper will bring the original RAP from the lab to the real world. 
    \item Algorithmically, we propose a theory-guided attack strategy DRAP as a correction of RAP. It generates reverse adversarial perturbations tailored to each of the diverse surrogate models, which are selected based on two dimensions of diversity, to effectively find flat local minima among them.
    \item Empirically, we demonstrate the soundness of our attack by conducting comprehensive experiments on NIPS2017 and CIFAR-10 datasets against various target models. We also perform ablative studies to further understand the contribution of the two optimization signals and to verify the relevance of our theoretical findings.
\end{itemize}

\section{Related Work}\label{sec:Related work}
Transfer-based attacks are motivated by the observation that AEs generated to deceive the surrogate model can also deceive the target model, even when their architectures differ significantly, as long as both models are solving the same task \cite{papernot2016transferability}. One of the seminal \bluetwo{works}, Iterative Fast Gradient Sign Method (I-FGSM)\cite{ifgsm}, generates adversarial examples by iteratively performing the fast gradient step, establishing a solid foundation for this area of research. However, it has been shown that I-FGSM often converges to poor local minima, resulting in low transferability \cite{difgsm}.

\textbf{Optimization-based attacks} To improve transferability, better optimization algorithms are proposed to escape from poor local minima and yield AEs with better transferability, such as MI-FGSM\cite{mifgsm}, NI-FGSM\cite{nisi} and PI-FGSM\cite{pifgsm}.
Recently, the connection between loss landscape flatness and transferability has been extensively studied empirically \cite{rap, pgn, cwa, tpa, mef, fem}. Unfortunately, few works build a clear theoretical relationship between them. Our previous work, RAP \cite{rap}, is the seminal work pursuing flatness of loss landscape for AEs. This idea is further formulated as a min-max bi-level optimization problem. PGN \cite{pgn} also intuitively assumes that AEs at flat local regions tend to have better transferability and penalizes the gradient norm. CWA\cite{cwa} derives an optimization objective involving minimization of Hessian matrix’s F-norm, thus \bluetwo{pursuing} flatness to boost transferability through a SAM-like strategy \cite{sam}.

\textbf{Feature-based attacks}
Methods from this perspective distort intermediate layer features by designing a new loss function. ILA \cite{ila} aims to use the suboptimal perturbation found by a basic attack as a proxy, deviating from it to increase the perturbation norm. Since increasing the norm in the image space is perceptible, ILA opts to increase the norm in the feature space instead. FIA \cite{fia} generates AEs by distorting object-related features, where the feature importance is defined by gradient\bluetwo{s}. Beyond FIA, NAA \cite{naa} provides more accurate measures of neuron importance.

\textbf{Input-transformation-based attacks} 
Relatedly, a wide range of methods aim to simulate diverse models by applying input transformations on benign images, thus mitigating overfitting to surrogate models. For instance, DI2-FGSM \cite{difgsm} applies random resizing and padding with a certain probability. SI-FGSM \cite{nisi} enhances transferability by scaling. Admix \cite{admix} incorporates information from images in other classes by combining two images in a master-slave manner. TI-FGSM \cite{ti} utilizes translational shifts on the input image. SSA \cite{ssa} generates diverse spectrum saliency maps to augment models, while SIA \cite{sia} applies local transformations across different regions of input to generate more diverse transformed images.

\textbf{Model-based attacks} Meanwhile, several methods have been proposed to enhance transferability from the model-centric perspective. One primary category focuses on model tuning. For instance, DRA \cite{dra} trains a score network to estimate ground-truth data score and \bluetwo{uses} the estimated score to update AE through Langevin dynamics. GhostNet \cite{ghostnet} dynamically generates a vast number of ghost networks by applying erosion to specific intermediate structures of the base network. Bayesian attack \cite{bayesian} models the Bayesian posterior of the surrogate model, enabling an ensemble of infinitely many models. Another category emphasizes fusion strategies, which aim to reconcile gradients from multiple surrogate models to better capture intrinsic transfer information \cite{cwa, adaea, smer}. Notable methods include Logit-ensemble \cite{mifgsm}, which attacks multiple models simultaneously by fusing their logit outputs and SVRE \cite{svre}, which reduces gradient variance within ensemble.

\section{Preliminaries}
\subsection{Transfer-Based Adversarial Attack} \label{sec:Notation and Attack Setup}
We first introduce the preliminaries about adversarial examples and specify a threat model under the naive transfer-based black-box setting. Let $\mathcal{X} \subseteq \mathbb{R}^d$ and $\mathcal{Y} \subseteq \mathbb{R}^k$ be the original feature space and the label space. Let $\mathcal{F}: \mathcal{X} \rightarrow \mathcal{Y}$ be the set of possible image classifiers for a given task, where each $f(\x,\w) \in \mathcal{F}$ is a classifier mapping $\mathcal{X}$ to $\mathcal{Y}$, parameterized by $\w \in \mathcal{W}$.

Consider the general setting in a black-box adversarial attack against the target model $\mathcal{M}_\mathcal{T}\in \mathcal{F}$. For a benign image $\x \in \mathcal{X}$ and its ground truth label $y\in \mathcal{Y}$, the objective of the adversary is to find a perturbation $\bm{\xi} \in \mathbb{R}^d$, leading to an adversarial example, \ie,  $\hat{\bm{x}}=\x+\bm{\xi}$, such that $\mathcal{M}_\mathcal{T}(\hat{\bm{x}})\neq y$ (untargeted attacks) or $\mathcal{M}_\mathcal{T}(\hat{\bm{x}})= y_t$ (targeted attacks) with $y_t\in \mathcal{Y}\backslash\{y\}$. Besides, due to the \bluetwo{requirement of imperceptibility}, $\hat{\bm{x}}$ should be constructed within the neighborhood of an input image $\x$, \ie, $\hat{\mathcal{X}}_{\x, \gamma} =\left\{\hat{\bm{x}}:\left\|\hat{\x}-\x\right\|_{\infty} \leq \gamma \right\}$, dubbed \blue{\textit{$L_{\infty}$-norm ball centred at $\x$ of radius $\gamma$}}. For clarity, hereafter we denote it as $\hat{\mathcal{X}}$. $\gamma \geq 0$ is a pre-defined perturbation norm budget, and $\|\cdot\|_{\infty}$ denotes the $L_{\infty}$-norm. Among all adversarial attack strategies, transfer-based attacks stem from the observation that adversarial samples crafted to deceive a white-box surrogate model set $\mathcal{M}_\mathcal{S}\subset \mathcal{F}$ have the capability to deceive a black-box target model $\mathcal{M}_\mathcal{T}$, provided that they are engaged in solving identical tasks. Generally, naive transfer-based attacks choose a single or a subset of arbitrary DNNs as surrogate models. \bluetwo{Letting} $\ell$ be the adversarial loss function, one can seek the AE by solving the constrained optimization problem on $\mathcal{M}_\mathcal{S}$:
\begin{equation}\label{base_attack}
    \underset{\hat{\bm{x}}}{\arg \min }\frac{1}{|\mathcal{M}_\mathcal{S}|} \sum_{f_i(\cdot,\w_i) \in \mathcal{M}_\mathcal{S}} \hspace{-1.2em} \ell\left(f_i(\hat{\bm{x}},\w_i), y\right), \text { s.t. }\left\|\hat{\bm{x}}-\x\right\|_{\infty} \leq \gamma.
\end{equation}
The above $\ell(\cdot, \cdot)$ is often instantiated as the negative cross-entropy function for untargeted attacks, while the cross-entropy function \wrt~the target label $y_t$ for targeted attacks.

\subsection{PAC-Bayes Bound} \label{sec:PAC-Bayes Bound}
We then introduce the PAC model. We assume a distribution $\mathcal{D}$ from which the training instances $\x_1, \x_2,\ldots, \x_n$ are independently sampled to form a set $\mathcal{M}$, a prior distribution $\mathcal{P}$ on an arbitrary concept $c \in C$ which is independent of the training set $\mathcal{M}$, and a posterior distribution $\mathcal{Q}$ on $c$ which depends on $\mathcal{M}$. Given any instance $\x$ and concept $c$, the loss function of $\x$ on $c$ is given by $\ell(\x,c)\in [0,1]$. We define risk $\ell(c)$ to be the expectation over sampling $\x$ of $\ell(\x,c)$, \ie, $\mathbb{E}_{\x \sim \mathcal{D}}[\ell(\x,c)]$, and empirical risk $\hat{\ell}(c)$ to be $\frac{1}{n}\sum_{i=1}^n \ell(\x_i,c)$.

\begin{theorem}\label{theorem_ori_pac}
    \textbf{(PAC-Bayes \cite{pac,intriguing})} For any prior distribution $\mathcal{P}$ on the concept $c$, $0<\delta<1$, with probability $1-\delta$ over the draw of training set $\mathcal{M}$ with size $n \in \mathbb{N}$, for any \bluetwo{distribution} $\mathcal{Q}$ on $c$, the following bound holds:
\begin{equation}\label{equ_ori_pac}
    \begin{aligned}
    \mathbb{E}_{\mathcal{Q}}[\ell(c)] \leq \mathbb{E}_{\mathcal{Q}}[\hat{\ell}(c)] +\sqrt{\frac{KL(\mathcal{Q}||\mathcal{P})+\log\frac{n}{\delta}}{2(n-1)}}.
\end{aligned}
\end{equation}
\end{theorem}
The PAC-Bayes theorem can bound the generalization error between the test loss $\ell(c)$ and the training loss $\hat{\ell}(c)$ of a distribution $\mathcal{Q}$ on the concept $c$ that depends on the training set, in terms of the KL divergence between $\mathcal{P}$ and $\mathcal{Q}$. 
In transfer-based adversarial attacks, it may be tempting to directly use the PAC-Bayes theorem to derive the transferability bound. However, one of the cornerstone assumptions underlying the PAC's success is that ``test'' samples should share the same distribution as ``training'' samples. \bluetwo{Unfortunately}, the independent and identically distributed (i.i.d.) assumption does not generally hold in the black-box setting. 
For instance, consider using ResNet-50 as a surrogate model and ViT as a target model: although both are trained on the same dataset, they differ substantially in architecture and training strategies. These differences induce a surrogate–target distribution shift at the model level, violating the i.i.d. assumption required by standard PAC-Bayes analysis. This model-level non-i.i.d. discrepancy contributes directly to the transferability gap and complicates the theoretical analysis. Consequently, naively applying PAC-Bayes under the i.i.d. assumption risks producing bounds that are misleading in the black-box setting.

\subsection{$\phi$-divergence} 
In light of unseen target models, we reformulate another inducement of AEs' unsatisfactory transferability as the surrogate-target model shift. A successful AE should hopefully behave robustly under the shift. A key component in tackling the shift is to study the difference between surrogate and target models. In our work, we define a new discrepancy between surrogate and target model distributions based on the variational representation of $\phi$-divergences. Here we review the definition of $\phi$-divergence and its variational representation.

\begin{definition}[\textbf{$\phi$-divergence\cite{csiszar1967information}}] \label{def_f}
Consider two probability distributions $\mu$ and $\nu$ with $\mu$ absolutely continuous \wrt~$\nu$. Assume both distributions are absolutely continuous \wrt~measure $d\w$, with densities $p_\mu$ and $p_\nu$, respectively, on domain $\mathcal{W} \subset \mathbb{R}^{|\w|}$. Let $\phi: \mathbb{R}_{+} \rightarrow \mathbb{R}$ be a convex, lower semi-continuous function satisfying $\phi(1)=0$. The $\phi$-divergence $D_\phi$ is defined as:
\begin{equation}\label{equ_f}
    D_\phi(\mu \| \nu)=\int p_\nu(\w) \phi\left(\frac{p_\mu(\w)}{p_\nu(\w)}\right) d\w.
\end{equation}
\end{definition}
$\phi$-divergence measures the difference between two given probability distributions. A large class of popular statistical divergences could be recovered from $\phi$-divergences as special cases of Equation (\ref{equ_f}). For example, $\phi(\bluetwo{x})=\frac{1}{2}|\bluetwo{x}-1|$ recovers the TV distance, \ie, $\text{TV}(\mu \| \nu)=\frac{1}{2} \int|p_\mu(\w)-p_\nu(\w)| d\w$. $\phi(\bluetwo{x})={\left(\bluetwo{x}-1\right)}^2$ recovers the $\chi^2$ divergence, \ie, $\chi^2(\mu \| \nu)=\int \frac{(p_\mu(\w)-p_\nu(\w))^2}{p_\nu(\w)} d\w$ \cite{manyf}.
Note that $\phi$-divergence also has a variational representation formula which converts its calculation into an optimization problem over a function space, offering a valuable mathematical view for the similarity between probability distributions \cite{birrell,agrawal}.

\begin{lemma}\label{lemma_f_var}
    \textbf{(Variational formula of $\phi$-divergences, Theorem 1 \cite{birrell})} Let $\phi^*$ be the Fenchel conjugate function of $\phi$, \ie, $\phi^*(t)=\sup _{x \in\mathrm{dom} \phi}\left\{xt-\phi(x)\right\})$. With ${G}$ encompassing all bounded measurable functions, let $\mathcal{G}$ be the family of functions with 
\begin{equation}
    {G} \subset \mathcal{G} \subset L^1(\mu).
\end{equation}
For any family of transformations
\begin{equation}
    \mathcal{T} \subset\left\{T=T(g), \text{ such that } T: \mathcal{G} \mapsto L^1(\mu)\right\}.
\end{equation}
Then any $\phi$-divergence can be written as:
\begin{equation}\label{equ_f_var}
    \begin{aligned}
D_\phi(\mu \| \nu) & = \sup _{g \in \mathcal{G}}\{ \sup _{T \in \mathcal{T}}\{ \mathbb{E}_{\w \sim \mu}[T\left(g(\w)\right)] \\ & -\mathbb{E}_{\w \sim \nu}\left[\phi^*\left(T\left(g(\w)\right)\right)\right]\}\}.
\end{aligned}
\end{equation}
\end{lemma}
Taking the affine transformation as an example, $T_{\alpha,t}=tg+\alpha$ with $ \alpha,t \in \mathbb{R}$, leads to the variational formula:
\begin{equation}
    \begin{aligned}
D_\phi(\mu \| \nu) & = \sup _{g \in \mathcal{G}, t\in \mathbb{R}} \mathbb{E}_{\w \sim \mu}[tg(\w)] \\ & -\inf _{\alpha \in \mathbb{R}}\left\{\mathbb{E}_{\w \sim \nu}\left[\phi^*(tg(\w)+\alpha)\right]-\alpha\right\}.
\end{aligned}
\end{equation}
The variational representation in Lemma \ref{lemma_f_var} yields a lower bound of $\phi$-divergence when $\mathcal{G}$ and $\mathcal{T}$ contain only a subset of all possible functions.

\section{A Theoretical Guarantee on Adversarial Transferability}\label{main_section}
In this section, we warm up by formalizing the transfer-based attack as a target adversarial risk minimization problem (Section \ref{sec:Formalizing Transfer-Based Attacks}). Through decomposing the target risk into the surrogate adversarial risk and the transferability gap, and deriving the bounds for each part (\bluetwo{Sections} \ref{sec: Model-Discrepancy-Based Bound} and \ref{sec: PAC-Bayesian Bound on Surrogate Risk}), we derive a transferability bound that provides a theoretical guarantee on the adversarial transferability (Section \ref{sec:Transferability Guarantees for Transfer-Based Attacks}). Finally, we establish an attack framework from our bound that generalizes previous works as special cases (Section \ref{sec:Comparison with Others}). 
In the following, we mainly focus on discussing the interpretations and implications of the theorems, and we refer readers to \textit{Appendix} \ref{app: proof} for proof details.

\subsection{Formalizing Transfer-Based Attacks}\label{sec:Formalizing Transfer-Based Attacks}

We start by defining the model distributions and the notion of risks we are concerned with.

\begin{definition}[\textbf{Model distribution}]\label{def_model_dist}
Let $\mathcal{F}$ be the set of possible model architectures for a given task, each function $\hat{f}\in\mathcal{F}$ is a parametric family of models, where $\hat{f}(\cdot,\hat{\w}): \mathcal{X}\rightarrow \mathcal{Y}$ is an example with parameter $\hat{\w} \in\mathbb{R}^{|\hat{\w}|}$. The parameter space induced by $\hat{f}$ is $\hat{\mathcal{W}}=\{\hat{\w}: \hat{\w}\in \mathbb{R}^{|\hat{\w}|}\}$.
We define a model distribution by a distribution over function $P({\hat{f}}(\cdot,\hat{\w}))$, induced by a generic distribution ${P}(\hat{\w})$ over parameters $\hat{\w}$ combined with a model architecture $\hat{f}(\cdot,\hat{\w})$. Typically, different model distributions may have different architectures. We assume that there exists a function $f$ with parameter space $\mathcal{W}\subseteq \mathbb{R}^{|\w|}$ so that arbitrary $P({\hat{f}}(\cdot,\hat{\w}))$ can fit into its architecture with a converted parameter distribution $P(\w)$, \ie, $P(f(\cdot,\w))=P(\hat{f}(\cdot,\hat{\w}))$. By doing so, we remark that $P(f(\cdot,\w))$ is sufficiently general so as to define any model distribution on a common space $\mathcal{W}$, such that any model distribution is absolutely continuous \wrt~measure $d\w$, with density function $p(\w)$. For the sake of clarity, we hereinafter omit the function form $f$ from $P(f(\cdot,\w))$ and instead use the distribution on the underlying parametrization $P(\w)$ to describe a model distribution.
\end{definition}

In the context of transfer-based attacks, once the attacker builds the surrogate model set, the {\textit{surrogate model distribution}} $P_\mathcal{S}$ is observed, with density $p_\mathcal{S}$. Since attackers could customize surrogate models, $P_\mathcal{S}$ could be defined as a set of distributional components, \ie, $P_\mathcal{S}=\left\{P_{\mathcal{S}_i}\right\}_{i=1}^I$, \bluetwo{where} $I$ is the total number of surrogate components owned by the attacker. In Section \ref{sec: Narrow the surrogate-target discrepancy}, we will show that multiple surrogate components help to produce better attack performance. For clarity, and without loss of generality, in this section we consider $P_\mathcal{S}$ integrally.
At test time, the attacker is facing any possible target \bluetwo{model} from the unobserved {\textit{target model distribution}} $P_\mathcal{T}$, with density $p_\mathcal{T}$.

\begin{definition}[\textbf{Adversarial risk and empirical adversarial risk}]\label{def_risk}
Consider a loss $\ell: \mathcal{Y} \times \mathcal{Y} \rightarrow \mathbb{R}_0^{+}$.  Let $P_\mathcal{D}$ be a model distribution. Assuming the AE $\hat{\bm{x}}$ as defined in Section \ref{sec:Notation and Attack Setup}, we can define its adversarial risk on $P_\mathcal{D}$ by:
\begin{equation}
    R_{\mathcal{D}}(\hat{\bm{x}})=\mathbb{E}_{\w \sim P_\mathcal{D}}[\ell(f(\hat{\bm{x}}, \w), y)].
\end{equation}
$R_{\mathcal{D}}(\hat{\bm{x}})$ characterizes the attack failure of an AE on $P_\mathcal{D}$. 
We sample $K$ i.i.d. models $\left\{{\w_i}\right\}_{i=1}^K\sim P_\mathcal{D}$, forming a set $\mathcal{M}_\mathcal{D}$ of size $K$. Given $\mathcal{M}_\mathcal{D}$, we can define an empirical adversarial risk for $\hat{\bm{x}}$ by:
\begin{equation}
    R_{\hat{\mathcal{D}}}(\hat{\bm{x}})=\frac{1}{K} \sum_{\w_i \in \mathcal{M}_\mathcal{D}} \ell\left(f\left(\hat{\bm{x}}, \w_i\right), y\right).
\end{equation}
\end{definition}
The adversarial risk measures the expected attack error that an AE \bluetwo{makes} according to the model distribution. For both adversarial risk and empirical adversarial risk, higher values indicate worse attack performance. 

The task of transfer-based attacks is to find an AE $\hat{\bm{x}}$ that successfully attacks target models drawn from $P_\mathcal{T}$, \ie, to minimize its attack failures.
We formalize untargeted transfer-based attacks as a risk minimization problem under $P_\mathcal{T}$, \ie, seeking \bluetwo{an} $\hat{\bm{x}} \in \hat{\mathcal{X}}$ that minimizes the \textit{target adversarial risk} defined as follows:
\begin{align}
\min_{\hat{\bm{x}} \in \hat{\mathcal{X}}}&R_\mathcal{T}(\hat{\bm{x}}),\\
\text{{where }} R_\mathcal{T}(\hat{\bm{x}})=\mathbb{E}_{\w \sim P_\mathcal{T}}&[\ell(f(\hat{\bm{x}},\w),y)].& \rule{5.1em}{0em}
\end{align}
\noindent The risk definition for targeted attacks is analogously obtained by substituting $y$ with the target label $y_t$. By unifying targeted and untargeted attacks within a single risk minimization framework, we restrict our following analysis to the untargeted case without loss of generality, and the analysis for targeted attacks can be trivially recovered by adopting the targeted risk. 

Under the black-box setting, no information about $P_\mathcal{T}$ is available during the attack, making it impossible to optimize $R_\mathcal{T}(\hat{\bm{x}})$. In practice, attackers commonly resort to an alternative risk minimization, \ie, minimizing the \textit{surrogate adversarial risk} $R_{\mathcal{S}}(\hat{\bm{x}})$, which is measured over the self-chosen surrogate distribution $P_\mathcal{S}$, with the expectation of achieving good transferability. $R_{\mathcal{S}}(\hat{\bm{x}})$ is defined as follows:
\begin{equation}
    R_{\mathcal{S}}(\hat{\bm{x}})= \mathbb{E}_{\w \sim P_{\mathcal{S}}}[\ell(f(\hat{\bm{x}}, \w), y)].
\end{equation}
However, the surrogate model distribution and the inaccessible target model distribution may differ significantly. As a result, the transfer gap between the target adversarial risk and the empirical surrogate adversarial risk becomes even worse due to this distribution shift, making the attack performance unsatisfactory.

\textbf{Risk Decomposition} To derive a bound on the transferability to the target model distribution of an AE optimized under the surrogate model distribution, we first decompose the target adversarial risk as follows:
\begin{equation}\label{eq_decompose}
    R_{\mathcal{T}}(\hat{\bm{x}})=\underbrace{R_{\mathcal{T}}(\hat{\bm{x}})-R_{{\mathcal{S}}}(\hat{\bm{x}})}_{\mathcal{E}_{\text {trans}}\left(\hat{\bm{x}}\right)}+R_{\mathcal{S}}(\hat{\bm{x}}).
\end{equation}
According to the above decomposition, it is clear that solely minimizing the surrogate adversarial risk using some attack strategies cannot guarantee \bluetwo{a decrease in} the target adversarial risk. The transferability gap $\mathcal{E}_{\text {trans}}$, which captures the dissimilarity between the surrogate and target model distributions relevant to the context of adversarial transferability, should also be taken into account. 
Thus, in the following, we will derive an upper bound of the target adversarial risk through \bluetwo{separately} deriving the upper bounds of the transferability gap and the surrogate adversarial risk.

\subsection{Model-Discrepancy-Based Bound on Transferability Gap}\label{sec: Model-Discrepancy-Based Bound}

Equation \ref{eq_decompose} tells that the transferability gap $\mathcal{E}_{\text {trans}}$ depends on the discrepancy between $P_{\mathcal{S}}$ and $P_{\mathcal{T}}$. Thus, we first define a model discrepancy tailored to comparing model distributions in the context of transfer-based adversarial attacks, which is crucial for deriving the subsequent bound on the transferability gap and consequently designing our attack strategy. Specifically, according to the variational formula of $\phi$-divergences (cf. Lemma \ref{lemma_f_var}), we introduce the adversarial model discrepancy $\mathrm{D}_\phi^{\hat{\mathcal{X}_{r}}}$, as follows.

\begin{definition}[\textbf{Adversarial model discrepancy}] \label{def_our_f}
For any surrogate model distribution $P_\mathcal{S}$, target model distribution $P_\mathcal{T}$, and any ${r}\geq0$, the \textit{localized adversarial space} $\hat{\mathcal{X}_{r}}$ is defined as:
\begin{equation}
    \hat{\mathcal{X}_{r}}=\left\{\hat{\bm{x}} \in \hat{\mathcal{X}} \mid R_{\mathcal{S}}\left(\hat{\bm{x}}\right) \leq r\right\}.
\end{equation}
Based on $\hat{\mathcal{X}_{r}}$, let $\hat{\mathcal{G}_r}$ be a set of measurable functions, \ie, $\hat{\mathcal{G}_r}=\{ \ell\left(f\left(\hat{\bm{x}}^{\prime}, \w\right), y\right): \hat{\bm{x}}^{\prime} \in \hat{\mathcal{X}}_r\}$. We define the \textit{adversarial model discrepancy} $\mathrm{D}_\phi^{\hat{\mathcal{X}_{r}}}$ between $P_\mathcal{S}$ and $P_\mathcal{T}$ as:
\begin{align}
    \mathrm{D}_\phi^{\hat{\mathcal{X}}_r}(P_\mathcal{T} \| P_\mathcal{S}) =  &\sup_{\hat{\bm{x}}^{\prime} \in \hat{\mathcal{X}}_r, t \in \mathbb{R}}  \mathbb{E}_{\w \sim P_\mathcal{T}}[t\ell(f({\hat{\bm{x}}^{\prime},\w}),y)]- \label{equ_our_f}\\
    &\inf_{\alpha \in  \mathbb{R}}\{\mathbb{E}_{\w \sim P_\mathcal{S}}[\phi^*(t\ell(f({\hat{\bm{x}}^{\prime},\w}),y)+\alpha)]-\alpha\}.\nonumber
\end{align}
\end{definition}
$\mathrm{D}_\phi^{\hat{\mathcal{X}}_r}(P_\mathcal{T} \| P_\mathcal{S})$ has some properties. 
(1) Restricting $\mathcal{G}$ to the subset $\hat{\mathcal{G}}_r$, $\mathrm{D}_\phi^{\hat{\mathcal{X}}_r}(P_\mathcal{T} \| P_\mathcal{S})$ discrepancy can be interpreted as a lower bound of a general class of $\phi$-divergences $D_\phi(P_\mathcal{T} \| P_\mathcal{S})$, this property is crucial for deriving a general attack framework in Section \ref{sec:Comparison with Others}. 
(2) It is also easy to see that $\mathrm{D}_\phi^{\hat{\mathcal{X}}_r}(P_\mathcal{T} \| P_\mathcal{S})$ is a monotonically increasing function \wrt~\bluetwo{$r$ over} $0\leq r\leq1$. 
(3) $\mathrm{D}_\phi^{\hat{\mathcal{X}}_r}(P_\mathcal{T} \| P_\mathcal{S}) \geq 0$. To explicitly see this, we consider $t=0$. By Lemma \ref{lemma_phi*}, $\inf \phi^*(\alpha)-\alpha=0$, \bluetwo{which} leads to $\mathbb{E}_{\w \sim P_\mathcal{T}}[t\ell(f({\hat{\bm{x}}^{\prime},\w}),y)]- \inf_{\alpha \in  \mathbb{R}}\{\mathbb{E}_{\w \sim P_\mathcal{S}}[\phi^*(t\ell(f({\hat{\bm{x}}^{\prime},\w}),y)+\alpha)]-\alpha\}=0$ when $t=0$, then we prove the non-negativity of $\mathrm{D}_\phi^{\hat{\mathcal{X}}_r}(P_\mathcal{T} \| P_\mathcal{S})$. 
Moreover, the equality holds when $\ell(f(\hat{\boldsymbol{x}}^{\prime}, \boldsymbol{w}), y)|_{\boldsymbol{w} \sim P_{\mathcal{S}}} \stackrel{d}{=} \ell(f\left(\hat{\boldsymbol{x}}^{\prime}, \boldsymbol{w}), y\right)|_{\boldsymbol{w} \sim P_{\mathcal{T}}}$, a weaker condition than $P_\mathcal{S}=P_\mathcal{T}$. By $\phi^*(\bm{x})\geq \bm{x}$, we have $\mathbb{E}_{\w \sim P_\mathcal{S}}[t\ell(f({\hat{\bm{x}}^{\prime},\w}),y)+\alpha]- \mathbb{E}_{\w \sim P_\mathcal{S}}[\phi^*(t\ell(f({\hat{\bm{x}}^{\prime},\w}),y)+\alpha)]\leq 0$, \bluetwo{which} leads to $\mathrm{D}_\phi^{\hat{\mathcal{X}}_r}(P_\mathcal{T}\| P_\mathcal{S}) = 0$.


\blue{Our adversarial model discrepancy is conceptually related to the localized discrepancies in domain adaptation, where one measures the shift between source and target \emph{data} distributions via an IPM or \bluetwo{a} $\phi$-divergence localized to a good hypothesis subset \cite{localized,f2}. In contrast, our $\mathrm{D}_\phi^{\hat{\mathcal{X}}_r}(P_\mathcal{T} \| P_\mathcal{S})$ compares surrogate vs. target \emph{model} distributions. Moreover, taking the supremum over $\hat{\bm{x}}^{\prime} \in \hat{\mathcal{X}}_r$, our localization restricts to \emph{adversarially significant} inputs (committing low surrogate attack error), \ie, evaluating models only on candidate AEs, thus capturing only discrepancy that matters for adversarial transferability.}

We are now ready to provide a bound on the transferability gap $\mathcal{E}_{\text {trans}}\left(\hat{\bm{x}}\right)$ in terms of the proposed $\mathrm{D}_\phi^{\hat{\mathcal{X}}_r}$ discrepancy.
\begin{theorem}[\textbf{Transferability gap bound}] \label{theorem1}
Define $K_{\mathcal{S}}^{\hat{\bm{x}}}\left(t\right) = \inf _\alpha \left\{\mathbb{E}_{\w \sim P_\mathcal{S}}\left[\phi^*\left(t \ell\left(f\left(\hat{\bm{x}},\w\right),y\right)+\alpha\right)\right]-\alpha\right\} -\mathbb{E}_{\w \sim P_\mathcal{S}}\left[t \ell\left(f\left(\hat{\bm{x}},\w\right),y\right)\right]$. Given the surrogate model distribution $P_\mathcal{S}$ and target model distribution $P_\mathcal{T}$, for any $\hat{\bm{x}} \in \hat{\mathcal{X}_r}$ and \bluetwo{constants} $c_1, c_2 \in [0, +\infty)$ \bluetwo{subject} to the constraint $K_{\mathcal{S}}^{\hat{\bm{x}}}(c_1)\leq c_1c_2\mathbb{E}_{\w \sim P_\mathcal{S}}\left[\ell\left(f(\hat{\bm{x}},\w),y\right)\right]$, we have
\begin{equation}\label{eq_gap_final}
    \mathcal{E}_{\text {trans}}\left(\hat{\bm{x}}\right) \leq 
    \frac{1}{c_1}\mathrm{D}_\phi^{\hat{\mathcal{X}}_r}\left({P_\mathcal{T}} \| {P_\mathcal{S}}\right)+c_2r.
\end{equation}
\end{theorem}

Theorem \ref{theorem1} bounds the transferability gap in terms of the adversarial model discrepancy between the surrogate and target model distributions, as well as a constant term related to localized adversarial space parameter $r$. 
As an instantiation of Theorem \ref{theorem1}, we \blue{first consider a simple yet intuitive} case of TV, namely $\mathrm{D}_{\text{TV}}^{\hat{\mathcal{X}}_r}\left({P_\mathcal{T}} \| {P_\mathcal{S}}\right)$. We have the following result:
\begin{corollary}[\blue{\textbf{TV Instantiation of Theorem \ref{theorem1}}}] \label{lemma_tv}
Suppose $\ell: \mathcal{Y} \times \mathcal{Y} \rightarrow [0,1]$. Given the surrogate model distribution $P_\mathcal{S}$ and target model distribution $P_\mathcal{T}$, for any $\hat{\bm{x}} \in \hat{\mathcal{X}_r}$ and constant $c_1$ satisfying $0\leq c_1\leq 1$, we have
\begin{equation}
     \mathcal{E}_{\text {trans}}\left(\hat{\bm{x}}\right) \leq \frac{1}{c_1}\mathrm{D}_\text{TV}^{\hat{\mathcal{X}}_r}(P_\mathcal{T} \| P_\mathcal{S}),
\end{equation} 
where $\mathrm{D}_\text{TV}^{\hat{\mathcal{X}}_r}(P_\mathcal{T} \| P_\mathcal{S})=\sup_{\hat{\bm{x}} \in \hat{\mathcal{X}}_r} \mid \mathbb{E}_{\w \sim P_\mathcal{T}}\left[\ell\left(f\left({\hat{\bm{x}},\w}\right),y\right)\right]-\mathbb{E}_{\w \sim P_\mathcal{S}}\left[\ell\left(f\left({\hat{\bm{x}},\w}\right),y\right)\right]\mid$.
\end{corollary}
{
\bluetable
The TV instantiation yields a clean bound that only depends on the worst-case mean shift of adversarial loss. However, it cannot downweight the discrepancy term via $\frac{1}{c_1}$ and admits no $c_1, c_2$ trade-off. One may expect more instantiations of Theorem \ref{theorem1} to $\phi$-divergence family, \eg, KL and $\chi^2$ divergences, and the corresponding explicit clues for the constraint $K_{\mathcal S}^{\hat{\bm x}}(t)\ \leq c_1 c_2\,\mathbb E_{P_{\mathcal S}}[\ell(\hat{\bm x},\bm w)]$. We defer them to \textit{Appendix} \ref{app:inst}. In short, the KL instantiation yields a similar bound (cf. Corollary \ref{corollary:KL}). When satisfying $c_1>0$,  $c_2 \geq \frac{e^{c_1}-1-c_1}{c_1}$,
\begin{equation}
\mathcal{E}_{\text{trans}}(\hat {\bm{x}})
\le \frac{1}{c_1}D_{\text{KL}}^{\hat{\mathcal X}_r}(P_{\mathcal T} \| P_{\mathcal S})
+c_2\,r,
\end{equation}
where $D_{\text{KL}}^{\hat{\mathcal X}_r}(P_{\mathcal T}\| P_{\mathcal S})=\sup_{\hat{\bm{x}} \in\hat{\mathcal X}_r,\,t\in\mathbb R} t\,\mathbb E_{\bm{w}\sim P_{\mathcal T}}[\ell(f(\hat {\bm{x}},\bm{w}),y)] -\log \mathbb E_{\bm{w}\sim P_{\mathcal S}}[e^{t\ell(f(\hat {\bm{x}},\bm{w}),y)}]$. One may choose $c_1>1$, therefore controlling the discrepancy term with $\frac{1}{c_1}<1$ at the cost of a larger $c_2$. This trade-off yields a potentially tighter bound when $r$ is small (cf. Remark \ref{remark:KL1}). 
Similarly, the $\chi^2$ instantiation yields a bound holds when $c_1>0, c_2 \geq \frac{{c_1}}{4} \cdot \frac{\operatorname{Var}_{P_{\mathcal S}}(\ell (f(\hat{\bm x}, \bm w), y))}{\mathbb{E}_{\bm w \sim  P_{\mathcal S}}[\ell (f(\hat{\bm x}, \bm w), y)]}$ (cf. Corollary \ref{corollary_chi2}):
\begin{align}
\mathcal{E}_{\text{trans}}(\hat {\bm{x}})
\le \frac{1}{c_1}D_{{{\chi}^2}}^{\hat{\mathcal X}_r}(P_{\mathcal T} \| P_{\mathcal S})+c_2\,r,
\end{align}
where $D^{\hat{\mathcal{X}}_r}_{\chi^2}(P_{\mathcal T} \| P_{\mathcal S})=\sup _{\hat{\bm{x}}^{\prime} \in \hat{\mathcal{X}}_r, t \in \mathbb{R}}\{t(\mathbb{E}_{\bm w \sim P_{\mathcal T}}[\ell(f(\hat{\bm x}^{\prime}, \bm w ), y)] \nonumber-\mathbb{E}_{\bm w \sim P_{\mathcal S}}[\ell(f(\hat{\bm x}^{\prime}, \bm w), y)])-\frac{t^2}{4} \operatorname{Var}_{\bm w \sim P_{\mathcal S}}(\ell(f(\hat{\bm x}^{\prime}, \bm w), y))\}$. The feasibility condition for $\chi^2$ yields a more favorable trade-off than KL (cf. Remark \ref{remark:chi2_1})\bluetwo{.}
}

According to risk decomposition, we have $R_{{\mathcal{T}}}(\hat{\bm{x}})=\mathcal{E}_{\text {trans}}+R_{{\mathcal{S}}}(\hat{\bm{x}})$, we need to minimize $\mathcal{E}_{\text {trans}}$ and $R_{{\mathcal{S}}}(\hat{\bm{x}})$ simultaneously to \bluetwo{ensure} the target attack performance. Combining $R_{{\mathcal{S}}}(\hat{\bm{x}})$ with the bound for $\mathcal{E}_{\text {trans}}$ in Theorem \ref{theorem1}, we have the following bound on the target adversarial risk:
\begin{equation}\label{eq_risk_bound}
    \begin{aligned}
        R_{\mathcal{T}}(\hat{\bm{x}}) \leq R_{\mathcal{S}}(\hat{\bm{x}}) + \frac{1}{c_1}\mathrm{D}_\phi^{\hat{\mathcal{X}}_r}\left({P_\mathcal{T}} \| {P_\mathcal{S}}\right)+c_2r.
    \end{aligned}
\end{equation}
The above bound yields the following result: Let $\hat{\bm{x}} \in \hat{\mathcal{X}}_r$ be an AE optimized by minimizing risk on the surrogate mixture $P_\mathcal{S}$. If $\hat{\bm{x}}$ can successfully attack over $P_{\mathcal{S}}$ seen during optimization, then $\hat{\bm{x}}$ has bounded risk over future target model distribution $P_\mathcal{T}$, if $P_\mathcal{T}$ has low adversarial model discrepancy with $P_{\mathcal{S}}$.

\subsection{PAC-Bayesian Bound on Surrogate Adversarial Risk}\label{sec: PAC-Bayesian Bound on Surrogate Risk}

Taking advantage of the above result, we are now able to conduct a transferability-guaranteed attack. In practice, the attacker typically only owns a finite surrogate set $\mathcal{M}_\mathcal{S}$ of size $K$, \ie, $\left\{{\w_k}\right\}_{k=1}^K\sim P_\mathcal{S}$, and the bound in Equation \ref{eq_risk_bound} needs to be estimated empirically. Hence, the next step in this section involves obtaining an empirical version of the bound in Equation \ref{eq_risk_bound}. Note that we present a concentration result solely for $R_{{\mathcal{S}}}(\hat{\bm{x}})$, as the target models are unavailable and the $\mathcal{E}_{\text {trans}}$-related terms are thus intractable. Introducing \bluetwo{their} computation offers no clear benefit. Nevertheless, these terms will serve to offer intuition for efficiently choosing surrogate models to control the transferability gap.

With the surrogate set $\mathcal{M}_\mathcal{S}$, \bluetwo{the objective of} minimizing \textit{empirical surrogate adversarial risk} $R_{\hat{\mathcal{S}}}(\hat{\bm{x}})$,
\begin{equation}
    R_{\hat{\mathcal{S}}}(\hat{\bm{x}})=\frac{1}{K} \sum_{\w_k \in \mathcal{M}_\mathcal{S}} \ell\left(f\left(\hat{\bm{x}}, \w_k\right), y\right),
\end{equation}
can have multiple local minima provide similar white-box attack loss but significantly different generalization on $R_{\mathcal{S}}(\hat{\bm{x}})$, and consequently black-box performance $R_{\mathcal{T}}(\hat{\bm{x}})$. Unfortunately, the typical optimization methods, such as PGD \cite{pgd} and I-FGSM \cite{ifgsm}, often lead to suboptimal transferability \cite{rap}. 

To bound $R_{{\mathcal{S}}}(\hat{\bm{x}})$, our goal suggests that \bluetwo{applying} PAC-Bayes theorem in Theorem \ref{theorem_ori_pac} may be fruitful. In our task, we have the concept is the adversarial example. The instance refers to the model. The risk refers to the adversarial risk. The generalization error measures how well the generated AE transfers from the employed samples to the surrogate model distribution (see these in Lemma \ref{lemma_pac_trans}). Under the PAC-Bayesian framework, we derive a bound for surrogate adversarial risk such that it could be estimated from finite models sampled from $P_\mathcal{S}$:

\begin{theorem}[\textbf{Surrogate risk bound}] \label{theorem_pac}
For any $\rho>0$, $0<\delta<1$, model distribution $P_\mathcal{S}$, and $\hat{\bm{x}} \in \hat{\mathcal{X}_r}$, with probability $1-\delta$ over the choice of surrogate model set $\mathcal{M}_\mathcal{S}\sim P_\mathcal{S}$ with size $K \in \mathbb{N}$, we have
\begin{equation}\label{equ_pac}
\begin{aligned}
    &R_{\mathcal{S}}(\hat{\bm{x}})
     \leq \max_{\|\boldsymbol{\epsilon}\|_2 \leq \rho}R_{\hat{\mathcal{S}}}(\hat{\bm{x}}+\bm{\epsilon})+\\
     &\sqrt{\frac{\frac{d}{2}\log (1+\frac{\gamma^2}{\rho^2}(1+\sqrt{\frac{\log K}{d}})^2)+\log\frac{K}{\delta}+\tilde{\mathcal{O}}(1)}{2(K-1)}}.
\end{aligned}
\end{equation}
where $\tilde{\mathcal{O}}(1)$ term \blue{corresponds} to $\varepsilon=\frac{1}{2}+2\log (2+3d+6r^2K+4d\log(\sqrt{d}+\sqrt{logK}))$.
\end{theorem}
As we can see, this PAC-Bayes bound depends on two terms. The first one is the supremum of empirical surrogate risk over perturbed AE $\hat{\bm{x}}+\bm{\epsilon}$, which denotes the worst-case attack error \bluetwo{within} neighborhood \bluetwo{around} $\hat{\bm{x}}$. The second one is a confidence bound which tells the effect of number of surrogate samples $K$ on transferability bound. If $\hat{\bm{x}}$ is optimized over enough samples, this term can be reduced, and one can use the first term as an upper bound estimator of surrogate risk.

\subsection{Transferability Guarantees for Transfer-Based Attacks}\label{sec:Transferability Guarantees for Transfer-Based Attacks}

Plugging the bounds in Theorem \ref{theorem1} and Theorem \ref{theorem_pac} into Equation \ref{eq_decompose} yields our final transferability PAC bound on target adversarial risk:

\begin{theorem}[\textbf{Transferability PAC bound}] \label{theorem_main}
Given the surrogate model distribution $P_\mathcal{S}$ and target model distribution $P_\mathcal{T}$. For any $\hat{\bm{x}} \in \hat{\mathcal{X}_r}$ and \bluetwo{constants} $c_1, c_2 \in [0, +\infty)$ \bluetwo{subject} to the constraint $K_{\mathcal{S}}^{\hat{\bm{x}}}(c_1)\leq c_1 c_2\mathbb{E}_{\w \sim P_\mathcal{S}}\left[\ell\left(f(\hat{\bm{x}},\w),y\right)\right]$, with probability $1-\delta$ over surrogate model set $\mathcal{M}_{\mathcal{S}}=\left\{\w_j\right\}_{j=1}^{K}$ generated from distribution $P_\mathcal{S}$, we have
    \begin{equation}\label{equ_main}
    \begin{aligned}
    R_{\mathcal{T}}(\hat{\bm{x}})
    \leq \max _{\|\bm{\epsilon}\|_2 \leq \rho} R_{\hat{\mathcal{S}}}(\hat{\bm{x}}+\bm{\epsilon})
    + \frac{1}{c_1}\mathrm{D}_\phi^{\hat{\mathcal{X}}_r}({P_\mathcal{T}} \| {P_\mathcal{S}})+c_2r
    +\varepsilon _{\text{PAC}},
    \end{aligned}
    \end{equation}
    where $\varepsilon _{\text{PAC}}=\sqrt{\frac{\frac{d}{2}\log (1+\frac{\gamma^2}{\rho^2}(1+\sqrt{\frac{\log K}{d}})^2)+\log\frac{K}{\delta}+\tilde{\mathcal{O}}(1)}{2(K-1)}}$.
\end{theorem}
Rewriting the above bound, we have:
\begin{equation}\label{equ_main_sharp}
    \begin{aligned}
    R_{\mathcal{T}}(\hat{\bm{x}}) & \leq R_{\hat{\mathcal{S}}}(\hat{\bm{x}})+\underbrace{\max _{\|\bm{\epsilon}\|_2 \leq \rho} R_{\hat{\mathcal{S}}}(\hat{\bm{x}}+\bm{\epsilon})-R_{\hat{\mathcal{S}}}(\hat{\bm{x}})}_{\text {sharpness}} \\ 
    & +\frac{1}{c_1}\mathrm{D}_\phi^{\hat{\mathcal{X}}_r}\left({P_\mathcal{T}} \| {P_\mathcal{S}}\right)+c_2r
    +\varepsilon _{\text{PAC}},
\end{aligned}
\end{equation}
where the terms in the curly bracket depict the sharpness of $R_{\hat{\mathcal{S}}}$ at $\hat{\bm{x}}$ as it measures the difference of risk between $\hat{\bm{x}}$ and the worst-case point in the neighborhood of $\hat{\bm{x}}$ \cite{keskar2017large}. A low value of the sharpness term indicates that $\hat{\bm{x}}$ is located \bluetwo{in} the flat region of the loss landscape. 
\blue{For $\mathrm{D}_\phi^{\hat{\mathcal{X}}_r}$ term, one can also turn to instantiations of Theorem \ref{theorem1} (cf. Corollary \ref{lemma_tv},\ref{corollary:KL},\ref{corollary_chi2}) to obtain concrete results specialized for popular choices \bluetwo{of} $\phi$-divergences (such as $\mathrm{D}_{\text{TV}}^{\hat{\mathcal{X}}_r}$) and the corresponding explicit conditions for the constraint.}
As a result, the target adversarial risk of an AE $\hat{\bm{x}}$ can be bounded in terms of (1) the white-box attack performance of $\hat{\bm{x}}$ against $\mathcal{M}_\mathcal{S}$, (2) the sharpness of $R_{\hat{\mathcal{S}}}$ at $\hat{\bm{x}}$, (3) the adversarial model discrepancy, (4) a constant term related to localized adversarial space parameter $r$ and  (5) a confidence bound. 
Finally, we provide a guarantee that $\hat{\bm{x}}$ will ``transfer well'' on target model distribution $P_\mathcal{T}$, even when solely minimizing the empirical risk over the surrogate model set $\mathcal{M}_\mathcal{S}$.
This bound inspires our basic plan of attack: controlling the adversarial model discrepancy, the attacker can expect that finding a flat minimum on empirical surrogate adversarial risk will lead to better AE transferability.


\subsection{A General Transfer-Based Attack Framework}\label{sec:Comparison with Others}

\begin{table*}[]
\centering
\caption{Comparison of bounds for target adversarial risk $R_{\mathcal{T}}(\hat{\bm{x}})$. The ``bound'' means the respective attack tries to find an AE $\hat{\bm{x}}$ which minimizes $R_{\mathcal{T}}(\hat{\bm{x}})$ by alternatively minimizing this objective.}
\label{tab:bounds}
\resizebox{\textwidth}{!}{
\bluetable%
\begin{tabular}{@{}cccc@{}}
\toprule
 & \multicolumn{3}{c}{Target risk $R_{\mathcal{T}}(\hat{\bm{x}}) \leq \underbrace{R_{\hat{\mathcal{S}}}(\hat{\bm{x}})+\mathbf{r}(\mathcal{\mathcal{S}})}_{\text {for bounding } R_{{\mathcal{S}}}(\hat{\bm{x}})}  +\underbrace{\eta \mathrm{D}_\phi\left({P_\mathcal{T}} \| {P_\mathcal{S}}\right)}_{\text{for bounding }\mathcal{E}_{\text {trans}}(\hat{\bm{x}})}+\text{constant}$} \\ \cmidrule(l){2-4}
\multirow{-3}{*}{Method} & \multicolumn{1}{c}{Empirical surrogate risk $R_{\hat{\mathcal{S}}}(\hat{\bm{x}})$} & \multicolumn{1}{c}{Surrogate model regularization $\mathbf{r}(\mathcal{\mathcal{S}})$} & \multicolumn{1}{c}{{$\phi$-divergence $\mathrm{D}_\phi\left({P_\mathcal{T}} \| {P_\mathcal{S}}\right)$}} \\ \midrule
MI\cite{mifgsm},NI\cite{nisi},PI\cite{pifgsm},VT\cite{vt}& \begin{tabular}[c]{@{}c@{}}Design elaborate optimizers to minimize \\ $ \ell\left(f\left(\hat{\bm{x}}, \w\right), y\right)$ \end{tabular} & $\backslash$ & $\backslash$ \\
\midrule
SVRE\cite{svre},AdaEA\cite{adaea},SMER\cite{smer} & \begin{tabular}[c]{@{}c@{}}Design fusion strategies to minimize \\ $\ell\left(\sum_{\w_k \in \mathcal{M}_\mathcal{S}}\omega_k f\left(\hat{\x}, \w_k\right), y\right)$ \end{tabular} & $\backslash$ & $\backslash$ \\
\midrule
ILA\cite{ila} &\begin{tabular}[c]{@{}c@{}} Intermediate Level Attack Projection Loss\\ $-\left(f_l\left(\hat{\bm{x}}^\prime\right)-f_l(\x)\right)\left(f_l\left(\hat{\bm{x}}\right)-f_l(\x)\right)$\end{tabular} & $\backslash$ & $\backslash$ \\
\midrule
FIA\cite{fia} & \begin{tabular}[c]{@{}c@{}} Feature Importance-\bluetwo{A}ware Loss\\ $\sum\left(\bar{\Delta}_l^{\boldsymbol{x}} \odot f_l\left(\hat{\bm{x}}\right)\right)$ \end{tabular} & $\backslash$ & $\backslash$ \\
\midrule
NAA\cite{naa} & \begin{tabular}[c]{@{}c@{}} Neuron Attribution-\bluetwo{B}ased Loss \\ $\sum_{\substack{A_{l_j} \geq 0 \\ f_{l_j} \in f_l}} f_p\left(A_{l_j}\right)-\gamma \cdot \sum_{\substack{A_{l_j}<0 \\ f_{l_j} \in f_l}} f_n\left(-A_{l_j}\right)$ \end{tabular} & $\backslash$ & $\backslash$ \\
\midrule
\multicolumn{1}{c}{\begin{tabular}[c]{@{}c@{}}RAP\cite{rap},\\ MEF\cite{mef}\end{tabular}} & \multicolumn{1}{c}{$\ell\left(f\left(\hat{\bm{x}}, \w\right), y\right)$} & \multicolumn{1}{c}{\begin{tabular}[c]{@{}c@{}} Sharpness (worst-case/smoothed)\\ $\max _{\|\bm{\epsilon}\|_p \leq \rho} \left[\ell\left(f\left(\hat{\bm{x}}+\bm{\epsilon}, \w\right), y\right)-\ell\left(f\left(\hat{\bm{x}}, \w\right), y\right)\right]$,\\ $\max _{\|\epsilon\|_p<\rho} \mathbb{E}_{\bm{\delta} \sim \operatorname{Unif}\left(\mathcal{B}_{\xi}(0)\right)} \left[\ell(f(\hat{\bm{x}}+\bm{\epsilon}+\bm{\delta}, \bm{w}), y)-\ell\left(f\left(\hat{\bm{x}}, \w\right), y\right)\right]$\end{tabular}} & $\backslash$ \\
\midrule
\begin{tabular}[c]{@{}c@{}}PGN\cite{pgn},\\ TPA\cite{tpa}\end{tabular} & $\ell\left(f\left(\hat{\bm{x}}, \w\right), y\right)$ & \begin{tabular}[c]{@{}c@{}} Gradient norm (maximum/expected)\\ $\max _{\|\bm{\epsilon}\|_p \leq \rho}\left\|\nabla \ell\left(f\left(\hat{\bm{x}}+\bm{\epsilon}, \w\right), y\right)\right\|_2$,\\ $\mathbb{E}_{\bm{\epsilon} \sim \operatorname{Unif}\left(\mathcal{B}_\rho(0)\right)}\left[\|\nabla \ell(f(\hat{\bm{x}}+\bm{\epsilon}, \bm{w}), y)\|_2\right]$\end{tabular}  & $\backslash$ \\
\midrule
CWA\cite{cwa}  & $\frac{1}{|\mathcal{M}_\mathcal{S}|} \sum_{\w_k \in \mathcal{M}_\mathcal{S}} \ell\left(f\left(\hat{\bm{x}}, \w_k\right), y\right)$ & Hessian matrix $\frac{1}{|\mathcal{M}_\mathcal{S}|}\sum_{k=1}^{|\mathcal{M}_\mathcal{S}|}\|{\bm H}_k\|_F$ & $\backslash$ \\
\midrule
DI\cite{difgsm},TI\cite{ti},SI\cite{nisi},Admix\cite{admix},SIA\cite{sia},SSA\cite{ssa} & $ \ell\left(f\left(\mathcal{T}(\hat{\bm{x}}), \w\right), y\right)$ & $\backslash$ & Simulate different models with input transformations $\mathcal{T}$ \\
\midrule
GhostNet\cite{ghostnet},Bayesian\cite{bayesian},LGV\cite{lgv} & $ \frac{1}{|\mathcal{M}_\mathcal{S}|} \sum_{\w_k \in \mathcal{M}_\mathcal{S}} \ell\left(f\left(\hat{\bm{x}}, \w_k\right), y\right)$ & $\backslash$ & Generate diverse variants from a base surrogate model\\
\midrule
\rowcolor[HTML]{D9D9D9} 
Our  & $ \frac{1}{|\mathcal{M}_\mathcal{S}|} \sum_{\w_k \in \mathcal{M}_\mathcal{S}} \ell\left(f\left(\hat{\bm{x}}, \w_k\right), y\right)$  & \begin{tabular}[c]{@{}c@{}}Sharpness over $\mathcal{M}_\mathcal{S}=\{\w_k\}_{k=1}^{|\mathcal{M}_\mathcal{S}|} \sim P_\mathcal{S}$ \\ $\frac{1}{|\mathcal{M}_\mathcal{S}|}\max _{\|\bm{\epsilon}\|_p \leq \rho}\sum_{\w_k \in \mathcal{M}_\mathcal{S}} \left[\ell\left(f\left(\hat{\bm{x}}+\bm{\epsilon}, \w_k\right), y\right)-\ell\left(f\left(\hat{\bm{x}}, \w_k\right), y\right)\right]$ \end{tabular} & \begin{tabular}[c]{@{}c@{}} Narrow adversarial model discrepancy \\ $\mathrm{D}^{\hat{\mathcal{X}}_r}_\phi\left({P_\mathcal{T}} \| {P_\mathcal{S}}\right)$ \end{tabular} \\ \bottomrule
\end{tabular}%
}
\vspace{-10pt}
\end{table*}

In this section, we present an attack framework which generalizes previous transfer-based attacks. Through the lens of our framework, we revisit these attacks, especially RAP, and compare them with our transferability bound in Table \ref{tab:bounds}. The analysis shows that, while these attacks improve transferability through either finding better local minima in the surrogate loss landscape or tackling model shift, they do not consider both optimization signals simultaneously to achieve comprehensive transferability. Moreover, their principles of controlling the transferability gap are less tight than our $\mathrm{D}_\phi^{\hat{\mathcal{X}_r}}$ discrepancy and may result in unnecessary overestimation of the target risk bound. Experimental results confirm that considering surrogate adversarial risk and transferability gap simultaneously and properly leads to significant gains (see Tables \ref{tab: main}, \ref{tab:main_t}).

Abstracted from our main result (cf. Equation \ref{equ_main_sharp}), we give \textit{a general framework} for bounding the target adversarial risk:
\begin{equation}\label{equ_framework}
    R_{\mathcal{T}}(\hat{\bm{x}}) \leq \underbrace{R_{\hat{\mathcal{S}}}(\hat{\bm{x}})+\mathbf{r}(\mathcal{\mathcal{S}})}_{\text {for bounding } R_{{\mathcal{S}}}(\hat{\bm{x}})}  +\underbrace{\eta \mathrm{D}_\phi\left({P_\mathcal{T}} \| {P_\mathcal{S}}\right)}_{\text{for bounding }\mathcal{E}_{\text {trans}}(\hat{\bm{x}})}+\text{ constant}.
\end{equation}
\noindent where $\eta$ is a weight that trades off transferability with attack performance on surrogates. Within the first curly bracket, $\mathbf{r}(\mathcal{\mathcal{S}})$ represents some form of regularization for surrogate models, \eg, sharpness, which interacts with the empirical surrogate risk $R_{\hat{\mathcal{S}}}(\hat{\bm{x}})$ to upper bound the surrogate risk $R_{{\mathcal{S}}}(\hat{\bm{x}})$. Within the second curly bracket for bounding $\mathcal{E}_{\text {trans}}\left(\hat{\bm{x}}\right)$, we replace $\mathrm{D}_\phi^{\hat{\mathcal{X}_r}}\left({P_\mathcal{T}} \| {P_\mathcal{S}}\right)$ in Equation \ref{equ_main_sharp} \bluetwo{with} $\mathrm{D}_\phi\left({P_\mathcal{T}} \| {P_\mathcal{S}}\right)$ without violating the bound, as the variational representation $\mathrm{D}_\phi^{\hat{\mathcal{X}_r}}$ is a lower bound of the $\phi$-divergence $\mathrm{D}_\phi$.

Taking a second look at previous attacks within the above framework, some methods (\eg, MI, NI, PI, VT) bound $R_{\mathcal{T}}(\hat{\bm{x}})$ solely via the empirical surrogate risk $R_{\hat{\mathcal{S}}}(\hat{\bm{x}})$ on one or several arbitrarily selected neural networks. The various gradient-based optimization algorithms they introduce to minimize $R_{\hat{\mathcal{S}}}(\hat{\bm{x}})$ could help escape poor minima, thus improving the optimization for $R_{\mathcal{S}}(\hat{\bm{x}})$. \blue{From the model-fusion perspective, methods such as SVRE, AdaEA, \bluetwo{and} SMER also focus on the empirical risk, but better synchronize the optimization directions across surrogates via ensemble strategies such as reweighting.} Similarly, feature-based methods (\eg, ILA, FIA, NAA) focus on optimizing $R_{\hat{\mathcal{S}}}(\hat{\bm{x}})$ by designing alternative loss functions which distort intermediate layer features rather than the final outputs. Going beyond simply accounting for the empirical risk, RAP, \blue{MEF}, PGN, \blue{TPA and CWA} bound $R_{\mathcal{T}}(\hat{\bm{x}})$ by $R_{\hat{\mathcal{S}}}(\hat{\bm{x}})$ in conjunction with specific surrogate model regularizations. However, the above methods \bluetwo{incorrectly} rely on an invalid i.i.d. assumption and overlook the transferability gap $\mathcal{E}_{\text {trans}}\left(\hat{\bm{x}}\right)$ in the target risk bound, thus achieving suboptimal attack performance.

Rather than relying solely on a few surrogate models, input-transformation-based methods (\eg, DI, TI, SI, Admix, SIA, SSA) and model-tuning-based methods (\eg, GhostNet, Bayesian attack, LGV) augment a base surrogate model into an infinitely large set of models, aiming to align with those seen during inference. The common underlying assumption of these attacks is that the transferability of adversarial examples can be improved by attacking \textit{more} models simultaneously, so efforts focus on obtaining as many different surrogate models as possible at a low computational cost \cite{nisi, sia, bayesian}. 
By trying to place point masses at locations given by samples from the target model distribution, we can view these methods as simulating $P_\mathcal{T}$ with $P_\mathcal{S}$, which aligns with minimizing $\mathrm{D}_\phi\left({P_\mathcal{T}} \| {P_\mathcal{S}}\right)$ in our framework, given that $\phi$-divergences measure the difference between two given probability distributions. However, we demonstrate that $\mathrm{D}_\phi$ may overestimate the transferability gap compared to our adversarial model discrepancy $\mathrm{D}_\phi^{\hat{\mathcal{X}_r}}$, which only captures practically significant difference in adversarial vulnerability patterns. Moreover, despite mitigating the surrogate-target shift, these methods have not explored the AE optimizer in depth and typically default to the empirical risk minimization (ERM) approach \bluetwo{in} I-FGSM.

In summary, existing attacks typically control only one of two relevant terms in Equation \ref{equ_framework}, neither of which is desirable nor sufficient to achieve satisfactory transferability. In contrast, our proposed attack accounts for $R_{\mathcal{S}}(\hat{\bm{x}})$ and $\mathcal{E}_{\text {trans}}$ jointly and properly, providing a more comprehensive guarantee on the transferability of AEs.

\textbf{Detailed comparison with RAP \cite{rap}}
The relationship between flatness and transferability \bluetwo{was} initially explored in our prior work, RAP. However, this relationship is only intuitively assumed through an illustration (see Figure 1(b) in the RAP paper) and empirically validated without a rigorous theoretical grounding. In this work, our newly derived bound in Equation \ref{equ_main_sharp} provides a theoretical interpretation of the relationship between flat regions with low loss in loss landscape and the transferability of the AE, offering solid theoretical \bluetwo{support for} the assumption of RAP.
Moreover, inspired by this theoretical insight, our proposed attack plan in this paper advances RAP from two perspectives. First, we highlight the importance of the marginalization of surrogate models and seek flat minima over enough samples from the model distribution. Second, we explicitly require that surrogate models be carefully selected to narrow the adversarial model discrepancy as \bluetwo{much} as possible, as sharpness by itself is insufficient to guarantee high transferability. From this point of view, RAP could be considered a degraded version of our new attack.

\section{Attack Algorithm}
\subsection{From the Bound to an Attack}
We now exploit the above theoretical results to derive a novel practical attack algorithm. We will show how our theoretical analyses help make algorithmic extensions to \bluetwo{the} original RAP. 

Inspired by Theorem \ref{theorem_main}, we first sample the surrogate model set $\mathcal{M}_\mathcal{S}$ from $P_\mathcal{S}$, where $P_\mathcal{S}$ is designed to control the instantiated adversarial model discrepancy $\mathrm{D}_{\phi}^{\hat{\mathcal{X}}_r}(P_\mathcal{T} \| P_\mathcal{S})$ with $P_\mathcal{T}$. Then we propose to minimize the target adversarial risk $R_\mathcal{T}(\hat{\bm{x}})$ by the following optimization problem:
\begin{equation}\label{eq_obj}
 \begin{aligned}
    \min_{\hat{\bm{x}} \in \hat{\mathcal{X}}}\max _{\|\bm{\epsilon}\|_p \leq \rho}\frac{1}{K} \sum_{\w_k \in \mathcal{M}_\mathcal{S}} {\ell}\left(f\left(\hat{\bm{x}}+\bm{\epsilon}, \w_k\right), y\right),
\end{aligned}    
\end{equation}
where $\ell$ is a surrogate loss to minimize the empirical surrogate adversarial risk. We generalize from $L_2$ norm to $L_p$ norm to make it adaptive to the popular constraints (\ie, $L_2$, $L_\infty$) on examples in adversarial attacks. As a result, Theorem \ref{theorem_main} results in a strategy comprising two components: collecting surrogate models guided by the instantiated adversarial model discrepancy term and seeking a flat minimum $\hat{\bm{x}}$ according to a min-max term. We will discuss them separately in the following sections.

\subsection{Narrow the Surrogate-Target Discrepancy}\label{sec: Narrow the surrogate-target discrepancy}

{
\bluetable
To make the attack power of an AE invariant across models, it is crucial to build a surrogate distribution $P_\mathcal{S}$ that narrows the adversarial model discrepancy $\mathrm{D}_{\phi}^{\hat{\mathcal{X}}_r}(P_\mathcal{T} \| P_\mathcal{S})$ with $P_\mathcal{T}$, thereby controlling the transferability gap. Although $P_{\mathcal T}$ is unknown and thus the discrepancy term cannot be optimized directly, its instantiations under popular $\phi$-divergences provide intuitive and actionable guidance on what properties $P_{\mathcal S}$ should satisfy.
 
Our algorithmic building of $P_{\mathcal S}$ is directly inspired by instantiating the adversarial model discrepancy $D_{\phi}^{\hat{\mathcal X}_r}(P_{\mathcal T}\|P_{\mathcal S})$ with TV, KL, and $\chi^2$ divergences in Corollaries \ref{lemma_tv},\ref{corollary:KL},\ref{corollary_chi2}, respectively. The corresponding Remarks \ref{remark:TV2},\ref{remark:KL2},\ref{remark:chi2_2} discuss in detail how these instantiations inspire practical principles for surrogate construction. Below we summarize key takeaways and refer interested readers to \textit{\bluetwo{Appendices}} \ref{app:inst} for full analyses. 
\begin{itemize}[leftmargin=15pt,itemsep=1pt,topsep=0pt]
\item \textbf{TV instantiation} emphasizes the mismatch in the \emph{mean} adversarial vulnerability pattern. Thus $P_{\mathcal S}$ should be constructed by combining multiple components whose mean vulnerability patterns are diverse rather than redundant, so that their aggregate behavior can smooth the risk from potentially unmatched surrogate choices.
\item \textbf{KL instantiation} penalizes the \emph{under-coverage} of vulnerability patterns that appear with non-negligible probability under $P_{\mathcal T}$. The aforementioned multi-component construction enables $P_{\mathcal S}$ to cover more heterogeneous vulnerability patterns than a single relatively concentrated component could, thus better reducing the risk that a plausible target vulnerability pattern is under-represented.
\item \textbf{$\chi^2$ instantiation} can be reduced by encouraging \emph{variance} of surrogate models' adversarial vulnerability at two levels via a variance decomposition. In addition to enlarging the variance across components as under TV and KL, it further motivates ensuring non-trivial variance among surrogate samples drawn from the same component.
\end{itemize}
Taken together, these three instantiations emphasize different risks and thus offer complementary guidance, \ie, a two-dimensional diversity requirement for $P_{\mathcal S}$, which we formalize as \emph{between-distribution} and \emph{within-distribution} diversity. By enforcing them simultaneously, we control the adversarial model discrepancy, thereby improving adversarial transferability.

\textbf{Between-distribution diversity}
We model $P_{\mathcal{S}}$ as a mixture of $I$ surrogate components $P_\mathcal S=\frac{1}{I}\sum_{i=1}^I P_{\mathcal S_i}$, with each component $P_{\mathcal{S}_i}$ being a distribution over model weights.
Between-distribution diversity requires each component in general to exhibit {diverse adversarial vulnerabilities} \wrt~potential AEs.}

Indeed, the adversarial vulnerability of DNNs is shaped by multiple factors, including model architectures and training objectives \cite{blackboxbench}. While input-transformation-based methods can enhance the surrogate model space from a base model by applying transformations to its input, they may also introduce substantial redundancy in model behaviors \wrt~adversarial vulnerability. 
Instead, we prefer to combine multiple posterior distributions over the model weights, each independently trained on different architectures and with different training strategies, as components composing our surrogate model distribution. By doing so, we aim at enriching model behaviors in $P_{\mathcal{S}}$ \wrt~adversarial vulnerability. In particular, we must carefully choose surrogate components so that their vulnerabilities differ significantly and thus contribute efficiently to the overall goal. Empirical analyses have investigated the varying adversarial vulnerabilities across different models. 
From the perspective of model architecture family, \cite{blackboxbench,VIT2CNN} \bluetwo{indicate} that the adversarial vulnerabilities of convolutional neural networks (convnets) and metaformers\footnote{Metaformer \cite{poolformer} is a general architecture abstracted from Transformers and their variants.} differ significantly.
From the perspective of training strategies, \cite{blackboxbench} observes that normally trained and adversarially trained models exhibit distinct types of adversarial vulnerabilities.
These insights suggest that, from an adversarial standpoint, diversity in $P_\mathcal{S}$ can be efficiently obtained by simultaneously including the distributions of models from four prototypical categories (which we call \textit{prototypes}): normal and adversarial versions of both convnet and metaformer. Model distributions across these prototypes exhibit significantly different vulnerabilities, making them ideal surrogate components.


\textbf{Within-distribution diversity} 
{
\bluetable
In practice, attackers need to sample surrogate models from each component to conduct their attacks, forming a finite surrogate set $\mathcal{M}_\mathcal{S}$ of size $K=In$, where $n$ i.i.d. surrogate models are sampled from each component $P_{\mathcal{S}_i}$, \ie, $\{{\w^i_j}\}_{j=1}^n\sim P_{\mathcal{S}_i}$. Within-distribution diversity requires diverse vulnerability patterns within each component and leads to model samples from the same component exhibiting diverse prediction behaviors.
}

As observed in \cite{fge}, the optima in the DNN loss surface are in fact connected, rather than isolated, forming a valley of low loss. This valley contains many high-performing and complementary models, which produce meaningfully different predictions, leading to a diverse variety of prediction behaviors. Therefore, samples from the distribution centered at the loss valley are preferred surrogate models. With constant learning rates, gathering the trajectory of weights traversed by SGD is approximately sampling from the distribution centered at the minimum of the loss \cite{stephan2017stochastic}. Finally, we propose to gather SGD proposals during training models from each of aforementioned four prototypes and use proposals as sampled surrogate models within each component. By doing so, we achieve diversity within each surrogate component, while also maintaining component-level diversity by the prototype mixture.


\subsection{Find a Flat Optimum from a Diverse Set of Surrogates}\label{sec:Find a Flat Optimum from a Diverse Set of Surrogates}
Even with the elaborately designed surrogate model set, simply crafting AEs with an ERM optimizer over these models \bluetwo{cannot} be strong enough. Based on the perspective of loss landscape flatness in Theorem \ref{theorem_main}, we theoretically demonstrate that optimizing the AE's flatness strengthens its transferability. 
{An intuitive interpretation is that when pursuing a flat minimum among diverse models, it is more likely to remain in flat areas when applied to unseen target models. As a result, a small shift in the target model’s loss landscape would not significantly increase the attack loss, making the AE less likely to fail.}
In this section, an optimization strategy, an upgraded version of the original RAP which is more compatible with a set of diverse surrogate models, is proposed to optimize flatness effectively and efficiently.

 \textbf{A general flatness-aware optimization} Original RAP solves the bi-level optimization problem as in Equation \ref{eq_obj},
 \begin{equation}\label{equ_rap_goal}
     \min_{\hat{\bm{x}} \in \hat{\mathcal{X}}}\max _{\|\bm{\epsilon}\|_\infty \leq \rho}\frac{1}{K} \sum_{\w_k \in \mathcal{M}_\mathcal{S}} {\ell}\left(f\left(\hat{\bm{x}}+\bm{\epsilon}, \w_k\right), y\right),
 \end{equation}
 by iteratively optimizing the inner maximization \bluetwo{problem} and the outer minimization problem on the surrogate set $\mathcal{M}_\mathcal{S}$. In particular, at each iteration, fixing AE $\hat{\bm{x}}$, the inner maximization optimizes reverse perturbation $\bm{\epsilon}$ via a $T$-step I-FGSM. At each step, $\bm{\epsilon}$ is updated as follows:
\begin{equation}\label{equ_rap}
    \bm{\epsilon} \leftarrow \bm{\epsilon}+\beta_{\bm{\epsilon}}\cdot \operatorname{sign}(\nabla_{\bm{\epsilon}}\frac{1}{K}\sum_{\w_k \in \mathcal{M}_\mathcal{S}}{\ell}(f(\hat{\bm{x}}+\bm{\epsilon}, \w_k), y)),
\end{equation}
where $|\mathcal{M}_\mathcal{S}|=K$, $\beta_{\bm{\epsilon}}$ is the inner step size and $\bm{\epsilon}$ is initialized \bluetwo{to} 0.
Then, fixing reverse perturbation $\bm{\epsilon}$, the outer minimization \bluetwo{updates} AE $\hat{\bm{x}}$ with the gradient calculated by minimizing the empirical surrogate adversarial risk \wrt~$\hat{\bm{x}}+\bm{\epsilon}$:

\begin{equation}\label{equ_rap_update}
\hat{\bm{x}} \leftarrow \Pi_{\gamma}[\hat{\bm{x}}-\beta_{\hat{\bm{x}}} \cdot \operatorname{sign}(\nabla_{\hat{\bm{x}}}\frac{1}{K}\sum_{\w_k \in \mathcal{M}_\mathcal{S}}{\ell}(f(\hat{\bm{x}}+\bm{\epsilon}, \w_k), y))],
\end{equation}
where $\beta_{\hat{\bm{x}}}$ is the outer step size, $\Pi_{\gamma}(\cdot)$ restricts current AE to be within \blue{$\hat{\mathcal{X}}$}, and $\hat{\bm{x}}$ is initialized \bluetwo{to} the benign image. Note that after optimizing the loss of reversely perturbed AE $\hat{\bm{x}}+\bm{\epsilon}$, we should come back to the center point $\hat{\bm{x}}$ to conduct this update.

 \textbf{Model-specific reverse perturbations} In Section \ref{sec: Narrow the surrogate-target discrepancy}, we enrich the surrogate model space to narrow the surrogate-target discrepancy. When optimizing the flatness over these diverse models, original RAP in practice computes a global reverse perturbation, \ie, the $\bm{\epsilon}$ is maximized on an average of per-model losses and shared over the whole surrogate set $\mathcal{M}_\mathcal{S}$. However, \bluetwo{given} the diversity \wrt~adversarial vulnerability existed in $\mathcal{M}_\mathcal{S}$, each surrogate has its own worst-case reverse perturbation on AE and their optimization paths may conflict \cite{svre}, directly optimizing a common reverse perturbation updated by fusing over a set of independent update directions will result in a weaker reverse perturbation than model-specific reverse perturbations. This motivates us to replace Equation \ref{equ_rap_goal} \bluetwo{with} calculating reverse perturbations of different models separately to improve the effectiveness of RAP:
\begin{equation}\label{equ_newobj}
\min_{\hat{\bm{x}} \in \hat{\mathcal{X}}} \frac{1}{K} \sum_{\w_k \in \mathcal{M}_\mathcal{S}} \max _{\|\bm{\epsilon}_k\|_\infty \leq \rho}{\ell}\left(f\left(\hat{\bm{x}}+\bm{\epsilon}_k, \w_k\right), y\right),
\end{equation}
where $\bm{\epsilon}_k$ is calculated on individual models:
\begin{equation}\label{equ_perturb}
    \bm{\epsilon}_k \leftarrow \bm{\epsilon}_k+\beta_{\bm{\epsilon}} \cdot \operatorname{sign}\left(\nabla_{\bm{\epsilon}}{\ell}\left(f\left(\hat{\bm{x}}+\bm{\epsilon}_k, \w_k\right), y\right)\right),\w_k \in \mathcal{M}_\mathcal{S}.
\end{equation}
We call $\bm{\epsilon}_k$ a model-\textbf{D}iversity-compatible \textbf{R}everse \textbf{A}dversarial \textbf{P}erturbation (DRAP). Once DRAP is obtained, the outer minimization \wrt~$\hat{\bm{x}}$ in Equation \ref{equ_newobj} is performed. Given that $\mathcal{M}_\mathcal{S}$ may contain a large number of surrogate models (cf. Theorem \ref{theorem_pac}), directly computing full gradients over the entire surrogate set at each iteration can be computationally inefficient. To address this, we adopt the longitudinal manner update \cite{ghostnet}, where $\hat{\bm{x}}$ is updated iteratively across models in $\mathcal{M}_\mathcal{S}$, one at a time. This reduces memory and \bluetwo{computational} overhead while maintaining the diversity-aware objective: 
\begin{equation}\label{eq: update_ae}
\hat{\bm{x}} \leftarrow \Pi_{\gamma}\left[\hat{\bm{x}}-\beta_{\hat{\bm{x}}} \cdot \operatorname{sign}\left(\nabla_{\hat{\bm{x}}}{\ell}\left(f\left(\hat{\bm{x}}+\bm{\epsilon}_k, \w_k\right), y\right)\right)\right],
\end{equation}
where each update step corresponds to a single surrogate model $\w_k \in \mathcal{M}_\mathcal{S}$. We alternate between generating $\bm{\epsilon}_k$ and updating $\hat{\bm{x}}$ across the model set, which effectively integrates model-specific reverse perturbations while ensuring scalability.

 \textbf{Improving \bluetwo{optimization} stability} In the first several iterations of generating the AE, solving the min-max problem in Equation \ref{equ_newobj} may hinder the AE \bluetwo{from} efficiently converging to the region of high attack performance \cite{rap}. Evidence in Section \ref{alb_sec: sharp} shows that this phenomenon in RAP also exists in DRAP. A \textit{late-start} strategy has been proposed by RAP that only solves outer minimization w.r.t unperturbed $\hat{\bm{x}}$ at the early stage, and then \bluetwo{starts} RAP to seek flatness and further boost transferability. We also utilize this strategy during our optimization.

In addition to the late-start strategy, considering the diversity in models' loss landscapes inherent in DRAP, a velocity vector is accumulated in the gradient across iterations \cite{momentum}, with each iteration observing distinct surrogates, to stabilize the optimization path. Evidence soon presented in {Section \ref{alb_sec: dis}} demonstrates that, although the idea of momentum is widely employed in previous attacks \cite{nisi,pifgsm,cwa}, our method can avoid the gradient overaccumulation which may hinder the attack \cite{blackboxbench,svre,eval} effectiveness and best benefits from it.

Complete pseudo-code of DRAP is shown in Algorithm \ref{mainalgo}.

\vspace{-7pt}
\begin{algorithm}\label{mainalgo}
	\caption{Model-\textbf{D}iversity-Compatible \textbf{R}everse \textbf{A}dversarial \textbf{P}erturbation (DRAP) Algorithm} 
            \begin{algorithmic}[1]
            \STATE \textbf{Require}: benign data $(\x, y)$, perturbation budget $\gamma$, surrogate model distributions $\left\{P_{\mathcal{S}_i}\right\}_{i=1}^{I}$, number of samples within one component $n$, late start iteration number $n_{LS}$, inner step size $\beta_{\bm{\epsilon}}$, inner iteration number $T$, outer step size $\beta_{\hat{\bm{x}}}$, decay factor $\mu$.
            \STATE Initialize $\hat{\bm{x}} \leftarrow \bm{x},\bm{m}\leftarrow \bm{0}$;
		\FOR {$j=0,...,n-1$}
                \FOR {$i=0,...,I-1$}
                \STATE Sample a surrogate model $\w_k$ from $P_{\mathcal{S}_i}$;
                    \IF{$j \geq n_{LS}$}
                    \STATE \# Inner maximization
                    \STATE Initialize $\bm{\epsilon}_k \leftarrow 0$;
                    \FOR {$t=0,...,T-1$}
                        \STATE Update $\bm{\epsilon}_k$ using Equation \ref{equ_perturb};
                    \ENDFOR
                    \ENDIF
                    \STATE \# Outer minimization
                    \STATE Calculate $\bm{g}=\nabla_{\hat{\bm{x}}} {\ell}\left(f\left(\hat{\bm{x}}+\bm{\epsilon}_k, \w_k\right), y\right);$
                    \STATE Update momentum by $\bm{m}=\mu \cdot \bm{m}+\frac{\bm{g}}{\|\bm{g}\|_1}$;
                    \STATE Update $\hat{\bm{x}}=\Pi_{\gamma}\left[\hat{\bm{x}}-\beta_{\hat{\bm{x}}} \cdot \operatorname{sign}\left(\bm{m}\right)\right]$;
                \ENDFOR
		\ENDFOR
            \RETURN $\hat{\bm{x}}$.
	\end{algorithmic} 
\end{algorithm}
\vspace{-7pt}

\section{Experimental Evaluation}
In this section, we conduct comprehensive evaluations to illustrate the soundness of DRAP. Specifically, our experiments are designed to explore the answers to the following questions:
\begin{enumerate}
    \item \textit{How does DRAP compare to previous ones when conducting untargeted and targeted attacks?}
    \item \textit{As a key property of transfer-based attacks, is DRAP scalable to be combined with input-transformation-based methods to further boost transferability?}
    \item \textit{Which aspect of DRAP's optimization bound, the sharpness penalty or model discrepancy penalty, is \bluetwo{more} important?}
\end{enumerate}

We conduct the evaluations on ImageNet \cite{imagenet} and CIFAR-10 \cite{cifar10}. For ImageNet, we follow previous works \cite{ssa, ghostnet, pgn, ti, rap, fia} and use the ImageNet-compatible dataset \footnote{\url{https://github.com/tensorflow/cleverhans/tree/master/examples/nips17_adversarial_competition/dataset}} in the NIPS 2017 adversarial competition, which contains 1,000 images with a resolution of $299 \times 299 \times 3$. For CIFAR-10, we conduct experiments on its test set with 10,000 images. In the following, we only consider ImageNet experiments. We put CIFAR-10 experiment results and its detailed experimental protocol in \textit{Appendix} \ref{app: cifar}.

\subsection{Main Results}\label{Sec: Main Results}

\begin{table*}[t]
\centering
\caption{{Untargeted attack success rates (\%,$\uparrow$) on ImageNet dataset.} The AEs are crafted from five surrogate models (ResNet-50, ConvNeXt-T, ViT-B, ResNet-50(AT)and XCiT-S(AT)), against 31 target models falling into four prototypes (normally and adversarially trained convnets and metaformers). \textbf{Bold} denotes the best results and \underline{underlined} denotes the second best results. \blue{More \bluetwo{comparisons} to model-tuning\bluetwo{-}based and feature-based attacks \bluetwo{are} provided in Table \ref{tab: main_continued} of Appendix \ref{app: compare_f_m}.}}
\label{tab: main}
\resizebox{\textwidth}{!}{%
\begin{threeparttable}
\begin{tabular}{@{}cc|cccccccccccccccccccc@{}}
\toprule
\multicolumn{2}{c|}{\textbf{Target Model Set}} & \begin{tabular}[c]{@{}c@{}}I-FGSM\\ \cite{ifgsm}\end{tabular} & \begin{tabular}[c]{@{}c@{}}DI2-FGSM\\ \cite{difgsm}\end{tabular} & \begin{tabular}[c]{@{}c@{}}SI-FGSM\\ \cite{nisi}\end{tabular} & \begin{tabular}[c]{@{}c@{}}Admix\\ \cite{admix}\end{tabular} & \begin{tabular}[c]{@{}c@{}}TI-FGSM\\ \cite{ti}\end{tabular} & \begin{tabular}[c]{@{}c@{}}SSA\\ \cite{ssa}\end{tabular} & \begin{tabular}[c]{@{}c@{}}SIA\\ \cite{sia}\end{tabular} & \begin{tabular}[c]{@{}c@{}}MI-FGSM\\ \cite{mifgsm}\end{tabular} & \begin{tabular}[c]{@{}c@{}}PI-FGSM\\ \cite{pifgsm}\end{tabular} & \begin{tabular}[c]{@{}c@{}}VT-FGSM\\ \cite{vt}\end{tabular} & \begin{tabular}[c]{@{}c@{}}PGN\\ \cite{pgn}\end{tabular} & \blue{\begin{tabular}[c]{@{}c@{}}TPA\\ \cite{tpa}\end{tabular}} & \blue{\begin{tabular}[c]{@{}c@{}}FEM\\ \cite{fem}\end{tabular}} & \blue{\begin{tabular}[c]{@{}c@{}}MEF\\ \cite{mef}\end{tabular}} & \begin{tabular}[c]{@{}c@{}}SVRE\\ \cite{svre}\end{tabular} & \blue{\begin{tabular}[c]{@{}c@{}}AdaEA\\ \cite{adaea}\end{tabular}} & \blue{\begin{tabular}[c]{@{}c@{}}SMER\\ \cite{smer}\end{tabular}} & \begin{tabular}[c]{@{}c@{}}CWA\\ \cite{cwa}\end{tabular} & \begin{tabular}[c]{@{}c@{}}RAP\\ \cite{rap}\end{tabular} & \textbf{DRAP} \\ \midrule
\multicolumn{1}{c|}{} & AlexNet\cite{alexnet} & 44.7 & 47.6 & 46.7 & 49.6 & 45.7 & 53.1 & 55.5 & 49.2 & 49.7 & 44.9 & 49.2 & 50.2 & 51.9 & 50.5 & 46.6 & 56.2 & 56.4 & {\ul 57.2} & 46.6 & \textbf{68.5} \\
\multicolumn{1}{c|}{} & VGG-16-BN\cite{vgg} & 52.7 & 66.4 & 66.3 & 81.1 & 57.4 & 75.5 & \textbf{95.6} & 68.1 & 71.8 & 58.3 & 81.7 & 62.7 & 65.2 & 64.6 & 57.9 & 61.7 & 63.4 & 66.7 & 54.4 & {\ul 84.1} \\
\multicolumn{1}{c|}{} & DenseNet-201\cite{densenet} & 40.3 & 56.6 & 57.9 & 71.9 & 46.5 & 61.0 & \textbf{90.8} & 59.5 & 62.2 & 47.1 & 74.4 & 54.6 & 58.7 & 57.3 & 51.0 & 54.4 & 57.8 & 59.7 & 44.5 & {\ul 81.5} \\
\multicolumn{1}{c|}{} & GoogLeNet\cite{googlenet} & 32.2 & 42.8 & 42.2 & 55.0 & 35.3 & 53.7 & {\ul 73.0} & 45.6 & 48.4 & 36.1 & 56.1 & 42.7 & 51.5 & 48.6 & 38.1 & 47.5 & 49.6 & 50.4 & 36.7 & \textbf{73.9} \\
\multicolumn{1}{c|}{} & ShuffleNetV2\cite{shufflenet} & 42.7 & 51.9 & 52.4 & 63.9 & 44.2 & 62.2 & {\ul 77.3} & 54.8 & 56.7 & 45.1 & 63.1 & 51.3 & 58.3 & 56.4 & 48.6 & 57.8 & 59.8 & 61.2 & 46.1 & \textbf{82.1} \\
\multicolumn{1}{c|}{} & MobileNetV2\cite{mobilenetv2} & 47.6 & 62.4 & 61.5 & 75.3 & 52.2 & 70.7 & \textbf{93.6} & 63.4 & 68.6 & 53.9 & 76.3 & 57.9 & 66.3 & 64.5 & 56.4 & 63.6 & 64.5 & 67.3 & 51.7 & {\ul 87.2} \\
\multicolumn{1}{c|}{} & MobileNetV3-L\cite{mobilenetv3} & 33.5 & 49.5 & 45.5 & 60.5 & 37.4 & 65.1 & {\ul 83.1} & 48.5 & 52.8 & 38.0 & 64.8 & 44.4 & 55.9 & 52.7 & 41.4 & 53.8 & 57.0 & 59.6 & 39.2 & \textbf{85.5} \\
\multicolumn{1}{c|}{} & MNASNet\cite{mnasnet} & 42.1 & 57.2 & 56.8 & 72.9 & 46.6 & 68.2 & \textbf{92.3} & 58.8 & 63.8 & 48.8 & 72.1 & 55.2 & 62.2 & 62.0 & 52.2 & 58.1 & 62.7 & 63.8 & 49.0 & {\ul 87.6} \\
\multicolumn{1}{c|}{} & EfficientNet\cite{efficientnet} & 31.5 & 46.8 & 40.7 & 49.6 & 35.0 & 56.2 & \textbf{75.9} & 44.7 & 46.6 & 35.1 & 55.8 & 42.2 & 48.1 & 48.2 & 38.4 & 46.0 & 52.7 & 52.7 & 36.8 & {\ul 74.9} \\
\multicolumn{1}{c|}{} & ConvNeXt-L\cite{convnet} & 36.3 & 50.2 & 45.7 & 66.5 & 36.8 & 68.3 & \textbf{91.4} & 58.0 & 59.1 & 42.7 & 67.8 & 50.8 & 35.1 & 41.5 & 50.9 & 38.5 & 60.8 & 57.6 & 46.7 & {\ul 77.4} \\
\multicolumn{1}{c|}{\multirow{-11}{*}{\textbf{\begin{tabular}[c]{@{}c@{}}ConvNet\\ Set\end{tabular}}}} & \cellcolor[HTML]{D9D9D9}\textit{\textbf{Average}} & \cellcolor[HTML]{D9D9D9}40.4 & \cellcolor[HTML]{D9D9D9}53.1 & \cellcolor[HTML]{D9D9D9}51.6 & \cellcolor[HTML]{D9D9D9}64.6 & \cellcolor[HTML]{D9D9D9}43.7 & \cellcolor[HTML]{D9D9D9}63.4 & \cellcolor[HTML]{D9D9D9}\textbf{82.9} & \cellcolor[HTML]{D9D9D9}55.1 & \cellcolor[HTML]{D9D9D9}58.0 & \cellcolor[HTML]{D9D9D9}45.0 & \cellcolor[HTML]{D9D9D9}66.1 & \cellcolor[HTML]{D9D9D9}51.2 & \cellcolor[HTML]{D9D9D9}55.3 & \cellcolor[HTML]{D9D9D9}54.6 & \cellcolor[HTML]{D9D9D9}48.2 & \cellcolor[HTML]{D9D9D9}53.8 & \cellcolor[HTML]{D9D9D9}58.5 & \cellcolor[HTML]{D9D9D9}59.6 & \cellcolor[HTML]{D9D9D9}45.2 & \cellcolor[HTML]{D9D9D9}{\ul 80.3} \\ \midrule
\multicolumn{1}{c|}{} & ViT-S\cite{vit} & 10.0 & 20.2 & 13.1 & 18.8 & 12.3 & 24.9 & \textbf{40.2} & 19.2 & 20.2 & 11.6 & 22.8 & 16.9 & 18.0 & 22.2 & 16.1 & 17.7 & 22.5 & 22.1 & 16.9 & {\ul 38.8} \\
\multicolumn{1}{c|}{} & DeiT-S\cite{deit} & 14.1 & 26.5 & 17.3 & 23.5 & 17.3 & 32.4 & {\ul 45.2} & 25.2 & 25.1 & 15.6 & 26.9 & 22.4 & 23.7 & 25.8 & 20.1 & 25.0 & 27.8 & 29.4 & 20.0 & \textbf{53.5} \\
\multicolumn{1}{c|}{} & PoolFormer-S\cite{poolformer} & 29.0 & 49.6 & 37.2 & 51.5 & 33.4 & 62.1 & \textbf{86.2} & 46.0 & 49.9 & 33.4 & 58.1 & 39.8 & 44.6 & 45.0 & 40.3 & 38.5 & 46.0 & 46.6 & 36.2 & {\ul 71.8} \\
\multicolumn{1}{c|}{} & TNT-S\cite{tnt} & 13.3 & 26.8 & 17.6 & 25.5 & 16.2 & 36.1 & \textbf{57.3} & 25.5 & 26.2 & 16.5 & 29.3 & 21.2 & 23.3 & 27.5 & 21.5 & 23.4 & 27.0 & 27.8 & 22.2 & {\ul 52.5} \\
\multicolumn{1}{c|}{} & Swin-S\cite{swin} & 8.8 & 19.8 & 11.9 & 17.9 & 11.0 & 26.5 & \textbf{42.9} & 18.1 & 17.6 & 11.0 & 20.7 & 14.5 & 15.2 & 17.9 & 14.0 & 14.8 & 16.6 & 18.4 & 15.4 & {\ul 28.5} \\
\multicolumn{1}{c|}{} & XCiT-S\cite{xcit} & 11.3 & 27.9 & 13.2 & 16.1 & 11.8 & 29.1 & \textbf{43.0} & 18.4 & 18.9 & 12.2 & 19.7 & 16.7 & 17.8 & 19.5 & 16.5 & 17.4 & 19.5 & 20.3 & 16.3 & {\ul 30.6} \\
\multicolumn{1}{c|}{} & CaiT-S\cite{cait} & 5.0 & 19.2 & 6.8 & 8.6 & 6.1 & 17.4 & \textbf{29.4} & 10.1 & 10.4 & 6.2 & 12.2 & 9.2 & 11.3 & 12.0 & 7.8 & 9.8 & 11.0 & 12.5 & 9.5 & {\ul 22.8} \\
\multicolumn{1}{c|}{\multirow{-8}{*}{\textbf{\begin{tabular}[c]{@{}c@{}}Metaformer\\ Set\end{tabular}}}} & \cellcolor[HTML]{D9D9D9}\textit{\textbf{Average}} & \cellcolor[HTML]{D9D9D9}13.1 & \cellcolor[HTML]{D9D9D9}27.1 & \cellcolor[HTML]{D9D9D9}16.7 & \cellcolor[HTML]{D9D9D9}23.1 & \cellcolor[HTML]{D9D9D9}15.4 & \cellcolor[HTML]{D9D9D9}32.6 & \cellcolor[HTML]{D9D9D9}\textbf{49.2} & \cellcolor[HTML]{D9D9D9}23.2 & \cellcolor[HTML]{D9D9D9}24.0 & \cellcolor[HTML]{D9D9D9}15.2 & \cellcolor[HTML]{D9D9D9}27.1 & \cellcolor[HTML]{D9D9D9}20.1 & \cellcolor[HTML]{D9D9D9}22.0 & \cellcolor[HTML]{D9D9D9}24.3 & \cellcolor[HTML]{D9D9D9}19.5 & \cellcolor[HTML]{D9D9D9}20.9 & \cellcolor[HTML]{D9D9D9}24.3 & \cellcolor[HTML]{D9D9D9}25.3 & \cellcolor[HTML]{D9D9D9}19.5 & \cellcolor[HTML]{D9D9D9}{\ul 42.6} \\ \midrule
\multicolumn{1}{c|}{} & RaWideResNet-101-2\cite{peng2023robust} & 17.0 & 17.6 & 17.3 & 17.4 & 17.1 & 19.5 & 17.7 & 19.0 & 18.7 & 17.1 & 17.8 & 19.1 & 18.5 & 18.1 & 17.8 & 25.0 & {\ul 25.0} & 24.2 & 16.3 & \textbf{26.8} \\
\multicolumn{1}{c|}{} & WideResNet-50-2\cite{salman2020adversarially} & 21.9 & 22.6 & 22.1 & 22.1 & 22.3 & 25.3 & 23.1 & 24.0 & 24.0 & 22.0 & 23.5 & 24.5 & 23.9 & 23.4 & 23.3 & 31.7 & 32.2 & {\ul 32.5} & 22.0 & \textbf{34.3} \\
\multicolumn{1}{c|}{} & ResNet-50\cite{wong2020fast} & 39.7 & 40.2 & 39.9 & 40.4 & 40.6 & 43.4 & 41.9 & 41.7 & 41.8 & 39.6 & 41.3 & 42.6 & 42.1 & 41.1 & 41.1 & 47.6 & {\ul 48.6} & 47.6 & 39.7 & \textbf{51.4} \\
\multicolumn{1}{c|}{} & ConvNeXt-L\cite{liu2024comprehensive} & 10.4 & 10.8 & 10.5 & 10.8 & 10.8 & 12.5 & 11.3 & 11.7 & 11.7 & 10.4 & 11.5 & 11.6 & 11.4 & 11.2 & 11.2 & 16.4 & {\ul 16.5} & 16.4 & 10.7 & \textbf{17.6} \\
\multicolumn{1}{c|}{} & ConvNeXt-B\cite{liu2024comprehensive} & 10.6 & 10.9 & 10.4 & 11.0 & 11.1 & 13.3 & 11.2 & 12.3 & 12.4 & 10.7 & 12.0 & 12.6 & 12.1 & 11.9 & 11.3 & 17.7 & 17.3 & {\ul 17.8} & 11.1 & \textbf{18.8} \\
\multicolumn{1}{c|}{} & ConvNeXt-L-ConvStem\cite{singh2024revisiting} & 10.2 & 10.6 & 9.9 & 10.7 & 10.3 & 12.5 & 10.8 & 11.2 & 11.1 & 10.1 & 11.0 & 11.3 & 10.7 & 11.0 & 10.7 & {\ul 16.1} & 15.3 & 15.9 & 10.5 & \textbf{16.8} \\
\multicolumn{1}{c|}{} & ConvNeXt-B-ConvStem\cite{singh2024revisiting} & 11.7 & 11.8 & 11.7 & 12.0 & 12.2 & 14.4 & 12.8 & 12.9 & 12.8 & 11.4 & 12.4 & 13.2 & 12.3 & 12.7 & 12.2 & {\ul 18.7} & {\ul 18.7} & {\ul 18.7} & 11.2 & \textbf{19.8} \\
\multicolumn{1}{c|}{} & Inc-v3$_{\textit{ens3}}$\cite{tf} & 9.7 & 13.7 & 10.3 & 10.9 & 9.8 & {\ul 19.1} & 17.8 & 12.7 & 13.8 & 10.1 & 13.0 & 12.2 & 13.9 & 13.5 & 11.6 & 16.6 & 16.1 & 17.1 & 11.6 & \textbf{24.2} \\
\multicolumn{1}{c|}{} & Inc-v3$_{\textit{ens4}}$\cite{tf} & 11.4 & 15.7 & 12.8 & 14.0 & 11.6 & {\ul 21.5} & 20.0 & 14.8 & 14.9 & 11.0 & 14.9 & 13.4 & 17.6 & 16.8 & 12.0 & 18.3 & 19.4 & 19.0 & 13.3 & \textbf{26.4} \\
\multicolumn{1}{c|}{} & IncRes-v2$_{\textit{ens}}$\cite{tf} & 3.6 & 6.4 & 4.6 & 5.5 & 3.9 & {\ul 10.4} & 8.7 & 6.1 & 6.5 & 4.0 & 7.3 & 5.7 & 7.2 & 7.1 & 4.5 & 8.2 & 9.1 & 9.0 & 5.6 & \textbf{13.2} \\
\multicolumn{1}{c|}{\multirow{-11}{*}{\textbf{\begin{tabular}[c]{@{}c@{}}ConvNet(AT)\\ Set\end{tabular}}}} & \cellcolor[HTML]{D9D9D9}\textit{\textbf{Average}} & \cellcolor[HTML]{D9D9D9}14.6 & \cellcolor[HTML]{D9D9D9}16.0 & \cellcolor[HTML]{D9D9D9}15.0 & \cellcolor[HTML]{D9D9D9}15.5 & \cellcolor[HTML]{D9D9D9}15.0 & \cellcolor[HTML]{D9D9D9}19.2 & \cellcolor[HTML]{D9D9D9}17.5 & \cellcolor[HTML]{D9D9D9}16.6 & \cellcolor[HTML]{D9D9D9}16.8 & \cellcolor[HTML]{D9D9D9}14.6 & \cellcolor[HTML]{D9D9D9}16.5 & \cellcolor[HTML]{D9D9D9}16.6 & \cellcolor[HTML]{D9D9D9}17.0 & \cellcolor[HTML]{D9D9D9}16.7 & \cellcolor[HTML]{D9D9D9}15.6 & \cellcolor[HTML]{D9D9D9}21.6 & \cellcolor[HTML]{D9D9D9}{\ul 21.8} & \cellcolor[HTML]{D9D9D9}{\ul 21.8} & \cellcolor[HTML]{D9D9D9}15.2 & \cellcolor[HTML]{D9D9D9}\textbf{24.9} \\ \midrule
\multicolumn{1}{c|}{} & Swin-B\cite{liu2024comprehensive} & 11.5 & 12.3 & 11.9 & 11.7 & 12.0 & 13.8 & 12.1 & 12.8 & 12.8 & 11.7 & 12.4 & 13.0 & 12.6 & 12.4 & 12.5 & {\ul 17.8} & 17.0 & 17.3 & 11.9 & \textbf{18.6} \\
\multicolumn{1}{c|}{} & Swin-L\cite{liu2024comprehensive} & 9.8 & 10.0 & 10.0 & 9.9 & 9.9 & 11.5 & 10.8 & 10.6 & 10.5 & 10.0 & 10.7 & 11.0 & 10.6 & 10.5 & 10.7 & {\ul 15.0} & 14.6 & 14.8 & 10.1 & \textbf{15.6} \\
\multicolumn{1}{c|}{} & XCiT-L\cite{debenedetti2023light} & 14.9 & 15.5 & 15.0 & 15.3 & 15.2 & 18.4 & 16.0 & 17.2 & 17.2 & 14.9 & 17.2 & 17.3 & 17.3 & 16.2 & 16.0 & {\ul 28.0} & 26.5 & 26.8 & 14.4 & \textbf{28.0} \\
\multicolumn{1}{c|}{} & ViT-B-ConvStem\cite{singh2024revisiting} & 11.2 & 11.6 & 11.5 & 11.6 & 11.3 & 12.9 & 12.1 & 12.4 & 12.2 & 11.4 & 12.4 & 12.5 & 12.2 & 12.1 & 11.8 & {\ul 17.9} & 17.7 & 17.8 & 11.7 & \textbf{18.6} \\
\multicolumn{1}{c|}{\multirow{-5}{*}{\textbf{\begin{tabular}[c]{@{}c@{}}Metaformer(AT)\\ Set\end{tabular}}}} & \cellcolor[HTML]{D9D9D9}\textit{\textbf{Average}} & \cellcolor[HTML]{D9D9D9}11.9 & \cellcolor[HTML]{D9D9D9}12.4 & \cellcolor[HTML]{D9D9D9}12.1 & \cellcolor[HTML]{D9D9D9}12.1 & \cellcolor[HTML]{D9D9D9}12.1 & \cellcolor[HTML]{D9D9D9}14.2 & \cellcolor[HTML]{D9D9D9}12.8 & \cellcolor[HTML]{D9D9D9}13.3 & \cellcolor[HTML]{D9D9D9}13.2 & \cellcolor[HTML]{D9D9D9}12.0 & \cellcolor[HTML]{D9D9D9}13.2 & \cellcolor[HTML]{D9D9D9}13.4 & \cellcolor[HTML]{D9D9D9}13.2 & \cellcolor[HTML]{D9D9D9}12.8 & \cellcolor[HTML]{D9D9D9}12.8 & \cellcolor[HTML]{D9D9D9}{\ul 19.7} & \cellcolor[HTML]{D9D9D9}19.0 & \cellcolor[HTML]{D9D9D9}19.2 & \cellcolor[HTML]{D9D9D9}12.0 & \cellcolor[HTML]{D9D9D9}\textbf{20.2} \\ \midrule
\rowcolor[HTML]{D9D9D9} 
\multicolumn{2}{c|}{\cellcolor[HTML]{D9D9D9}\textit{\textbf{Overall Average}}} & 22.2 & 30.0 & 26.8 & 32.6 & 24.0 & 35.8 & {\ul 45.1} & 30.1 & 31.2 & 24.2 & 34.5 & 28.2 & 30.0 & 30.1 & 26.6 & 31.6 & 33.8 & 34.5 & 25.4 & \textbf{46.2} \\ \bottomrule
\end{tabular}
\end{threeparttable}}
\vspace{-8pt}
\end{table*}

\textbf{Baselines} To answer question 1), \blue{we compare DRAP with baselines from four categories. Input-transformation-based attacks include} DI2-FGSM\cite{difgsm}, SI-FGSM\cite{nisi}, Admix\cite{admix}, TI-FGSM\cite{ti}, SSA\cite{ssa}, SIA\cite{sia}. Optimization-based attacks include MI-FGSM\cite{mifgsm}, PI-FGSM\cite{pifgsm}, VT-FGSM\cite{vt}, RAP\cite{rap}, PGN\cite{pgn}, CWA\cite{cwa}, \blue{TPA\cite{tpa}, FEM\cite{fem}, MFE\cite{mef}. Model-based attacks include tuning-focusing methods GhostNet\cite{ghostnet} and Bayesian attack\cite{bayesian}, and fusing-focusing methods} SVRE\cite{svre}, \blue{AdaEA\cite{adaea} and SMER\cite{smer}. Feature-based attacks include FIA\cite{fia} and NAA\cite{naa}.}

\textbf{Models} 
For surrogate models, \blue{we select surrogate architectures from four prototypes according to training strategy (normally vs. adversarially trained) and model family (convnet vs. metaformer)}: ResNet-50\cite{resnet} and ConvNeXt-T \cite{convnet} from normally trained convnets, ViT-B \cite{vit} from normally trained metaformers, ResNet-50(AT)\cite{resnetat} from adversarially trained convnets and XCiT-S(AT) \cite{debenedetti2023light} from adversarially trained metaformers. ConvNeXt-T is included alongside ResNet-50 because it's a special convnet which follows designs popularized by vision transformers. For DRAP, \blue{these five architectures lead to $I=5$ surrogate components. Within each component,} model samples are gathered as proposals at every epoch during fine-tuning the pretrained model of the corresponding surrogate architecture, which are optimized using their respective training recipes over $n=40$ additional epochs. \blue{In total, this yields $K=I \times n$ surrogate models.}
In order to get more diverse samples, we fine-tune the five pretrained models with relatively larger constant learning rates, specifically 0.05, 0.001, 0.05, 0.5, and 0.001, respectively, for ResNet-50, ConvNeXt-T, ViT-B, ResNet-50(AT), and XCiT-S(AT), while without significantly degrading their clean accuracy. 
\blue{As shown in Algorithm \ref{mainalgo}, DRAP sequentially draws samples from each of the $I$ surrogate components in turn and updates the AE on one surrogate model at a time until all $K$ models have been used, thereby combining within‑distribution and between‑distribution diversity.}
For compared methods, AEs are crafted on the five pretrained models of these five surrogate architectures \blue{by averaging the logits proposed in \cite{mifgsm} standardly following \cite{lpm, difgsm, ti, pgn}}. 
To evaluate the transferability of AEs, we collect 31 target models to ensure comprehensive coverage of diverse model architectures from the four prototypical categories, abbreviated as ConvNet Set, Metaformer Set, ConvNet(AT) Set, and Metaformer(AT) Set, as shown in \bluetwo{Table} \ref{tab: main}.

\textbf{Implementation Details} For the untargeted attack scenario, the adversarial perturbation is bounded by $\gamma=4/255$ with step size $\beta_{\hat{\bm{x}}}=2/255$ for all methods. For the targeted attack scenario, the adversarial perturbation is bounded by $\gamma=16/255$ with step size $\beta_{\hat{\bm{x}}}=8/255$ for all methods.
We set the iteration number \blue{$n_{iter}$} of MI, PI and CWA as 10 when conducting untargeted attacks, as suggested in their original papers, because their performance deteriorates for additional rounds. For RAP, the iteration number \blue{$n_{iter}$} is 400. Otherwise, the iteration number \blue{$n_{iter}$} is set as 200. 
For the hyper-parameters of DRAP, we set the number of samples within one surrogate component $n=40$, inner iteration number $T=5$, late start iteration number $n_{LS}=5$, inner step size $\beta_{\bm{\epsilon}}=0.1/255$, decay factor $\mu=1$. Note that the number of iterations for updating AE in DRAP is the same as others, as $n\times I=200$. For compared methods, we follow the protocol in \href{https://github.com/SCLBD/blackboxbench}{BlackboxBench} benchmark.

\textbf{Results of Untargeted Attacks} We first summarize the untargeted attack results on ImageNet dataset against convnet set, metaformer set, adversarially trained convnet set and adversarially trained metaformer set, as shown in Table \ref{tab: main}. DRAP achieves a substantial improvement in the average attack success rate across all target models compared to \blue{baselines from four categories of transfer-based attack.}
Taking a closer look at the comparison results, we \bluetwo{find} that, equipped with the same surrogate models, the state-of-the-art attack SIA is competitive on the relatively easier-to-attack normally trained target model sets. However, its performance is unsatisfactory when attacking the two adversarially trained target model sets than DRAP. DRAP provides a larger performance gain on attacking models with the defense mechanism while maintaining an acceptable performance on normal models, striking the balance among the whole target model sets. These results suggest that striving for flatness among all surrogate models \bluetwo{while simultaneously} considering model diversity could provide a strong guarantee on the transferability of AEs, regardless of the robustness of the target model sets.
\blue{AdaEA, SMER and }CWA are also shown as promising prior methods which could effectively utilize all diverse surrogate models simultaneously. Taking CWA as an example, by attacking surrogates' common weakness and optimizing the flatness, CWA \bluetwo{achieves} improved attack performance on the challenging adversarially trained model sets. However, it fails to explicitly address the surrogate-target model gap. Furthermore, its use of a universal perturbation across the model ensemble may hinder the optimization of flatness across diverse surrogates, resulting in an overall lower attack success rate compared to DRAP. In {Section \ref{alb_sec: sharp}} Q1-2, we further explore various solvers for optimizing flatness among diverse models. To sum up, we view the consistency with which DRAP outperforms the best prior methods, which change across different model sets, as a major advantage of the proposed method.
Notably, in this experiment, the unseen target models are drawn from the same prototypical model sets as the surrogate models. In the ablative study in {Section \ref{alb_sec: dis}} Q2-2, we consider stricter attack scenarios where the surrogate-target model shifts are more prominent to further evaluate the effectiveness of DRAP.

\textbf{Results of Targeted Attacks} The targeted attack results of baseline attacks and DRAP are shown in Table \ref{tab:main_t}. At first glance, it is evident that targeted attacks pose a more challenging scenario, particularly on the two model sets equipped with defense mechanisms. \blue{This is because, in the targeted setting, a successful AE must lie in a typically small target-class region shared across models, and this feasible region can be further constrained by defense models, making targeted attacks much more sensitive to surrogate–target shift than untargeted attacks.} Despite using a looser perturbation constraint, most methods \blue{without explicit mechanisms to mitigate the shift} will fail to induce even a single misclassification as the target label on the defense models. Among the prior methods, SIA demonstrates competitive performance in the targeted setting. For instance, it achieves relatively larger improvements on the two normal model sets, with success rates of 62.0\% and 41.5\%, respectively. \blue{AdaEA, SMER, and} CWA are the only prior methods that report success on the challenging defense model sets. \blue{However, by encouraging AEs to lie in flat low-loss regions of the loss landscapes across diverse surrogates, DRAP makes the optimized AEs more likely to also reside in flat low-loss regions on unseen targets that resist model shifts, leading to the best attack success rate by a significant margin across various target models. The results indicate that DRAP is also more effective under the targeted attack scenario, highlighting its strong robustness to model shift.}

\begin{table*}[t]
\centering
\caption{{Targeted attack success rates (\%,$\uparrow$) on ImageNet dataset.} The AEs are crafted from five surrogate models (ResNet-50, ConvNeXt-T, ViT-B, ResNet-50(AT) and XCiT-S(AT)), against 31 target models falling into four prototypes (normally and adversarially trained convnets and metaformers). The results are averaged on each model \bluetwo{set}. Full results broken down into each models are shown in {\textit{Appendix} \ref{app: full}}. \textbf{Bold} denotes the best results and \underline{underlined} denotes the second best results.}
\label{tab:main_t}
\resizebox{\textwidth}{!}{%
\begin{tabular}{@{}c|cccccccccccccccccccc@{}}
\toprule
\textbf{Target Model Set} & \begin{tabular}[c]{@{}c@{}}I-FGSM\\ \cite{ifgsm}\end{tabular} & \begin{tabular}[c]{@{}c@{}}DI2-FGSM\\ \cite{difgsm}\end{tabular} & \begin{tabular}[c]{@{}c@{}}SI-FGSM\\ \cite{nisi}\end{tabular} & \begin{tabular}[c]{@{}c@{}}Admix\\ \cite{admix}\end{tabular} & \begin{tabular}[c]{@{}c@{}}TI-FGSM\\ \cite{ti}\end{tabular} & \begin{tabular}[c]{@{}c@{}}SSA\\ \cite{ssa}\end{tabular} & \begin{tabular}[c]{@{}c@{}}SIA\\ \cite{sia}\end{tabular} & \begin{tabular}[c]{@{}c@{}}MI-FGSM\\ \cite{mifgsm}\end{tabular} & \begin{tabular}[c]{@{}c@{}}PI-FGSM\\ \cite{pifgsm}\end{tabular} & \begin{tabular}[c]{@{}c@{}}VT-FGSM\\ \cite{vt}\end{tabular} & \begin{tabular}[c]{@{}c@{}}PGN\\ \cite{pgn}\end{tabular} & \blue{\begin{tabular}[c]{@{}c@{}}TPA\\ \cite{tpa}\end{tabular}} & \blue{\begin{tabular}[c]{@{}c@{}}FEM\\ \cite{fem}\end{tabular}} & \blue{\begin{tabular}[c]{@{}c@{}}MEF\\ \cite{mef}\end{tabular}} & \begin{tabular}[c]{@{}c@{}}SVRE\\ \cite{svre}\end{tabular} & \blue{\begin{tabular}[c]{@{}c@{}}AdaEA\\ \cite{adaea}\end{tabular}} & \blue{\begin{tabular}[c]{@{}c@{}}SMER\\ \cite{smer}\end{tabular}} & \begin{tabular}[c]{@{}c@{}}CWA\\ \cite{cwa}\end{tabular} & \begin{tabular}[c]{@{}c@{}}RAP\\ \cite{rap}\end{tabular} & \textbf{DRAP} \\ \midrule
\textbf{ConvNet Set} & 5.4 & 22.0 & 11.6 & 11.3 & 8.2 & 29.8 & {\ul 62.0} & 5.3 & 9.7 & 6.9 & 37.4 & 1.5 & 6.3 & 41.0 & 18.5 & 34.6 & 21.3 & 36.0 & 14.5 & \textbf{77.6} \\
\textbf{Metaformer Set} & 0.5 & 11.6 & 1.7 & 1.6 & 1.1 & 16.5 & {\ul 41.5} & 1.0 & 1.8 & 1.1 & 13.3 & 0.2 & 1.4 & 20.6 & 10.5 & 18.9 & 12.7 & 21.9 & 3.1 & \textbf{56.4} \\
\textbf{ConvNet(AT) Set} & 0.0 & 0.0 & 0.0 & 0.0 & 0.0 & 0.3 & 0.1 & 0.0 & 0.1 & 0.0 & 0.1 & 0.0 & 0.0 & 0.8 & 0.1 & {\ul 13.3} & 3.5 & 6.4 & 0.1 & \textbf{13.4} \\
\textbf{Metaformer(AT) Set} & 0.0 & 0.0 & 0.0 & 0.0 & 0.0 & 0.1 & 0.0 & 0.0 & 0.0 & 0.0 & 0.1 & 0.0 & 0.0 & 0.1 & 0.0 & \textbf{11.4} & 3.4 & {\ul 7.5} & 0.2 & \textbf{11.4} \\ \midrule
\rowcolor[HTML]{D9D9D9} 
\textit{\textbf{Overall Average}} & 1.8 & 9.7 & 4.1 & 4.0 & 2.9 & 13.4 & {\ul 29.4} & 1.9 & 3.6 & 2.5 & 15.1 & 0.5 & 2.4 & 18.1 & 8.4 & 21.2 & 11.3 & 19.5 & 5.4 & \textbf{43.5} \\ \bottomrule
\end{tabular}
}
\vspace{-10pt}
\end{table*}

\subsection{Composition with Input-Transformation-Based Attacks}
\bluetwo{Besides} improving transferability from an optimization perspective as considered in DRAP, another related panoply of methods \bluetwo{introduces} randomness into the input via various transformations. Prior research has demonstrated that combining these two perspectives could achieve state-of-the-art transferability. As the outer minimization with $\hat{\bm{x}}$ in DRAP, \ie, Equation \ref{eq: update_ae}, could be solved by any off-the-shelf strategies, including the input-transformation-based methods, our method could also be seamlessly combined with them. To answer question 2), we explore the behavior of DRAP, as well as some well-known optimization-based attacks such as MI-FGSM\cite{mifgsm}, PI-FGSM\cite{pifgsm}, VT\cite{vt}, our original RAP\cite{rap}, PGN\cite{pgn} and CWA\cite{cwa}, when combined with input-transformation-based attacks, namely, DI2-FGSM\cite{difgsm}, TI-FGSM\cite{ti}, Admix\cite{admix} and SIA\cite{sia}. 
We omit the combination of PGN and SIA due to the heavy computational demands. The experimental protocol follows the untargeted setting in Section \ref{Sec: Main Results}. 
The resultant composite attack performance is shown in Table \ref{tab:composite}. 
As can be observed, for both convnets and metaformers and for both normally trained and adversarially trained models, combining DRAP with existing input-transformation-based attacks can significantly improve the base version, leading to a new state-of-the-art attack performance. For example, SIA achieves a competitive average attack success rate of 45.1\% (cf. Table \ref{tab: main}), while integrating with DRAP further improves it by a clear margin of 6.0\%. These remarkable improvements validate the scalability of our method when combined with others to further boost adversarial transferability.
Additionally, we view the consistency with which the extensions of DRAP outperform those of prior optimization-based methods, confirming the superiority of DRAP.

\begin{table}[t]
\centering
\caption{Attack success rates (\%, $\uparrow$) of MI, PI, VT, RAP, PGN, CWA and DRAP, when it is integrated with DI, TI, Admix and SIA, respectively. The indentation denotes combination. The results are averaged on each model \bluetwo{set}.}
\label{tab:composite}
\resizebox{0.97\columnwidth}{!}{%
\begin{tabular}{@{}l|cccc|c@{}}
\toprule
\multicolumn{1}{c|}{\textbf{Attack}} & \textbf{\begin{tabular}[c]{@{}c@{}}ConvNet\\ Set\end{tabular}} & \textbf{\begin{tabular}[c]{@{}c@{}}Metaformer\\ Set\end{tabular}} & \textbf{\begin{tabular}[c]{@{}c@{}}ConvNet\\(AT) Set\end{tabular}} & \textbf{\begin{tabular}[c]{@{}c@{}}Metaformer\\(AT) Set\end{tabular}} & \textit{\textbf{\begin{tabular}[c]{@{}c@{}}Overall\\ Average\end{tabular}}} \\ \midrule 
DI2-FGSM & 53.1 & 27.1 & 16.0 & 12.4 & 30.0 
\\
\quad + MI & 75.0 & 50.6 & 18.9 & 13.7 & 43.5 
\\
\quad + PI & 67.1 & 40.0 & 18.5 & 13.7 & 38.4 
\\
\quad + VT & 56.6 & 31.9 & 16.1 & 12.4 & 32.3 
\\
\quad + RAP & 51.8 & 34.0 & 17.5 & 12.5 & 31.6 
\\
\quad + PGN & 65.6 & 27.2 & 16.9 & 13.0 & 34.4 
\\
\quad + CWA & 67.4 & 39.6 & 23.7 & 19.5 & 40.8 
\\
\textbf{\quad + DRAP} & \textbf{83.2} & \textbf{56.6} & \textbf{27.7} & \textbf{20.6} & \textbf{51.2} 
\\ \midrule
TI-FGSM & 43.7 & 15.4 & 15.0 & 12.1 & 24.0 
\\
\quad + MI & 57.8 & 26.5 & 17.6 & 14.1 & 32.1 
\\
\quad + PI & 60.2 & 27.4 & 17.6 & 13.9 & 33.1 
\\
\quad + VT & 48.1 & 17.8 & 15.1 & 12.2 & 26.0 
\\
\quad + RAP & 46.1 & 21.7 & 15.6 & 12.2 & 26.4 
\\
\quad + PGN & 66.4 & 29.7 & 17.7 & 13.6 & 35.6 
\\
\quad + CWA & 60.6 & 26.8 & 22.5 & 19.3 & 35.3 
\\
\textbf{\quad + DRAP} & \textbf{79.2} & \textbf{44.4} & \textbf{25.6} & \textbf{20.2} & \textbf{46.5 }
\\ \midrule
Admix & 64.6 & 23.1 & 15.5 & 12.1 & 32.6 
\\
\quad + MI & 66.4 & 29.2 & 17.3 & 13.5 & 35.4 
\\
\quad + PI & 64.0 & 27.1 & 17.2 & 13.4 & 34.0 
\\
\quad + VT & 58.0 & 21.4 & 15.1 & 12.3 & 30.0 
\\
\quad + RAP & 56.7 & 26.2 & 16.4 & 12.6 & 31.1 
\\
\quad + PGN & 65.3 & 26.9 & 16.5 & 12.7 & 34.1 
\\
\quad + CWA & 64.0 & 25.7 & 22.0 & 19.1 & 36.0 
\\
\textbf{\quad + DRAP} & \textbf{82.4} & \textbf{43.2} & \textbf{24.5} & \textbf{19.7} & \textbf{46.8 }
\\ \midrule
SIA & 82.9 & 49.2 & 17.5 & 12.8 & 45.1 
\\
\quad + MI & 83.0 & 55.5 & 20.4 & 13.9 & 47.7 
\\
\quad + PI & 82.8 & 52.0 & 20.1 & 13.8 & 46.7 
\\
\quad + VT & 83.1 & 51.5 & 18.0 & 12.8 & 45.9 
\\
\quad + RAP & 72.5 & 45.5 & 19.1 & 13.0 & 41.5 
\\
\quad + PGN & - & - & - & - & - 
\\
\quad + CWA & 83.4 & 52.0 & 23.8 & 17.9 & 48.7 
\\
\textbf{\quad + DRAP} & \textbf{86.9} & \textbf{55.3} & \textbf{25.6} & \textbf{18.2} & \textbf{51.1} \\ \bottomrule
\end{tabular}%
}
\vspace{-10pt}
\end{table}

\subsection{Ablative Study}
From the results above, we conclude that DRAP inspired from the theoretical bound could learn an AE with strong transferability toward target models. To answer question 3), we need to gain a deeper insight into the rationale behind its superior attack performance. 
In this subsection, we disentangle the two distinct optimization signals within the bound: a sharpness penalty for pursuing a flat local minimum $\ell_{sharp}=\max _{\|\bm{\epsilon}\|_\infty \leq \rho}R_{\hat{\mathcal{S}}}(\hat{\bm{x}}+\bm{\epsilon})-R_{\hat{\mathcal{S}}}(\hat{\bm{x}})$ and a model discrepancy penalty for narrowing the surrogate-target shift $\ell_{dis}=\mathrm{D}_\phi^{\hat{\mathcal{X}}_r}\left({P_\mathcal{T}} \| P_{\mathcal{S}}\right)$.
We conduct ablative studies to explore the impact of each aspect of DRAP: first to determine the importance of the sharpness penalty term in our bound and the effectiveness of \bluetwo{the} proposed model-diversity-compatible optimization algorithm (cf. Algorithm \ref{mainalgo}), second to determine the importance of the model discrepancy penalty term in our bound and the effectiveness of the strategy to choose surrogate models. We evaluate these ablations on ImageNet using the same untargeted experimental protocol as in Section \ref{Sec: Main Results}.

\subsubsection{On the Sharpness Penalty}\label{alb_sec: sharp}

\begin{table}[t]
\centering
\caption{Ablating the sharpness penalty term ( - $\ell_{sharp}$) and late start strategy ( - late start) from DRAP's combinations with I-FGSM, TI-FGSM, DI2-FGSM and Admix. The results are averaged on each model \bluetwo{set}. \textbf{Bold} denotes the best results and \underline{underlined} denotes the second best results.}
\label{table:alb_sharp}
\resizebox{\columnwidth}{!}{%
\begin{tabular}{@{}l|cccc|c@{}}
\toprule
\multicolumn{1}{c|}{DRAP} & \textbf{\begin{tabular}[c]{@{}c@{}}ConvNet\\ Set\end{tabular}} & \textbf{\begin{tabular}[c]{@{}c@{}}Metaformer\\ Set\end{tabular}} & \textbf{\begin{tabular}[c]{@{}c@{}}ConvNet\\ (AT) Set\end{tabular}} & \textbf{\begin{tabular}[c]{@{}c@{}}Metaformer\\ (AT) Set\end{tabular}} & \textit{\textbf{\begin{tabular}[c]{@{}c@{}}Overall\\ Average\end{tabular}}} \\ \midrule
+ I-FGSM & \textbf{80.2} & \textbf{42.6} & \textbf{24.9} & {\ul 20.2} & \textbf{46.1} \\
\quad - late start & {\ul 78.8} & {\ul 42.2} & \textbf{24.9} & \textbf{20.3} & {\ul 45.6} \\
\quad - $\ell_{sharp}$ & 77.7 & 36.3 & 24.1 & 20.0 & 43.6 \\ \midrule
+ TI-FGSM & \textbf{79.2} & \textbf{44.4} & {\ul 25.6} & \textbf{20.2} & \textbf{46.5} \\
\quad - late start & {\ul 78.1} & {\ul 43.2} & \textbf{25.8} & \textbf{20.2} & {\ul 45.9} \\
\quad - $\ell_{sharp}$ & 77.8 & 38.5 & 24.8 & {\ul 20} & 44.3 \\ \midrule
+ DI2-FGSM & \textbf{83.2} & \textbf{56.6} & {\ul 27.7} & {\ul 20.6} & \textbf{51.2} \\
\quad - late start & 81.3 & {\ul 55.0} & \textbf{28.0} & \textbf{20.7} & {\ul 50.3} \\
\quad - $\ell_{sharp}$ & {\ul 82.7} & 53.2 & 27.1 & 20.5 & 50.1 \\ \midrule
+ Admix & \textbf{82.4} & \textbf{43.2} & \textbf{24.5} & \textbf{19.7} & \textbf{46.8} \\
\quad - late start & 81.4 & {\ul 40.5} & \textbf{24.5} & {\ul 19.6} & {\ul 45.8} \\
\quad - $\ell_{sharp}$ & {\ul 81.7} & 40.1 & {\ul 24.0} & \textbf{19.7} & 45.7 \\ \bottomrule
\end{tabular}%
}
\vspace{-10pt}
\end{table}

\textbf{Q1-1: Is the flatness beneficial for boosting the transferability?} 
First we analyze the importance of optimizing flatness in boosting transferability. We formalize this study as ablating the sharpness penalty term $\ell_{sharp}$ from our optimization objective and evaluating the ablated objective by reporting the attack success rate of I-FGSM, DI2-FGSM, TI-FGSM and Admix combined with DRAP. We use the same untargeted experimental protocol as in Section \ref{Sec: Main Results} but with the ablated objective. The results are presented in Table \ref{table:alb_sharp}. Within each combination, the first row represents our \bluetwo{combined} method with the full optimization objective, applying the sharpness penalty starting at iteration $n_{LS}$. The second row represents the same objective but applies the sharpness penalty from iteration 0, implying an ablation on the late-start strategy. The third row represents the \bluetwo{combined} method without penalizing the sharpness of AE. Across all combinations, DRAP with complete objective consistently outperforms attacks that solely penalize model discrepancy, regardless of whether the late-start strategy is used.
Furthermore, we see a stronger attack performance of DRAP with late start strategy, which helps stabilize convergence.
The results validate our theoretical result from Theorem \ref{theorem_main} that controlling the surrogate-target shift, finding a flat \bluetwo{minimum} of surrogate adversarial risk \bluetwo{leads} to improved transferability.
A detailed parameter study on the impact of late start iteration number $n_{LS}$ is provided in \textit{Appendix} \ref{app:addresult_para}.

\textbf{Q1-2: How to effectively optimize flatness of AE across a diverse set of surrogate models?}
Given a set of diverse surrogate models obtained following Section \ref{sec: Narrow the surrogate-target discrepancy}, aside from our proposed algorithm (cf. Algorithm \ref{mainalgo}), there are other possible solvers to optimize flatness across this set. Here, we consider two implementations inspired by RAP and CWA, both of which aim to generate adversarial examples within flat local regions. We use Flat-RAP and Flat-CWA as the shorthands for optimizing flatness across diverse surrogate models using strategies of RAP and CWA, respectively.

To elaborate, our proposed algorithm boosts the flatness via a min-max bi-level optimization framework. It finds the worst-case reverse perturbation specific to each surrogate model at the inner step (refer to Equation \ref{equ_perturb}) and updates AE toward the point where added with the model-specific perturbation could minimize the attack loss on the one surrogate model at the outer step (refer to Equation \ref{eq: update_ae}).
This algorithm is expected to seek out AEs whose entire neighborhoods have uniformly low empirical surrogate adversarial risk value, \ie, AEs \bluetwo{located} in the flat regions of each of the diverse models.
However, from the perspective of RAP, Flat-RAP applies the inner maximization and outer minimization on the whole surrogate model set, \ie, a global reverse perturbation (refer to Equation \ref{equ_rap}) and global update direction (refer to Equation \ref{equ_rap_update}) are obtained. 
Furthermore, \bluetwo{through} the lens of CWA, Flat-CWA substitutes the outer step of Flat-RAP with successively performing updates using each surrogate model to pursue a common weakness of model ensemble:
\begin{equation}\label{eq: cwaupdate_ae}
\hat{\bm{x}} \leftarrow \Pi_{\gamma}\left[\hat{\bm{x}}-\beta_{\hat{\bm{x}}} \cdot \operatorname{sign}\left(\nabla_{\hat{\bm{x}}} \ell\left(f\left(\hat{\bm{x}}+\bm{\epsilon}, \w_k\right), y\right)\right)\right],\w_k \in \mathcal{M}_\mathcal{S},
\end{equation}
where $\bm{\epsilon}$ is a global reverse perturbation same as in Flat-RAP. 
We provide their pseudocodes and implementation details in \textit{Appendix} \ref{app: flat}. Note that though our derivation of Flat-RAP and Flat-CWA define the objective over the entire model set, in practice, we compute the gradient per-batch.

\begin{table}[t]
\centering
\caption{{Attack success rates (\%, $\uparrow$) of Flat-RAP, Flat-CWA and DRAP.} The results are averaged on each model \bluetwo{set}. \textbf{Bold} denotes the best results and \underline{underlined} denotes the second best results.}
\label{tab:flat}
\resizebox{0.96\columnwidth}{!}{%
\begin{tabular}{@{}c|cccc|c@{}}
\toprule
 \textbf{Attack} & \textbf{\begin{tabular}[c]{@{}c@{}}ConvNet\\ Set\end{tabular}} & \textbf{\begin{tabular}[c]{@{}c@{}}Metaformer\\ Set\end{tabular}} & \textbf{\begin{tabular}[c]{@{}c@{}}ConvNet\\(AT) Set\end{tabular}} & \textbf{\begin{tabular}[c]{@{}c@{}}Metaformer\\ (AT) Set\end{tabular}} & \textit{\textbf{\begin{tabular}[c]{@{}c@{}}Overall\\ Average\end{tabular}}} \\ \midrule
Flat-RAP & 77.4 & \textbf{47.0} & 19.9 & 14.3 & 43.9 \\
Flat-CWA & {\ul 80.2} & 39.0 & {\ul 23.9} & {\ul 19.8} & {\ul 44.9} \\
DRAP & \textbf{80.3} & {\ul 42.6} & \textbf{24.9} & \textbf{20.2} & \textbf{46.2} \\ \bottomrule
\end{tabular}%
}
\vspace{-12pt}
\end{table}

As shown in Table \ref{tab:flat}, we find that our method achieves higher transferability than Flat-RAP and Flat-CWA. It supports our hypothesis that for a set of surrogate models with distinct loss landscapes, a globally calculated $\bm{\epsilon}$ could not help to find a flat local minimum of surrogate adversarial risk, as it fails to orient $\hat{\bm{x}}$ to the real worst-case local neighborhood of each model. Consequently, the crafted $\hat{\bm{x}}$ will \bluetwo{be located in} sharp regions of the target model's landscape, slight changes in the loss landscape will cause a significant increase in attack loss.

\subsubsection{On the Model Discrepancy Penalty}\label{alb_sec: dis}

\textbf{Q2-1: Can within-distribution diversity boost transferability, and if so, how efficiently?} 
To evaluate the impact of incorporating within-distribution diversity when optimizing flatness on narrowing the surrogate-target model discrepancy and controlling the transferability gap, we conduct an ablation study. Specifically, we vary the parameter $n$, which directly controls the number of diverse surrogate models sampled from each surrogate model component during generating AEs, and thus determines the level of within-distribution diversity in the surrogate model set. We range $n$ from 0 to 40 with a granularity of 5, keeping all other hyper-parameters consistent with Section \ref{Sec: Main Results}. When $n=5$, $n_{LS}$ is set to 0.
Larger values of $n$ correspond to stronger within-distribution diversity, whereas smaller values gradually diminish this influence. At $n=0$, within-distribution diversity vanishes entirely, and our method reduces to locating an AE that resides in the flat regions of the loss landscapes of the five pretrained surrogate models.
As shown in Figure \ref{fig:iter}(a), on the whole target model sets, increasing within-distribution diversity consistently improves attack performance over the baseline case ($n=0$), with peak performance observed at $n=40$. This result convincingly validates the idea that \textit{encouraging AEs to locate within flat regions of the loss landscapes across diverse models increases the likelihood of their generalization to flat regions in unseen models. Consequently, slight changes in the loss landscape are less likely to cause a significant increase in attack loss, thereby improving the transferability of the attack}.

\begin{figure*}[t]
 \setlength{\abovecaptionskip}{-0.1cm}
\setlength{\belowcaptionskip}{-0.1cm}
 \centering
\includegraphics[width=0.9\linewidth]{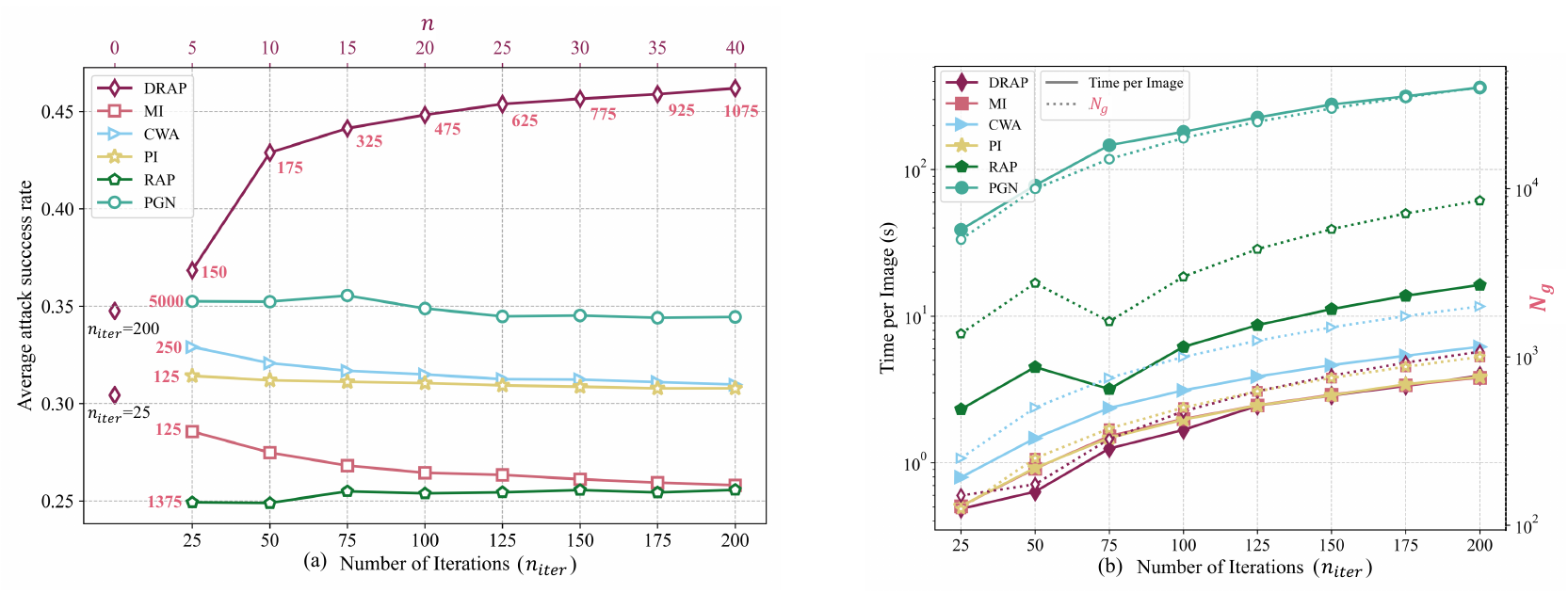}
 \caption{
 \textbf{(a)} The average attack success rate of different attacks \wrt~different \bluetwo{numbers} of iterations. The two isolated points correspond to DRAP without within-distribution diversity, \ie, $n=0$, under two different iteration numbers. The \textcolor[HTML]{D34563}{pink} numbers represent \blue{$N_g$}, the number of gradient calculations required by each method under specific iteration numbers.
 \blue{\textbf{(b)} The empirical wall-clock time (solid curves, left $y$-axis) and the theoretical number of gradient computations $N_g$ (dotted curves, right $y$-axis) required by different attacks to generate one AE \wrt~different \bluetwo{numbers} of iterations. Note that the relative ranking of different attacks in wall-clock time closely matches that implied by $N_g$, demonstrating that $N_g$ can appropriately reflect computational cost. }}\label{fig:iter}
 \vspace{-10pt}
\end{figure*}

\blue{However, to find a reverse perturbation, DRAP takes $T$ additional forward–backward gradient calculations in each iteration on top of the one required to update the AE itself. As a higher within-distribution diversity requires more iterations, one may wonder \bluetwo{about} the computational efficiency of DRAP. 
To conduct a comparison on computational efficiency with other methods, we evaluate our method as well as others under iterations from 25 to 200, and report both the required number of gradient calculations $N_g$ and the achieved attack performance in Figure \ref{fig:iter}(a). 
Here $N_g$ is formally defined as the total number of forward–backward gradient calculation operations, \ie,  $\nabla_{\hat{\bm{x}}} {\ell}\left(f\left(\hat{\bm{x}}, \w\right), y\right)$, which is the most computationally expensive operation in generating AEs and thus theoretically reflects the computational cost of an attack.
Clearer expressions of $N_g$ as a function of iteration numbers $n_{iter}$ for different attacks are summarized in Table \ref{tab:cost} in \textit{Appendix}.
We observe that 25 iterations on DRAP \bluetwo{are} enough to provide a significant performance gain over others, while the required $N_g$ is only slightly higher than MI-FGSM and PI-FGSM but much lower than others. Moreover, when more computational resources are available to run additional iterations, DRAP still outperforms compared attacks with its $N_g$ curve in Figure \ref{fig:iter}(b) remaining relatively low.
The results indicate that equipping surrogate models with within-distribution diversity enables our strategy of encouraging AEs toward flat regions to efficiently improve the transferability, thus enabling a practical trade-off between efficiency and attack performance.}

\blue{To further empirically validate the computational efficiency of DRAP, we also show the runtime per AE and the computational memory of different attacks in Figure \ref{fig:iter}(b) and Table \ref{table:memory}, respectively. In particular, DRAP achieves a competitive runtime, which is in alignment with the theoretical estimation $N_g$. Meanwhile, DRAP requires the lowest peak allocated GPU memory compared with others, as it only backpropagates through one surrogate model at each update rather than the whole surrogate set as designed in Section \ref{sec:Find a Flat Optimum from a Diverse Set of Surrogates}, demonstrating that the within-distribution diversity can effectively and efficiently boost attack performance. 
We note that DRAP requires storing $n$ SGD proposals for each surrogate component when constructing the surrogate set, which incurs an $n$-fold overhead in model storage compared with other attacks. 
The detailed experimental setup for runtime and memory measurements is provided in \textit{Appendix} \ref{app: runtime_setting}.}
 
\begin{table}[t]
\centering
\caption{\blue{The peak allocated GPU memory (mean $\pm$ std over different $n_{iter}$, as the memory usage is nearly invariant to iterations) of different attacks.}}
\label{table:memory}
\resizebox{9cm}{!}{%
\bluetable
\begin{tabular}{@{}l|cccccc@{}}
\toprule
& MI            & PI            & PGN           & CWA           & RAP           & DRAP          \\ \midrule
\textbf{Memory} & 17.47 $\pm$ 0.019 & 17.69 $\pm$ 0.024 & 21.67 $\pm$ 0.011 & 17.87 $\pm$ 0.019 & 17.86 $\pm$ 0.015 & 10.19 $\pm$ 0.009 \\ \bottomrule
\end{tabular}
}
 \vspace{-18pt}
\end{table}

One can observe from Figure \ref{fig:iter}(a) that the attack performance of MI-FGSM, PI-FGSM, and CWA---methods that accumulate gradients at each iteration---will not benefit from more iterations. Careful tuning of the iteration number is crucial for them, as excessive iterations can lead to gradient overaccumulation, which negatively impacts transferability \cite{blackboxbench, svre, eval}. In contrast, momentum better synergizes with our method, stabilizing the update directions that vary significantly across diverse surrogate models. Therefore, more iterations (\ie, more diverse surrogate models), better transferability.

\begin{figure*}[t]
 \setlength{\abovecaptionskip}{-0.1cm}
\setlength{\belowcaptionskip}{-0.1cm}
 \centering
\includegraphics[width=0.97\linewidth]{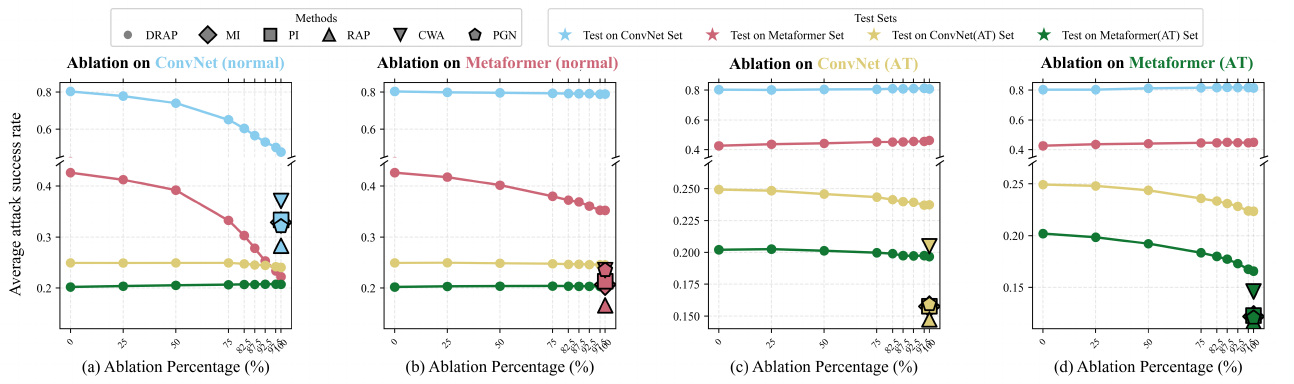}
 \caption{\blue{Attack success rates under prototype ablation. (a)–(d) correspond to ablating normally trained convnet, normally trained metaformer, adversarially trained convnet, and adversarially trained metaformer prototypes from the surrogate model set, respectively. As indicated in the legend, the four colored curves/markers denote attack success rates on the four target sets for each attack. For DRAP, we gradually ablate surrogate models from the selected prototypical set under the percentage on the $x$-axis while fixing the other prototypes, and report attack success rates on the four target sets. For the five baseline attacks, we remove the surrogate models of the selected prototype and report attack success rates on the corresponding target set. Other experimental protocol is the same as the untargeted setting in Section \ref{Sec: Main Results}.}}\label{fig:imbalance}
 \vspace{-10pt}
\end{figure*}

\textbf{Q2-2: Can between-distribution diversity boost transferability, and if so, how efficiently?} 
\blue{We have shown that the high within-distribution diversity could effectively and efficiently improve transferability. Here we further validate the necessity of another dimension of model diversity, between-distribution diversity. 
Specifically, for each prototypical model set in turn, we progressively ablate its surrogate models while fixing the other prototypical sets, and generate AEs with DRAP at different ablation ratios, thereby gradually weakening the between-distribution diversity. We then evaluate these AEs on the four target model sets, with particular focus on the set corresponding to the ablated prototype.
As shown in Figure \ref{fig:imbalance}, for each ablated prototype, the attack success rate on the corresponding target set exhibits a clear decreasing trend as that prototype is gradually removed from surrogate models. When one prototype is completely removed (100\% ablation), the transferability to the corresponding target set is evaluated under a possibly severe surrogate-target shift, and Figure \ref{fig:imbalance} shows that the average attack success rates on the four shifted target sets drop to 47.8\%, 35.2\%, 23.7\% and 16.6\%, which are 32.5\%, 7.4\%, 1.2\% and 3.6\% lower than those obtained when the surrogate distributional components are diverse enough. This demonstrates that the between-distribution diversity achieved by the four proposed prototypes indeed plays an important role in generating more transferable AEs.
Meanwhile, Figure~\ref{fig:imbalance} also shows that DRAP remains robust when the number of surrogate models from one prototype is substantially reduced. When the number of surrogate models from one prototype is ablated to 50\%, the attack performance on both the corresponding target set and the other three target sets is not significantly affected. Even under a severe 75\% ablation, DRAP still performs much better than the extreme case where this prototypical model set is entirely absent. These observations further confirm between-distribution diversity is crucial for improving transferability.}

\blue{However, given the endless evolution of model architectures, achieving the ideal between-distribution diversity with all prototypes being fully covered remains challenging. This raises concerns about whether DRAP can maintain its superiority when AEs are expected to transfer to unseen prototypes with unexpected adversarial vulnerabilities in practice. To investigate this, Figure \ref{fig:imbalance} also reports, as isolated markers at 100\% ablation, the attack success rates of five baseline attacks under the same surrogate-target shift setting. Across all four types of prototype ablations, DRAP consistently outperforms others on the target set corresponding to the ablated prototype, indicating that choosing diverse surrogate \bluetwo{component} distributions, albeit imperfect, is still efficient enough for DRAP to simulate possible surrogate-target shift in loss landscape. Consequently, seeking flat minima over these distributions significantly narrows the transferability gap.}

To analyze the effect of each dimension of model diversity (\ie, within-distribution and between\bluetwo{-}distribution), we ask Q2-1 and Q2-2. From the results, we find that both dimensions are useful for enhancing the transferability of AEs, as they both narrow the surrogate-target model discrepancy term $\ell_{dis}$ from different levels. Hence, by leveraging them simultaneously, our method significantly \bluetwo{decreases} the transferability gap.

\blue{\textbf{Q2-3: How to build a surrogate set for DRAP in practice?} Here we provide a simple decision checklist for practitioners when constructing a surrogate set for DRAP under limited model availability by summarizing the above analyses and supplementing some new findings.}

\begin{itemize}[leftmargin=15pt,itemsep=1pt,topsep=0pt]
\item \blue{\emph{For each surrogate component, use as many models as the budget allows.} Increasing the number of models within one surrogate component $n$ enriches the within-distribution diversity, thus improving transferability, as shown in Q2-1.}

\item \blue{\emph{Keep as many prototypes as possible, even if some prototypes contain far fewer surrogate models than others.} Our prototype ablation in Figure \ref{fig:imbalance} shows that removing any prototype leads to a clear drop in attack performance on the corresponding target set, and even a small number of surrogate models from a prototype remains beneficial. Thus, when adding or pruning surrogates, maintaining coverage over distinct prototypes to ensure the between-distribution diversity is the key point.}

\item \blue{\emph{\bluetwo{Prioritize} training strategy over architecture family.} Figure \ref{fig:imbalance} shows that, within either normal or adversarial training prototypes, ablating one model architecture family usually degrades AE's transferability to the other. However, ablating normally trained models has little impact on AE's transferability to adversarially trained models and vice versa. This indicates that different training strategies induce more distinct types of adversarial vulnerabilities than different model architecture families. Therefore, if practitioners can only afford a few surrogate components, first ensure to mix normally trained and adversarially trained models before mixing convnets and metaformers.}

\item \blue{\emph{Within each prototype, the exact architecture is secondary, but an exception is the hybrid architecture which can be very useful under limited computational resources.} As discussed in Appendix \ref{app:addresult_conv}, the specific surrogate model architecture chosen to represent each prototype is less significant than the between-distribution diversity among prototypes in our method. 
An interesting exception is the hybrid architecture, such as ConvNeXt, which is a special convnet that adopts designs popularized by vision transformers. \bluetwo{Using} a single ConvNeXt-T as the normally trained prototype achieves performance comparable to \bluetwo{that of} using separate convnet and metaformer prototypes, indicating that such hybrid architectures can work as two-in-one surrogates covering both the convnet and metaformer prototypes in resource-limited scenarios.
}
\end{itemize}

\section{Conclusion}
In this paper, we first prove a bound that provides a \bluetwo{guarantee} on transferability error. We show that our bound builds a framework that generalizes previous approaches and presents a fresh avenue for the principled analysis of transfer-based attacks. Within our transferability bound, we justify the relationship between flatness and AE transferability and point out the adversarial model discrepancy as another key component in bounding transferability. We gain theoretical insights from the derived bound and make algorithmic extensions to our prior work, RAP. The proposed DRAP generates reverse adversarial perturbations tailored to each of the diverse surrogate models, which are selected based on two dimensions of diversity. We conduct extensive experiments on two datasets, covering untargeted and targeted attacks against standard and defense models. We also conduct ablative studies to explore the effect of different penalty terms to verify our theoretical findings.

\section*{\bluetwo{Acknowledgments}}
\bluetwo{Baoyuan Wu is supported by Guangdong Basic and Applied Basic Research Foundation (No. 2024B1515020095), Guangdong Provincial Program (No. 2023TQ07A352), Sub-topic of Key R\&D Projects of the Ministry of Science and Technology (No. 2023YFC3304804), Shenzhen Science and Technology Program (No. RCYX20210609103057050 and JCYJ20240813113608011), and Longgang District Key Laboratory of Intelligent Digital Economy Security. Rui Huang is supported by Shenzhen Science and Technology Program (No. JCYJ20220818103006012 and KJZD20240903100202004). Yanbo Fan is supported by New Generation Artificial Intelligence-National Science and Technology Major Project (No. 2025ZD0123502).}

\normalem
\bibliographystyle{IEEEtran}
\bibliography{Ref}


\section{Biography Section}
 \vspace{-25pt}
\begin{IEEEbiography}[{\vspace{-20pt}\includegraphics[width=0.7in,clip,keepaspectratio]{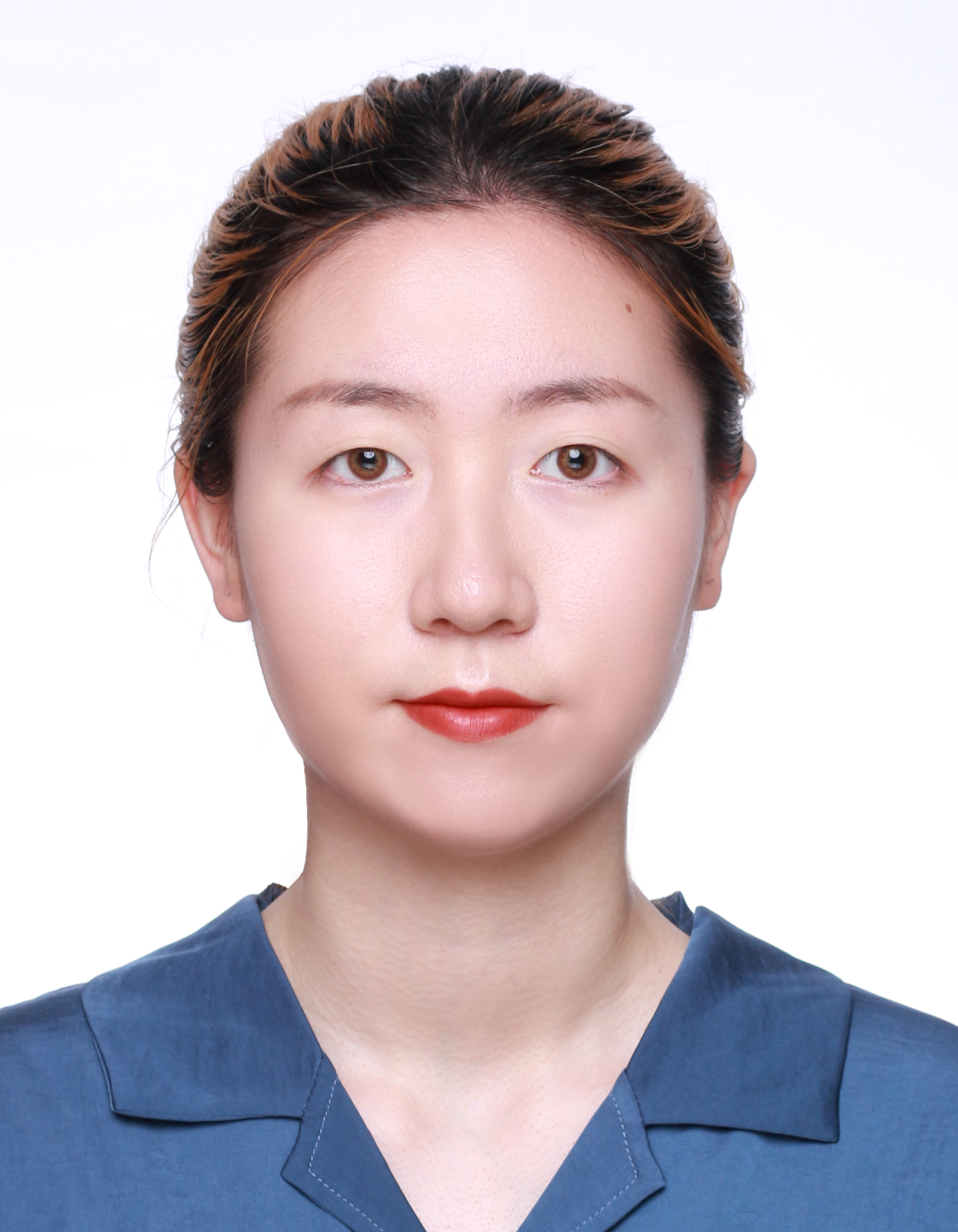}}]{Meixi Zheng}
received the bachelor’s degree and the master's degree from the Xidian University in 2019 and 2022. She is currently working toward the PhD degree with Computer Science Program in The Chinese University of Hong Kong, Shenzhen, supervised by Prof. Baoyuan Wu. Her research interest includes adversarial machine learning.
\end{IEEEbiography}
\vspace{-41pt}
\begin{IEEEbiography}[{\vspace{-20pt}\includegraphics[width=0.7in,clip,keepaspectratio]{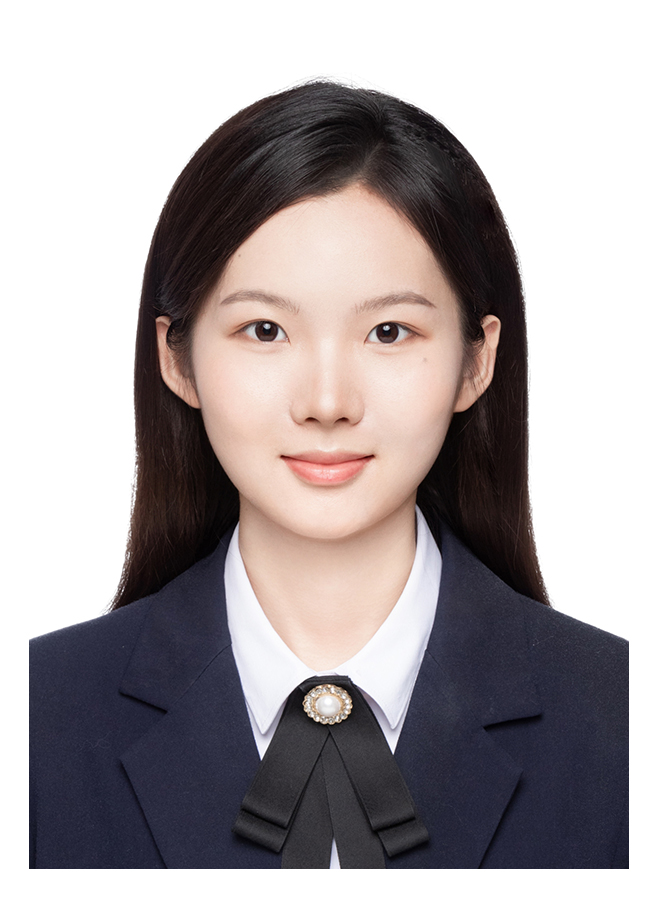}}]{Kehan Wu}
received the bachelor’s degree from Southwest Jiaotong University in 2023 and the Mphil degree from The Chinese University of Hong Kong in 2025 under the supervision of both Prof. Rui Huang and Prof. Baoyuan Wu. Her Mphil research focused on transfer-based black-box adversarial attacks.
\end{IEEEbiography}
\vspace{-41pt}
\begin{IEEEbiography}[{\vspace{-20pt}\includegraphics[width=0.7in,clip,keepaspectratio]{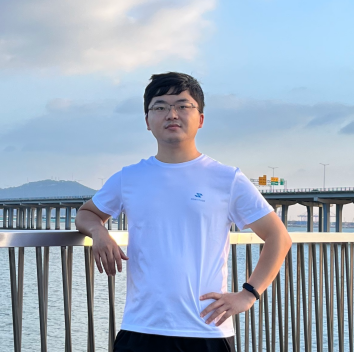}}]{Yanbo Fan}
is currently an Associate Professor of Nanjing University (Suzhou Campus), and he worked as  Senior Research Scientist at Tencent AI Lab and Ant Group, respectively before joining Nanjing University. He received his Ph.D. degree from Institute of Automation, Chinese Academy of Sciences, Beijing, China, in 2018, and his B.S. degree from Hunan University in 2013. His research interests are generative and trustworthy AI.
\end{IEEEbiography}
\vspace{-41pt}
\begin{IEEEbiography}[{\includegraphics[width=0.7in,clip,keepaspectratio]{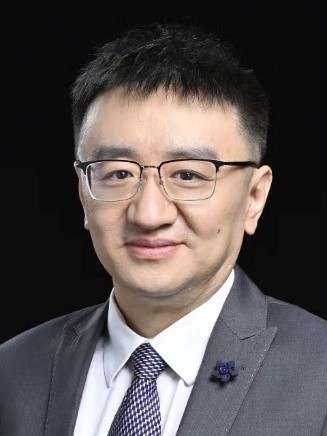}}]{Rui Huang}
(Member, IEEE) received his B.Sc. degree from Peking University in 1999, his M.Eng. from the Chinese Academy of Sciences in 2002, and his Ph.D. from Rutgers University in 2008. In 2012 to 2016, he was a researcher at NEC Laboratories China. He is currently a Tenured Associate Professor at The Chinese University of Hong Kong, Shenzhen. His current research interests focus on video analytics, robotic perception and navigation and autonomous driving. 
\end{IEEEbiography}
\vspace{-31pt}
\begin{IEEEbiography}[{\vspace{-20pt}\includegraphics[width=0.7in,clip,keepaspectratio]{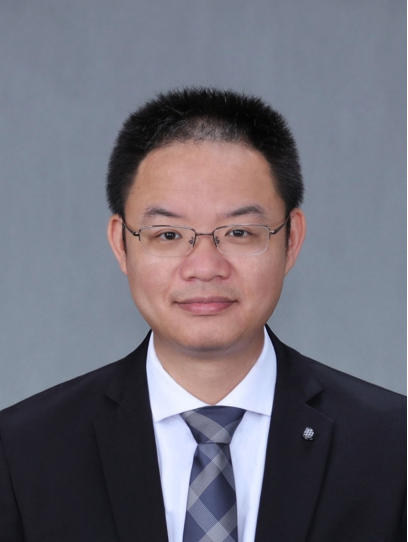}}]{Baoyuan Wu} Dr. Baoyuan Wu is a Tenured Associate Professor of School of Artificial Intelligence, The Chinese University of Hong Kong, Shenzhen, Guangdong, 518172, P.R. China. His research interests are trustworthy and generative AI. He is currently serving as Senior Area Editor of IEEE TIFS. He is IEEE Senior Member.
\end{IEEEbiography}

\clearpage

\setcounter{section}{0}
\setcounter{equation}{0}
\numberwithin{equation}{section}
\setcounter{page}{1}
\long\def\comment#1{}
\renewcommand\thesection{\Alph{section}}

{\appendices
\section{Proofs}\label{app: proof}

\subsection{Technical Lemmas}

\begin{lemma}\label{lemma_phi*}
Let $\phi: \mathbb{R}_{+} \rightarrow \mathbb{R}$ be a convex, lower semi-continuous function satisfying $\phi(1)=0$, and $\phi^*$ be the Fenchel conjugate function of $\phi$, then $\phi^*(\bm{x})\geq \bm{x}$.
\begin{proof}
By definition, $\phi^*(\bm{x})=\sup _{\bm{t} \in\text{dom} \phi}\left\{\bm{t}\bm{x}-\phi(\bm{t})\right\}\geq \bm{t}\bm{x}-\phi(\bm{t})$. When $\phi(1)=0$, we have $\phi^*(\bm{x})\geq 1 \cdot \bm{x}-\phi(1)=\bm{x}$.
\end{proof}
\end{lemma}

Introducing the PAC model into transfer-based adversarial attacks to bound the surrogate risk, we have:
\begin{lemma}\label{lemma_pac_trans}
    \textbf{(PAC-Bayes \cite{pac,intriguing} )} For any model distribution $P_\mathcal{S}$, prior distribution $\mathcal{P}$ on $\hat{\mathcal{X}}$, $0<\delta<1$, with probability $1-\delta$ over the choice of surrogate model set $\mathcal{M}_\mathcal{S}\sim P_\mathcal{S}$ with size $K \in \mathbb{N}$, for any distributions $\mathcal{Q}$ on $\hat{\mathcal{X}}$, the following bound holds:
\begin{equation}\label{equ_pac_trans}
    \begin{aligned}
    \mathbb{E}_{\hat{\bm{x}} \sim \mathcal{Q}}[R_\mathcal{S}(\hat{\bm{x}})] \leq \mathbb{E}_{\hat{\bm{x}} \sim \mathcal{Q}}[R_{\hat{\mathcal{S}}}(\hat{\bm{x}})] +\sqrt{\frac{\text{KL}(\mathcal{Q}||\mathcal{P})+\log\frac{K}{\delta}}{2(K-1)}}.
\end{aligned}
\end{equation}
\end{lemma}
In this bound, the $R_\mathcal{S}(\hat{\bm{x}})$ is the surrogate risk and $R_{\hat{\mathcal{S}}}(\hat{\bm{x}})$ is the empirical surrogate risk.
This PAC-Bayes bound implies \bluetwo{that, assuming the adversary has} enough surrogate model samples, the expected risk of an AE chosen from a distribution $\mathcal{Q}$ can be guaranteed by minimizing the measured loss of distribution $\mathcal{Q}$  and $\frac{KL(\mathcal{Q} \| \mathcal{P})}{n}$, naturally leading to the following optimization method:
\begin{itemize}
    \item[1.] Fix a distribution $\mathcal{P}$.
    \item[2.] Collect enough surrogate model instances from $P_\mathcal{S}$.
    \item[3.] Compute the optimal distribution $\mathcal{Q}$ that minimizes the error bound, the right hand side of Equation \ref{equ_pac_trans}.
    \item[4.] Return the crafted AE given by $\mathcal{Q}$.
\end{itemize}

{
\bluetable
In the below lemma, we use a simple case of the under-coverage event to demonstrate that KL-instantiated discrepancy term $D_{\text{KL}}^{\hat{\mathcal X}_r}(P_{\mathcal T}\|P_{\mathcal S})$ will penalize $P_{\mathcal S}$'s under-coverage of the vulnerability pattern to $P_{\mathcal T}$.

\begin{lemma}\label{lemma:KL-undercoverage}
Let $D_{\text{KL}}^{\hat{\mathcal X}_r}(P_{\mathcal T}\| P_{\mathcal S})=\sup_{\hat{\bm{x}} \in\hat{\mathcal X}_r,\,t\in\mathbb R} t\,\mathbb E_{\bm{w}\sim P_{\mathcal T}}[\ell(f(\hat {\bm{x}},\bm{w}),y)] -\log \mathbb E_{\bm{w}\sim P_{\mathcal S}}[e^{t\ell(f(\hat {\bm{x}},\bm{w}),y)}]$.
Assume there exist $\hat{\bm x}\in\hat{\mathcal X}_r$ and a subset $A\subseteq\mathcal W$ such that
$$
\ell(f(\hat{\bm x},\bm w),y)=
\begin{cases}
\ell_A, & \bm w\in A,\\
\ell_B, & \bm w\notin A,
\end{cases}
\qquad\text{with }\ell_A\neq \ell_B.
$$
Let $p:=P_{\mathcal T}(A)$ and $q:=P_{\mathcal S}(A)$. Then
\begin{align*}
D_{\text{KL}}^{\hat{\mathcal X}_r}(P_{\mathcal T}\|P_{\mathcal S})
&\geq \text{KL}\left(\text{Bern}(p)\ \middle\|\ \text{Bern}(q)\right).\\
\end{align*}
\end{lemma}

\begin{proof}
Fix the $\hat{\bm x}\in\hat{\mathcal X}_r$ in the assumption and $\ell(\bm w)$ abbreviates $\ell(f(\hat{\bm x},\bm w),y)$. By the definition of $D_{\text{KL}}^{\hat{\mathcal X}_r}$, we have
$$D_{\text{KL}}^{\hat{\mathcal X}_r}(P_{\mathcal T}\|P_{\mathcal S}) \geq \sup_{t\in\mathbb R} t\,\mathbb E_{P_{\mathcal T}}[\ell(\bm w)] -\log \mathbb E_{P_{\mathcal S}}[e^{t\ell(\bm w)}].$$
Under the decomposition of $\ell(\bm w)$ into $\bm w \in A$ and $\bm w \notin A$, we have
$$\mathbb E_{P_{\mathcal T}}[\ell(\bm w)]
= p\ell_A+(1-p)\ell_B,$$
$$\mathbb E_{P_{\mathcal S}}[e^{t\ell(\bm w)}]
= q e^{t\ell_A}+(1-q)e^{t\ell_B}.$$
Therefore, let $F(t):= t\,\mathbb E_{P_{\mathcal T}}[\ell(\bm w)]-\log \mathbb E_{P_{\mathcal S}}[e^{t\ell(\bm w)}]$, for any $t\in\mathbb R$,

\begin{align*}
F(t)
&= t(p\ell_A+(1-p)\ell_B)-\log(q e^{t\ell_A}+(1-q)e^{t\ell_B})\\
&= t(p\ell_A+(1-p)\ell_B)-\log(e^{t\ell_B}[(1-q)+q e^{t(\ell_A-\ell_B)}])\\
&= t\,p(\ell_A-\ell_B) - \log((1-q)+q e^{t(\ell_A-\ell_B)}).
\end{align*}
Let $s:=t\Delta$. Since we assume $\Delta:=\ell_A-\ell_B\neq 0$, as $t$ ranges over $\mathbb R$, so does $s$.
Thus,
$$\sup_{t\in\mathbb R} F(t)=\sup_{s\in\mathbb R}\Big(ps-\log\big((1-q)+q e^s\big)\Big).$$

Let $f(s):=ps-\log((1-q)+q e^s)$. Its derivative is
$$f'(s)=p-\frac{q e^s}{(1-q)+q e^s}.$$
Setting $f'(s)=0$ yields the maximizer
$$s^\star=\log\frac{p(1-q)}{q(1-p)}.$$
Plugging $s^\star$ back, we note that
$$(1-q)+q e^{s^\star}=(1-q)+q\cdot\frac{p(1-q)}{q(1-p)}=\frac{1-q}{1-p},$$
and hence
\begin{align*}
\sup_{s\in\mathbb R} f(s)
&= f(s^\star)
= p\log\frac{p(1-q)}{q(1-p)}-\log\frac{1-q}{1-p}\\
&= p\log\frac{p}{q}+(1-p)\log\frac{1-p}{1-q}\\
&= \text{KL}\left(\text{Bern}(p)\ \middle\|\ \text{Bern}(q)\right).
\end{align*}
Combining the above inequalities gives
$$D_{\text{KL}}^{\hat{\mathcal X}_r}(P_{\mathcal T}\|P_{\mathcal S})
\geq \sup_{t\in\mathbb R}F(t)
= \text{KL}\left(\text{Bern}(p)\ \middle\|\ \text{Bern}(q)\right),$$
which completes the proof.
\end{proof}

\begin{remark}
Intuitively, $A$ is a subset of models representing distinct adversarial vulnerability patterns such that $\ell(f(\hat{\bm x},\bm w),y)$ takes different characteristic levels on $A$ and its complement $A^c$, i.e., $\ell_A\neq \ell_B$.
Consider the case where $q<p$, which exactly corresponds to $P_{\mathcal T}$ placing non-negligible probability mass on a subset of models exhibiting different adversarial vulnerability patterns from what $P_{\mathcal S}$ focuses on, then $\text{KL}\left(\text{Bern}(p) \| \text{Bern}(q)\right)$ increases as $q$ \bluetwo{decreases}. This \bluetwo{demonstrates} that $D_{\text{KL}}^{\hat{\mathcal X}_r}$ can penalize $P_{\mathcal{S}}$'s under-coverage of target vulnerability patterns.
\end{remark}
}

\subsection{Proof of Theorem \ref{theorem1}}

In black-box adversarial attacks, it \bluetwo{is} tempting to establish the error bound for the target model \bluetwo{that includes} a discrepancy measure (in Definition  \ref{def_our_f}). 
Consider a loss function $\ell(y_1, y_2)$, such that $\ell: \mathcal{Y} \times \mathcal{Y} \rightarrow \mathbb{R}_0^{+}$. Then we can define a population risk by $R_{P}(\hat{\bm{x}}):=\mathbb{E}_{\w \sim {P}}[\ell(f(\hat{\bm{x}},\w),y)]$. 
Given two distributions $P_\mathcal{S}$ and $P_\mathcal{T}$, the following \bluetwo{lemma} shows that the difference of risks over $P_\mathcal{S}$ and $P_\mathcal{T}$ can be bounded by the adversarial model discrepancy between $P_\mathcal{S}$ and $P_\mathcal{T}$. The proof technique we used here is inspired from Wang \etal \cite{f2}.

\noindent \textbf{Theorem 2. (Transferability gap bound)} Define $K_{\mathcal{S}}^{\hat{\bm{x}}}\left(t\right) = \inf _\alpha \left\{\mathbb{E}_{\w \sim P_\mathcal{S}}\left[\phi^*\left(t \ell\left(f\left(\hat{\bm{x}},\w\right),y\right)+\alpha\right)\right]-\alpha\right\} -\mathbb{E}_{\w \sim P_\mathcal{S}}\left[t \ell\left(f\left(\hat{\bm{x}},\w\right),y\right)\right]$. Given the surrogate model distribution $P_\mathcal{S}$ and target model distribution $P_\mathcal{T}$, for any $\hat{\bm{x}} \in \hat{\mathcal{X}_r}$ and \bluetwo{constants} $c_1, c_2 \in [0, +\infty)$ \bluetwo{subject} to the constraint $K_{\mathcal{S}}^{\hat{\bm{x}}}(c_1)\leq c_1c_2\mathbb{E}_{\w \sim P_\mathcal{S}}\left[\ell\left(f(\hat{\bm{x}},\w),y\right)\right]$, we have
\begin{equation}\label{eq_them2_7}
    \mathcal{E}_{\text {trans}}\left(\hat{\bm{x}}\right) \leq  \frac{1}{c_1}\mathrm{D}_\phi^{\hat{\mathcal{X}}_r}\left({P_\mathcal{T}} \| {P_\mathcal{S}}\right)+c_2r.
\end{equation}

\begin{proof}
Firstly, $K_{\mathcal{S}}^{\hat{\bm{x}}}(t)$, which depends on both $t \in \mathbb{R}$ and $\hat{\bm{x}}\in \hat{\mathcal{X}}_r$, is defined as follows:
\begin{align*}
    K_{\mathcal{S}}^{\hat{\bm{x}}}(t) =& \inf _\alpha \left\{\mathbb{E}_{\w \sim P_\mathcal{S}}\left[\phi^*\left(t \ell\left(f(\hat{\bm{x}},\w),y\right)+\alpha\right)\right]-\alpha\right\} \\
    &-\mathbb{E}_{\w \sim P_\mathcal{S}}\left[t \ell\left(f(\hat{\bm{x}},\w),y\right)\right],
\end{align*}
and we define
\begin{align}
    K_\mathcal{S}^r(t)=\sup _{\hat{\bm{x}} \in \hat{\mathcal{X}}_r} K_{\mathcal{S}}^{\hat{\bm{x}}}(t).\label{eq_them2_1}
\end{align}

We denote for clarity 
\begin{align*}
    {I}_{\mathcal{S}}^{\hat{\bm{x}}}(t)=\inf _\alpha \left\{\mathbb{E}_{\w \sim P_\mathcal{S}}\left[\phi^*\left(t \ell\left(f(\hat{\bm{x}},\w),y\right)+\alpha\right)\right]-\alpha\right\},
\end{align*}
which is also the second term of $\mathrm{D}_\phi^{\hat{\mathcal{X}}_r}\left({P_\mathcal{T}} \| {P_\mathcal{S}}\right) $. Then $K_{\mathcal{S}}^{\hat{\bm{x}}}(t) = {I}_{\mathcal{S}}^{\hat{\bm{x}}}(t) -\mathbb{E}_{\w \sim P_\mathcal{S}}\left[t \ell\left(f(\hat{\bm{x}},\w),y\right)\right]$.

Therefore, for any $\hat{\bm{x}} \in \hat{\mathcal{X}_r}$ and $t \in \mathbb{R}$, we have the following inequality holds by Equation \ref{eq_them2_1}:
\begin{align}
    {I}_{\mathcal{S}}^{\hat{\bm{x}}}(t)-\mathbb{E}_{\w \sim P_\mathcal{S}}\left[t \ell\left(f(\hat{\bm{x}},\w),y\right)\right] \leq K_\mathcal{S}^r(t).\label{eq_them2_4}
\end{align}

Plugging in $\mathbb{E}_{\mathcal{T}}\left[t\ell\left(f(\hat{\bm{x}},\w),y\right)\right]$ to the both sides of Equation \ref{eq_them2_4} and rearranging terms leads to the following inequality:
\begin{align*}
&t\left(\mathbb{E}_{\w \sim P_\mathcal{T}}\left[\ell\left(f(\hat{\bm{x}},\w),y\right)\right]-\mathbb{E}_{\w \sim P_\mathcal{S}}\left[\ell\left(f(\hat{\bm{x}},\w),y\right)\right]\right)-K_\mathcal{S}^r(t) \\
&\leq t\mathbb{E}_{\w \sim P_\mathcal{T}}\left[\ell\left(f_{\hat{\bm{x}}}, f_{\hat{\bm{x}}}\right)\right]-{I}_{\mathcal{S}}^{\hat{\bm{x}}}(t).
\end{align*}

Since this inequality holds for any $\hat{\bm{x}} \in \hat{\mathcal{X}_r}$ and $t \in \mathbb{R}$, we have
\begin{align*}
&\sup_{t \in \mathbb{R}}t\left(\mathbb{E}_{\w \sim P_\mathcal{T}}\left[\ell\left(f(\hat{\bm{x}},\w),y\right)\right]-\mathbb{E}_{\w \sim P_\mathcal{S}}\left[\ell\left(f(\hat{\bm{x}},\w),y\right)\right]\right)-K_\mathcal{S}^r(t) \\
&\leq \sup_{t \in \mathbb{R}}t\mathbb{E}_{\w \sim P_\mathcal{T}}\left[\ell\left(f(\hat{\bm{x}},\w),y\right)\right]-{I}_{\mathcal{S}}^{\hat{\bm{x}}}(t)\\
&\leq \sup_{t \in \mathbb{R}, \hat{\bm{x}} \in \hat{\mathcal{X}_r}}t\mathbb{E}_{\w \sim P_\mathcal{T}}\left[\ell\left(f(\hat{\bm{x}},\w),y\right)\right]-{I}_{\mathcal{S}}^{\hat{\bm{x}}}(t)\\
&=\mathrm{D}_\phi^{\hat{\mathcal{X}}_r}\left({P_\mathcal{T}} \| {P_\mathcal{S}}\right).
\end{align*}

Notice that $\mathcal{E}_{\text {trans}}\left(\hat{\bm{x}}\right)=\mathbb{E}_{\w \sim P_\mathcal{T}}\left[\ell\left(f\left(\hat{\bm{x}},\w\right),y\right)\right]-\mathbb{E}_{\w \sim P_\mathcal{S}}\left[\ell\left(f\left(\hat{\bm{x}},\w\right),y\right)\right]$. Hence, we have
\begin{align}
    &\sup_{t \in \mathbb{R}}t\mathcal{E}_{\text {trans}}\left(\hat{\bm{x}}\right)-K_\mathcal{S}^r(t) \leq \mathrm{D}_\phi^{\hat{\mathcal{X}}_r}\left({P_\mathcal{T}} \| {P_\mathcal{S}}\right) \nonumber \\
     &t\mathcal{E}_{\text {trans}}\left(\hat{\bm{x}}\right) -K_\mathcal{S}^r(t)\leq \mathrm{D}_\phi^{\hat{\mathcal{X}}_r}\left({P_\mathcal{T}} \| {P_\mathcal{S}}\right) \quad \forall t \in \mathbb{R} \label{eq_them2_2}\\
     &\mathcal{E}_{\text {trans}}\left(\hat{\bm{x}}\right) \leq \frac{\mathrm{D}_\phi^{\hat{\mathcal{X}}_r}\left({P_\mathcal{T}} \| {P_\mathcal{S}}\right)+K_\mathcal{S}^r(t)}{t}\quad \forall t\geq 0 \label{eq_them2_3}\\
    &\mathcal{E}_{\text {trans}}\left(\hat{\bm{x}}\right) \leq  \inf_{t\geq 0} \frac{\mathrm{D}_\phi^{\hat{\mathcal{X}}_r}\left({P_\mathcal{T}} \| {P_\mathcal{S}}\right)+K_\mathcal{S}^r(t)}{t} \label{eq_them2_6},
\end{align}
the derivation from Equation \ref{eq_them2_2} to Equation \ref{eq_them2_3} restricts the range of $t$ to $t\geq 0$.

Since we assume that there exists \bluetwo{constants} $c_1, c_2 \in [0, +\infty)$ such that 
\begin{align}
    K_{\mathcal{S}}^{\hat{\bm{x}}}(c_1)\leq c_1c_2\mathbb{E}_{\w \sim P_\mathcal{S}}\left[\ell\left(f(\hat{\bm{x}},\w),y\right)\right], \label{eq_them2_5}
\end{align}
we have
\begin{align*}
    K_\mathcal{S}^r(c_1)&=\sup _{\hat{\bm{x}} \in \hat{\mathcal{X}}_r} K_{\mathcal{S}}^{\hat{\bm{x}}}(c_1)\\
    &\leq\sup _{\hat{\bm{x}} \in \hat{\mathcal{X}}_r} c_1c_2\mathbb{E}_{\w \sim P_\mathcal{S}}\left[\ell\left(f(\hat{\bm{x}},\w),y\right)\right]\\
    & \leq c_1c_2r,
\end{align*}
where the last inequality holds by the definition of $\hat{\mathcal{X}}_r$.

Substituting the above inequality into Equation \ref{eq_them2_6} and replacing $t$ with $c_1$, we have the following inequality holds for any $c_1, c_2$ subject to the constraint  \ref{eq_them2_5}:
\begin{align*}
    \mathcal{E}_{\text {trans}}\left(\hat{\bm{x}}\right) \leq  \frac{\mathrm{D}_\phi^{\hat{\mathcal{X}}_r}\left({P_\mathcal{T}} \| {P_\mathcal{S}}\right)}{c_1}+c_2r,
\end{align*}
which completes the proof of Equation \ref{eq_them2_7}.

\end{proof}

\subsection{Instantiation of Theorem \ref{theorem1} with Different $\phi$-\bluetwo{divergences}}\label{app:inst}

\blue{Our Theorem \ref{theorem1} provides a more general framework of transferability gap bound that encompasses the family of $\phi$-divergences. Below we provide some special cases under TV distance, KL divergence, and $\chi^2$ divergence.}

\noindent \textbf{Corollary 1. }\textbf{(TV Instantiation of Theorem \ref{theorem1})} Suppose $\ell: \mathcal{Y} \times \mathcal{Y} \rightarrow [0,1]$. Given the surrogate model distribution $P_\mathcal{S}$ and target model distribution $P_\mathcal{T}$, for any $\hat{\bm{x}} \in \hat{\mathcal{X}_r}$ and constant $c_1$ satisfying $0\leq c_1\leq 1$, we have
\begin{align*}
        \mathcal{E}_{\text{trans}}\left(\hat{\bm{x}}\right) \leq \frac{1}{c_1}\mathrm{D}_\text{TV}^{\hat{\mathcal{X}}_r}(P_\mathcal{T} \| P_\mathcal{S}) ,
\end{align*}
where 
\begin{align}
\mathrm{D}_\text{TV}^{\hat{\mathcal{X}}_r}(P_\mathcal{T} \| P_\mathcal{S})&=\sup_{\hat{\bm{x}}^\prime \in \hat{\mathcal{X}}_r} |\mathbb{E}_{\w \sim P_\mathcal{T}}[\ell(f({\hat{\bm{x}}^\prime,\w}),y)] \nonumber\\
&-\mathbb{E}_{\w \sim P_\mathcal{S}}[\ell(f({\hat{\bm{x}}^\prime,\w}),y)]|. \label{eq_tv_def}
\end{align}

\begin{proof}
Here, we instantiate the bound in Theorem \ref{theorem1} with TV distance. Specifically, let $\phi_\text{TV}({u})=\mid u-1\mid$, its convex conjugate function is $\phi^*_\text{TV}(v)=v$, where $v$ takes value in $ \left[-1,1\right]$. 
\begin{align}
&\mathrm{D}_\text{TV}^{\hat{\mathcal{X}}_r}(P_\mathcal{T} \| P_\mathcal{S}) \nonumber\\
=&\sup_{\hat{\bm{x}} \in \hat{\mathcal{X}}_r, -1\leq t \leq 1}  \mathbb{E}_{\w \sim P_\mathcal{T}}\left[t\ell\left(f\left({\hat{\bm{x}},\w}\right),y\right)\right]- \nonumber\\
&\inf_{\alpha \in  \mathbb{R}}\left\{\mathbb{E}_{\w \sim P_\mathcal{S}}\left[\phi_\text{TV}^*\left(t\ell\left(f\left({\hat{\bm{x}},\w}\right),y\right)+\alpha\right)\right]-\alpha\right\} \nonumber \\
\geq  & t\left( \mathbb{E}_{\w \sim P_\mathcal{T}}\left[\ell\left(f\left({\hat{\bm{x}},\w}\right),y\right)\right]- \mathbb{E}_{\w \sim P_\mathcal{S}}\left[\ell\left(f\left({\hat{\bm{x}},\w}\right),y\right)\right]\right) \label{eq_lemma_tv_1}\\
= & t\mathcal{E}_{\text {trans}}\left(\hat{\bm{x}}\right).
\end{align}
The above inequality \ref{eq_lemma_tv_1} holds for any $\hat{\bm{x}} \in \hat{\mathcal{X}}_r, -1\leq t \leq 1$. Since when $t=0$, this holds by the non-negativity of $\mathrm{D}_\phi^{\hat{\mathcal{X}}_r}$ discrepancy, and when $0< t\leq 1$, we have
\begin{align*}
   \mathcal{E}_{\text {trans}}\left(\hat{\bm{x}}\right) \leq \inf_{0< t\leq 1}\frac{\mathrm{D}_\text{TV}^{\hat{\mathcal{X}}_r}(P_\mathcal{T} \| P_\mathcal{S}) }{t}.
\end{align*}
Overall, we have 
\begin{align*}
   \mathcal{E}_{\text {trans}}\left(\hat{\bm{x}}\right) \leq \inf_{0\leq t\leq 1}\frac{\mathrm{D}_\text{TV}^{\hat{\mathcal{X}}_r}(P_\mathcal{T} \| P_\mathcal{S}) }{t}.
\end{align*}
Substituting $t$ with $c_1$ gives the desired results.

We further simplify $\mathrm{D}_\text{TV}^{\hat{\mathcal{X}}_r}(P_\mathcal{T} \| P_\mathcal{S}) \nonumber$ as follows

\begin{align}
&\mathrm{D}_\text{TV}^{\hat{\mathcal{X}}_r}(P_\mathcal{T} \| P_\mathcal{S}) \nonumber\\
=&\sup_{\hat{\bm{x}} \in \hat{\mathcal{X}}_r, -1\leq t \leq 1}  \mathbb{E}_{\w \sim P_\mathcal{T}}\left[t\ell\left(f\left({\hat{\bm{x}},\w}\right),y\right)\right]- \nonumber\\
&\inf_{\alpha \in  \mathbb{R}}\left\{\mathbb{E}_{\w \sim P_\mathcal{S}}\left[\phi_\text{TV}^*\left(t\ell\left(f\left({\hat{\bm{x}},\w}\right),y\right)+\alpha\right)\right]-\alpha\right\} \nonumber \\
= &\sup_{\hat{\bm{x}} \in \hat{\mathcal{X}}_r, -1\leq t \leq 1}  t(\mathbb{E}_{\w \sim P_\mathcal{T}}\left[\ell\left(f\left({\hat{\bm{x}},\w}\right),y\right)\right] \nonumber \\
&-\mathbb{E}_{\w \sim P_\mathcal{S}}\left[\ell\left(f\left({\hat{\bm{x}},\w}\right),y\right)\right]) \label{eq_lemma_tv_2} \\
= & \sup_{\hat{\bm{x}} \in \hat{\mathcal{X}}_r} \mid \mathbb{E}_{\w \sim P_\mathcal{T}}\left[\ell\left(f\left({\hat{\bm{x}},\w}\right),y\right)\right]-\mathbb{E}_{\w \sim P_\mathcal{S}}\left[\ell\left(f\left({\hat{\bm{x}},\w}\right),y\right)\right]\mid.\nonumber
\end{align}

The Equation \ref{eq_lemma_tv_2} is by the fact that $\sup _{t \in[-1,1]} t \cdot a=|a|$. The above derivation recovers the integral probability metric form of TV distance defined by \cite{muller1997integral}.
\end{proof}

\begin{remark}\label{remark:TV1}
The above lemma instantiates the bound in Theorem \ref{theorem1} and gives a clearer clue for the condition $c_1, c_2 \in [0, +\infty)$ \bluetwo{subject} to the constraint $K_{\mathcal{S}}^{\hat{\bm{x}}}(c_1)\leq c_1c_2\mathbb{E}_{\w \sim P_\mathcal{S}}\left[\ell\left(f(\hat{\bm{x}},\w),y\right)\right]$. In the case of TV, we explicitly show that the constraint simplifies to $0\leq c_1\leq 1, c_2=0$. This bound shares a similar form to the domain adaptation bounds in Theorem 2 of \cite{localized} when setting $c_1=1$. While our task does not require a separate ideal joint error term, this term is instead implicitly included within the discrepancy term.
\end{remark}

{
\bluetable
\begin{remark}\label{remark:TV2}
According to Equation \ref{eq_tv_def}, $\mathrm{D}_\text{TV}^{\hat{\mathcal{X}}_r}(P_\mathcal{T} \| P_\mathcal{S})$ quantifies the model discrepancy by identifying an input $\hat{\bm{x}}^\prime $, which, being within $ \hat{\mathcal{X}_r}$, implies its role as a potential AE. This $\hat{\bm{x}}^\prime $ attempts to differentiate between the expected adversarial loss $\mathbb{E}_{\w \sim P_\mathcal{T}}\left[\ell\left(f\left({\hat{\bm{x}}^\prime,\w}\right),y\right)\right]$ and $\mathbb{E}_{\w \sim P_\mathcal{S}}\left[\ell\left(f\left({\hat{\bm{x}}^\prime,\w}\right),y\right)\right]$, that is, separating $P_\mathcal{S}$ from $P_\mathcal{T}$ through the \emph{mean adversarial vulnerability pattern}. 
Ultimately, to efficiently control the transferability gap under TV, the goal of $P_\mathcal{S}$ is to mimic the mean adversarial vulnerability patterns of target models from $P_\mathcal{T}$, rather than faithfully approximate $P_\mathcal{T}$. The latter would require minimizing the TV distance, $\mathrm{D}_\text{TV}(P_\mathcal{T} \| P_\mathcal{S})=\int|p_\mathcal{S}(\w)-p_\mathcal{T}(\w)|d\w$, which provides a less tight bound to control the transferability gap as per in Section \ref{sec:Comparison with Others}.\\
We must carefully represent $P_\mathcal{S}$ in regions that contribute the most to mimic the vulnerability of future target models in $P_\mathcal{T}$. To achieve this, we model $P_\mathcal{S}$ as a mixture of distributional components, \ie, $P_\mathcal S=\frac{1}{I}\sum_{i=1}^I P_{\mathcal S_i}$ with $I$ \bluetwo{being} the total number of attacker-owned surrogate components, then $\mathbb{E}_{P_\mathcal S}(\ell)=\frac{1}{I}\sum_{i=1}^I \mathbb{E}_{P_{\mathcal S_i}}(\ell)$. Hence, reducing $\mathrm{D}_\text{TV}^{\hat{\mathcal X}_r}(P_\mathcal{T}\|P_\mathcal{S})$ amounts to matching the target mean pattern $\mathbb{E}_{P_\mathcal T}(\ell)$ with an average of component-wise surrogate mean patterns $\{\mathbb{E}_{P_{\mathcal S_i}}(\cdot)\}_{i=1}^I$. This motivates constructing $P_\mathcal{S}$ using diverse surrogate components $\left\{P_{\mathcal{S}_i}\right\}_{i=1}^I$ whose mean vulnerability patterns are complementary rather than redundant, dubbed \emph{between-distribution diversity}: if two components generally exhibit similar vulnerability, one will be largely redundant in averaging and contribute minimally to approximating target vulnerability. Moreover, averaging over diverse components helps to smooth the risk from potentially unmatched surrogate choices. 


\end{remark}

\begin{corollary}[\textbf{KL Instantiation of Theorem \ref{theorem1}}]\label{corollary:KL}
Suppose $\ell: \mathcal{Y} \times \mathcal{Y} \rightarrow [0,1]$. Given the surrogate model distribution $P_\mathcal{S}$ and target model distribution $P_\mathcal{T}$, for any $\hat{\bm{x}} \in \hat{\mathcal X}_r$, $c_1>0$, and $c_2 \geq \frac{e^{c_1}-1-c_1}{c_1}$, we have
\begin{align*}
\mathcal{E}_{\text{trans}}(\hat {\bm{x}})
\le \frac{1}{c_1}D_{\text{KL}}^{\hat{\mathcal X}_r}(P_{\mathcal T} \| P_{\mathcal S})
+c_2\,r,
\end{align*}
where $D_{\text{KL}}^{\hat{\mathcal X}_r}(P_{\mathcal T}\| P_{\mathcal S})=\sup_{\hat{\bm{x}} \in\hat{\mathcal X}_r,\,t\in\mathbb R} t\,\mathbb E_{\bm{w}\sim P_{\mathcal T}}[\ell(f(\hat {\bm{x}},\bm{w}),y)] -\log \mathbb E_{\bm{w}\sim P_{\mathcal S}}[e^{t\ell(f(\hat {\bm{x}},\bm{w}),y)}]$.
\end{corollary}

\begin{proof}
Here, we instantiate the bound in Theorem \ref{theorem1} with KL divergence. Specifically, let $\phi_{\text{KL}}(u)=u\log u-u+1$, its convex conjugate function is $\phi_{\text{KL}}^*(v)=e^v-1$. Substituting $\phi_{\text{KL}}^*$ into the adversarial model discrepancy in Definition \ref{def_our_f} yields:

\begin{align}
&\mathrm{D}_\text{KL}^{\hat{\mathcal{X}}_r}(P_\mathcal{T} \| P_\mathcal{S}) \nonumber\\
=&\sup_{\hat{\bm{x}}^{\prime} \in \hat{\mathcal{X}}_r,  t \in \mathbb{R}}  \mathbb{E}_{\w \sim P_\mathcal{T}}[t\ell(f({\hat{\bm{x}}^{\prime},\w}),y)]- \nonumber\\
&\inf_{\alpha \in  \mathbb{R}}\{\mathbb{E}_{\w \sim P_\mathcal{S}}[\phi_\text{KL}^*(t\ell(f({\hat{\bm{x}}^{\prime},\w}),y)+\alpha)]-\alpha\} \nonumber\\
=&\sup_{\hat{\bm{x}}^{\prime} \in \hat{\mathcal{X}}_r,  t \in \mathbb{R}}  \mathbb{E}_{\w \sim P_\mathcal{T}}[t\ell(f({\hat{\bm{x}}^{\prime},\w}),y)]- \nonumber\\
&\inf_{\alpha \in  \mathbb{R}}\{e^{\alpha}\mathbb{E}_{\w \sim P_\mathcal{S}}[e^{t\ell(f({\hat{\bm{x}}^{\prime},\w}),y)}]-1-\alpha\} \nonumber \\
=&\sup_{\hat{\bm{x}}^{\prime} \in \hat{\mathcal{X}}_r,  t \in \mathbb{R}}  \mathbb{E}_{\w \sim P_\mathcal{T}}[t\ell(f({\hat{\bm{x}}^{\prime},\w}),y)]- \nonumber\\
&\log \mathbb{E}_{\w \sim P_\mathcal{S}}[e^{t\ell(f({\hat{\bm{x}}^{\prime},\w}),y)}], \label{eq_kl1}
\end{align}
where Equation \ref{eq_kl1} holds \bluetwo{since} the solution of $\inf_{\alpha \in  \mathbb{R}}\{\mathbb{E}_{\w \sim P_\mathcal{S}}[e^{t\ell+\alpha}-1]-\alpha\}$ is $\log \mathbb{E}_{\w \sim P_\mathcal{S}}[e^{t\ell}]$, which is obtained at $\alpha^*=-\log \mathbb{E}_{\w \sim P_\mathcal{S}}[e^{t\ell}]$.

Similarly, $K_{\mathcal{S}}^{\hat{\bm{x}}}(t) $ can be written as
\begin{align}
    K_{\mathcal{S}}^{\hat{\bm{x}}}(t) =& \inf _\alpha \left\{\mathbb{E}_{\w \sim P_\mathcal{S}}\left[\phi^*\left(t \ell\left(f(\hat{\bm{x}},\w),y\right)+\alpha\right)\right]-\alpha\right\} \nonumber \\
    &-\mathbb{E}_{\w \sim P_\mathcal{S}}\left[t \ell\left(f(\hat{\bm{x}},\w),y\right)\right] \nonumber\\
    =& \log \mathbb{E}_{\w \sim P_\mathcal{S}}\left[e^{t\ell\left(f\left({\hat{\bm{x}},\w}\right),y\right)}\right]-\mathbb{E}_{\w \sim P_\mathcal{S}}\left[t \ell\left(f(\hat{\bm{x}},\w),y\right)\right],\label{eq_kl2}
\end{align}
which is the mean-removed CGF of the surrogate loss distribution.

Since $\ell: \mathcal{Y} \times \mathcal{Y} \rightarrow [0,1]$, by convexity of $\exp(\cdot)$, we have $e^{t\ell} \leq (1-\ell)e^0+\ell e^t \leq 1+(e^t-1)\ell$.
Taking the expectation and logarithm on both sides, we have
\begin{align*}
    \log \mathbb{E}_{\w \sim P_\mathcal{S}}[e^{t\ell}]\leq \log (1+(e^t-1)\mathbb{E}_{\w \sim P_\mathcal{S}}[\ell]).
\end{align*}
As $\log (1+z)\leq z$ holds when $z > -1$ and $(e^t-1)\mathbb{E}_{\w \sim P_\mathcal{S}}[\ell] \geq 0$, we have
\begin{align}
    \log \mathbb{E}_{\w \sim P_\mathcal{S}}[e^{t\ell}]\leq (e^t-1)\mathbb{E}_{\w \sim P_\mathcal{S}}[\ell]. \label{eq_kl3}
\end{align}

Plugging \bluetwo{Equation} \ref{eq_kl3} into \bluetwo{Equation} \ref{eq_kl2} yields
\begin{align*}
    K_{\mathcal{S}}^{\hat{\bm{x}}}(t) =& \log \mathbb{E}_{\w \sim P_\mathcal{S}}\left[e^{t\ell\left(f\left({\hat{\bm{x}},\w}\right),y\right)}\right] \nonumber\\
    &-\mathbb{E}_{\w \sim P_\mathcal{S}}\left[t \ell\left(f(\hat{\bm{x}},\w),y\right)\right]\\
    &\leq (e^t-1-t)\mathbb{E}_{\w \sim P_\mathcal{S}}[\ell\left(f(\hat{\bm{x}},\w),y\right)].
\end{align*}

Let $t=c_1>0$ and $c_2 \geq \frac{(e^{c_1}-1-{c_1})}{c_1}$ satisfy the constraint $K_{\mathcal{S}}^{\hat{\bm{x}}}(c_1)\leq c_1c_2\mathbb{E}_{\w \sim P_\mathcal{S}}\left[\ell\left(f(\hat{\bm{x}},\w),y\right)\right]$, then we have 
\begin{align*}
\mathcal{E}_{\text{trans}}(\hat {\bm{x}})
\le \frac{1}{c_1}D_{\text{KL}}^{\hat{\mathcal X}_r}(P_{\mathcal T} \| P_{\mathcal S})
+c_2\,r,
\end{align*}
which completes the proof.
\end{proof}

\begin{remark}\label{remark:KL1}
As $c_1 > 0$, the induced $c_2(c_1)=\frac{(e^{c_1}-1-{c_1})}{c_1}$ increases with $c_1$ and grows rapidly for large $c_1$. Meanwhile, by Taylor expansion, $c_2(c_1)=\frac{c_1}{2}+O(c_1^2)$ as $c_1\rightarrow0$. 
Therefore, a large $c_1$ will downweight the model discrepancy with $\frac{1}{c_1}$ term but amplify $r$ through $c_2(c_1)$. 
When $r$ is small which reflects a strong surrogate attack performance, we can choose a larger $c_2$, \eg, $c_2=1$, which is tolerable as it is multiplied by a vanishing $r$. This yields a coefficient of $\frac{1}{c_1}=0.796$ for the KL term.
However, if $r$ is not negligible, we prefer a smaller $c_2$, \eg, $c_2=0.1$, at the cost of a larger coefficient $\frac{1}{c_1}=5.3$ on the discrepancy term.
\end{remark}

\begin{remark}\label{remark:KL2}
Compared to the TV-instantiated discrepancy which matches \textit{means} of the adversarial vulnerability across $P_{\mathcal S}$ and $P_{\mathcal T}$ as discussed in Remark \ref{remark:TV2}, the KL-instantiated discrepancy is inherently \textit{tail-sensitive}. Specifically, Lemma~\ref{lemma:KL-undercoverage} formalizes that $D_{\text{KL}}^{\hat{\mathcal X}_r}$ can penalize $P_{\mathcal{S}}$'s under-coverage of target vulnerability patterns. Therefore, if $P_{\mathcal T}$ places non-negligible probability mass on a subset of models exhibiting different adversarial vulnerability patterns from what $P_{\mathcal S}$ focuses on, but $P_{\mathcal S}$ assigns these models, or their corresponding vulnerability \bluetwo{patterns}, extremely small probability, then $D_{\text{KL}}^{\hat{\mathcal X}_r}(P_{\mathcal T}\|P_{\mathcal S})$ can become very large.\\
Algorithmically, to control the KL-instantiated discrepancy term, the surrogate distribution $P_{\mathcal S}$ should aim to cover a broad set of plausible vulnerability patterns, so that no target pattern is missing. Since a single surrogate component typically captures limited vulnerability patterns, this provides theoretical support for constructing $P_{\mathcal S}$ with high between-distribution diversity, \ie, using a mixture of surrogate components that collectively cover heterogeneous vulnerability patterns. This reduces the risk that a \bluetwo{plausible} target vulnerability pattern is under-covered in $P_{\mathcal S}$, which is exactly the failure mode that the KL-instantiated discrepancy penalizes most strongly.
\end{remark}

\begin{corollary}[\textbf{$\chi^2$ Instantiation of Theorem \ref{theorem1}}]\label{corollary_chi2}
Suppose $\ell: \mathcal{Y} \times \mathcal{Y} \rightarrow [0,1]$. Given the surrogate model distribution $P_\mathcal{S}$ and target model distribution $P_\mathcal{T}$, for any $\hat{\bm{x}} \in \hat{\mathcal X}_r$, $c_1>0$, and $c_2 \geq \frac{{c_1}}{4} \cdot \frac{\operatorname{Var}_{P_{\mathcal S}}(\ell (f(\hat{\bm x}, \bm w), y))}{\mathbb{E}_{\bm w \sim  P_{\mathcal S}}[\ell (f(\hat{\bm x}, \bm w), y)]}$, we have
\begin{align*}
\mathcal{E}_{\text{trans}}(\hat {\bm{x}})
\le \frac{1}{c_1}D_{{{\chi}^2}}^{\hat{\mathcal X}_r}(P_{\mathcal T} \| P_{\mathcal S})+c_2\,r,
\end{align*}
where 
\begin{align}
    &D^{\hat{\mathcal{X}}_r}_{\chi^2}(P_{\mathcal T} \| P_{\mathcal S})=\sup _{\hat{\bm{x}}^{\prime} \in \hat{\mathcal{X}}_r, t \in \mathbb{R}}\{t(\mathbb{E}_{\bm w \sim P_{\mathcal T}}[\ell(f(\hat{\bm x}^{\prime}, \bm w ), y)] \nonumber \\
    &-\mathbb{E}_{\bm w \sim P_{\mathcal S}}[\ell(f(\hat{\bm x}^{\prime}, \bm w), y)])-\frac{t^2}{4} \operatorname{Var}_{\bm w \sim P_{\mathcal S}}(\ell(f(\hat{\bm x}^{\prime}, \bm w), y))\}.\label{eq_chi2_1}
\end{align}

\end{corollary}

\begin{proof}
Here, we instantiate the bound in Theorem \ref{theorem1} with Pearson $\chi^2$ divergence. Specifically, let $\phi_{\chi^2}(u)=(u-1)^2$, its convex conjugate function is $\phi_{\chi^2}^{\star}(v)=v+\frac{v^2}{4}$. Substituting $\phi_{\chi^2}^{\star}$ into the adversarial model discrepancy in Definition \ref{def_our_f} yields:
\begin{align*}
   D^{\hat{\mathcal{X}}_r}_{\chi^2}\left(P_{\mathcal T} \| P_{\mathcal S}\right) &=\sup _{\hat{\bm x}^{\prime} \in \hat{\mathcal{X}}_r, t \in \mathbb{R}} \mathbb{E}_{\bm w \sim P_{\mathcal T}}[t \ell]\\
    &-\inf _{\alpha \in \mathbb{R}}\left\{\mathbb{E}_{\bm w \sim P_{\mathcal S}}\left[\phi_{\chi^2}^*(t \ell+\alpha)\right]-\alpha\right\}\\
    &=\sup _{\hat{\bm x}^{\prime} \in \hat{\mathcal{X}}_r, t \in \mathbb{R}} \mathbb{E}_{\bm w \sim P_{\mathcal T}}[t \ell]\\
    &-\inf _{\alpha \in \mathbb{R}}\left\{t\mathbb{E}_{\bm w \sim P_{\mathcal S}}[\ell]+\frac{1}{4}\mathbb{E}_{\bm w \sim P_{\mathcal S}}\left[(t \ell+\alpha)^2\right]\right\},
\end{align*}
where $\ell$ abbreviates $\ell\left(f\left(\hat{\bm x}^{\prime}, \bm w\right), y\right)$.

Since the infimum of $\mathbb{E}_{\bm w \sim P_{\mathcal S}}\left[(t \ell+\alpha)^2\right]$ over $\alpha$ is $t^2 \operatorname{Var}_{P_{\mathcal S}}(\ell)$, obtained at $\alpha^*=-t \mathbb{E}_{\w \sim P_{\mathcal S}}[\ell]$, we have
\begin{align*}
    &D^{\hat{\mathcal{X}}_r}_{\chi^2}\left(P_{\mathcal T} \| P_{\mathcal S}\right)\\
    &=\sup _{\hat{\bm x}^{\prime} \in \hat{\mathcal{X}}_r, t \in \mathbb{R}} \left\{t\left(\mathbb{E}_{\bm w \sim P_{\mathcal T}}[\ell]-\mathbb{E}_{\bm w \sim P_{\mathcal S}}[\ell]\right)-\frac{t^2}{4} \operatorname{Var}_{P_{\mathcal S}}(\ell)\right\} .
\end{align*}

Similarly, $K_{\mathcal{S}}^{\hat{\bm{x}}}(t) $ can be written as
\begin{align}
    K_{\mathcal{S}}^{\hat{\bm{x}}}(t) =& \inf _\alpha \left\{\mathbb{E}_{\w \sim P_\mathcal{S}}\left[\phi_{\chi^2}^*\left(t \ell+\alpha\right)\right]-\alpha\right\}-\mathbb{E}_{\w \sim P_\mathcal{S}}\left[t \ell\right]\nonumber\\
    =&\frac{t^2}{4} \operatorname{Var}_{P_{\mathcal S}}(\ell).\label{eq_chi2_4}
\end{align}

Then the constraint $K^{\hat{\x}}_{\mathcal{S}}\left(c_1\right) \leq c_1 c_2 \mathbb{E}_{\bm w \sim  P_{\mathcal S}}[\ell]$ in Theorem \ref{theorem1} becomes
\begin{align*}
    \frac{{c_1}^2}{4} \operatorname{Var}_{P_{\mathcal S}}(\ell) \leq c_1 c_2 \mathbb{E}_{\bm w \sim  P_{\mathcal S}}[\ell],
\end{align*}
which is satisfied when
\begin{align*}
    c_1 >0,\quad c_2 \geq \frac{{c_1}}{4} \cdot \frac{\operatorname{Var}_{P_{\mathcal S}}(\ell)}{\mathbb{E}_{\bm w \sim  P_{\mathcal S}}[\ell]},
\end{align*}
and applying Theorem \ref{theorem1} completes the proof.
\end{proof}

\begin{remark}\label{remark:chi2_1}
Since $\ell \in[0,1]$, we have $\operatorname{Var}(\ell)=\mathbb{E}\left[\ell^2\right]-(\mathbb{E}[\ell])^2 \leq \mathbb{E}[\ell]-(\mathbb{E}[\ell])^2 \leq \mathbb{E}[\ell]$. Therefore, a simple sufficient condition is $c_2=\frac{c_1}{4}$ for any $c_1>0$. Under $\chi^2$, $c_2$ grows slower with $c_1$ than under KL. For example, $c_2=1$ leads to a coefficient of $c_2=0.25$ for the $\chi^2$ discrepancy term, thus providing a tighter trade-off.
\end{remark}

\begin{remark}\label{remark:chi2_2}
Optimizing over $t$ yields an analytical form of Equation \ref{eq_chi2_1}:
\begin{equation}\label{eq_chi2_2}
    D_{\chi^2}^{\hat{\mathcal{X}}_r}(P_{\mathcal T} \| P_{\mathcal S})=\sup_{\hat{\x}^{\prime} \in \hat{\mathcal{X}}_r} \frac{\Delta^2}{\operatorname{Var}_{P_{\mathcal S}}(\ell)},
\end{equation}
where $\Delta=\mathbb{E}_{\w \sim P_{\mathcal T}}[\ell]-\mathbb{E}_{\w \sim P_{\mathcal S}}[\ell]$. The optimal $t^{\star}$ is obtained at $\frac{2 \Delta}{{\operatorname{Var}_{P_{\mathcal S}}(\ell)}}$. 
In Equation \ref{eq_chi2_2}, the numerator measures the mean shift \wrt~adversarial vulnerability between target and surrogate model distributions, which aligns with the mean-matching view which is also emphasized under the TV instantiation, while the denominator measures the variance of surrogate models' vulnerability on that same AE, rendering $\chi^2$-instantiated discrepancy further redundancy-sensitive. Concretely, even a modest mean mismatch $\Delta$ can be amplified if surrogate models behave almost identically on some
$\hat{\x}'\in\hat{\mathcal X}_r$ so that $\operatorname{Var}_{P_{\mathcal S}}(\ell)$ becomes small.
In contrast, if surrogate models exhibit diverse patterns of adversarial vulnerability on $\hat{\x}'\in\hat{\mathcal{X}}_r$, then $\operatorname{Var}_{{P}_{\mathcal S}}(\ell)$ is larger, reducing $D_{\chi^2}^{\hat{\mathcal{X}}_r}$.\\
In practice, this highlights the principle of building $P_\mathcal{S}$ that it should span different vulnerability patterns. Specifically, building $P_\mathcal{S}$ multi-component explicitly helps here. To explicitly see this, if $P_{\mathcal S}$ is built as a mixture of components $\{{P}_{{\mathcal S}_i}\}_{i=1}^I$, then the surrogate variance admits the decomposition
\begin{align}\label{eq_chi2_3}
\operatorname{Var}_{P_{\mathcal S}}(\ell)
=\operatorname{Var}_{i}(\mathbb{E}_{P_{{\mathcal S}_i}}[\ell])+\mathbb{E}_{i}[\operatorname{Var}_{P_{{\mathcal S}_i}}(\ell)],
\end{align}
Increasing the first term $\operatorname{Var}_{i}(\mathbb{E}_{P_{{\mathcal S}_i}}[\ell])$ explicitly motivates including components with inherently different mean adversarial vulnerability patterns, \ie, between-distribution diversity, which \bluetwo{is} also motivated under TV and KL from different theoretical perspectives. 
Increasing the second term $\mathbb{E}_{i}[\operatorname{Var}_{P_{\mathcal{S}_i}}(\ell)]$ further explicitly emphasizes within-distribution diversity, \ie, for each component $P_{\mathcal{S}_i}$, ensuring non-trivial diversity among surrogate samples drawn from the same component.
In practice, attackers do not optimize AEs over the full distribution $P_{\mathcal S}$ but over a finite surrogate sample set $\mathcal M_{\mathcal S}$. Therefore, increasing the within-distribution diversity directly motivates attackers to collect behaviorally distinct surrogate samples within each component rather than adding behaviorally redundant models, which merely increase computation while contributing little to the within-component variance.
In summary, both dimensions of diversity enlarge the variance and thereby suppress the $\chi^2$-instantiated discrepancy, offering a theoretical rationale for constructing surrogate sets that are diverse both within and between distributional components.
\end{remark}


}

\subsection{Proof of Theorem \ref{theorem_pac}}\label{app: Proof of theorem_pac}
Below, we show an empirical estimation of the surrogate adversarial risk $R_{\mathcal{S}}(\hat{\bm{x}})$. The proof technique we used here is inspired from Foret \etal \cite{sam} and Chatterji \etal \cite{intriguing}.
\noindent \textbf{Theorem 3. (Surrogate risk bound)} For any $\rho>0$, $0<\delta<1$, model distribution $P_\mathcal{S}$, and $\hat{\bm{x}} \in \hat{\mathcal{X}_r}$, with probability $1-\delta$ over the choice of surrogate model set $\mathcal{M}_\mathcal{S}\sim P_\mathcal{S}$ with size $K \in \mathbb{N}$, we have
\begin{equation}\label{eq_theorem3_1}
\begin{aligned}
    &R_{\mathcal{S}}(\hat{\bm{x}})
     \leq \max_{\|\boldsymbol{\epsilon}\|_2 \leq \rho}R_{\hat{\mathcal{S}}}(\hat{\bm{x}}+\bm{\epsilon})+\\
     &\sqrt{\frac{\frac{d}{2}\log (1+\frac{\gamma^2}{\rho^2}(1+\sqrt{\frac{\log K}{d}})^2)+\log\frac{K}{\delta}+\tilde{\mathcal{O}}(1)}{2(K-1)}}.
\end{aligned}
\end{equation}
where $\tilde{\mathcal{O}}(1)$ term \blue{corresponds} to $\varepsilon=\frac{1}{2}+2\log (2+3d+6r^2K+4d\log(\sqrt{d}+\sqrt{logK}))$.
\begin{proof}
For $\hat{\bm{x}} \in \hat{\mathcal{X}_r}$ with $ \hat{\mathcal{X}_{r}}=\left\{\hat{\bm{x}} \in \hat{\mathcal{X}} \mid R_{\mathcal{S}}\left(\hat{\bm{x}}\right) \leq r\right\}$, if $r$ is a relatively small value, then we assume without loss of generality that adding a noise $\bm{\epsilon}$ following the Gaussian distribution, \ie, $\epsilon_i \sim \mathcal{N}(0,\sigma^2)$, on $\hat{\bm{x}}$ will not further reduce its surrogate adversarial risk. That is to say,
\begin{equation}\label{eq_theorem3_15}
     R_{\mathcal{S}}(\hat{\bm{x}})\leq \mathbb{E}_{\epsilon_i \sim \mathcal{N}(0,\sigma^2)}[R_\mathcal{S}(\hat{\bm{x}}+\bm{\epsilon})].
\end{equation}

From Lemma \ref{lemma_pac_trans}, we have that for any model distribution $P_\mathcal{S}$, $K \in \mathbb{N}$, prior distribution $\mathcal{P}$ over AEs, $0<\delta<1$, with probability $1-\delta$ over the choice of set $\mathcal{M}_\mathcal{S}\sim P_\mathcal{S}$ with size $K$, for any posterior distributions $\mathcal{Q}$ over AEs, the following bound holds:
\begin{align}\label{eq_theorem3_6}
     \mathbb{E}_{\hat{\bm{x}} \sim \mathcal{Q}}[R_\mathcal{S}(\hat{\bm{x}})] \leq \mathbb{E}_{\hat{\bm{x}} \sim \mathcal{Q}}[R_{\hat{\mathcal{S}}}(\hat{\bm{x}})] +\sqrt{\frac{\text{KL}(\mathcal{Q}||\mathcal{P})+\log\frac{K}{\delta}}{2(K-1)}}.
\end{align}

We first consider the KL divergence term in Equation \ref{eq_theorem3_6}. Let the prior $\mathcal{P}$ be a Gaussian distribution centered at the benign image $\bm{x}$ with covariance matrix ${\sigma}_p^2\bm{I}$, \ie, $\mathcal{P}=\mathcal{N}(\bm{x}, {\sigma}_p^2\bm{I})$. Let posterior $\mathcal{Q}$ be a Gaussian distribution centered at $\hat{\bm{x}}$ which is the benign image $\bm{x}$ with an additive adversarial perturbation $\bm{\xi}$ ---$\bm{x}+\bm{\xi}$, with covariance matrix ${\sigma}^2\bm{I}$, \ie, $\mathcal{Q}=\mathcal{N}(\bm{x}+\bm{\xi}, {\sigma}^2\bm{I})$, as if we have added the noise $\epsilon_i \sim \mathcal{N}(0,\sigma^2)$. Then KL divergence term can be written as:
\begin{align}\label{eq_theorem3_4}
    \text{KL}(\mathcal{Q} \| \mathcal{P})=\frac{1}{2}\left[\frac{d\sigma^2+\left\|\bm{\xi}\right\|_2^2}{\sigma_p^2}-d+d\log \left(\frac{\sigma_p^2}{\sigma^2}\right)\right].
\end{align}

The variance of prior $\sigma_p$ is selected to minimize the above KL divergence. Given $\sigma$, we differentiate $\text{KL}(\mathcal{Q} \| \mathcal{P})$ with respect to $\sigma_p$ and set the derivative to zero, solving 
\begin{align*}
    {\sigma_p^{*2}}=\sigma^2+{\|\bm{\xi}\|}_2^2/d.
\end{align*}
However, the prior is selected in advance and independent of the training data, which the posterior depends on. Hence, the above solution is not allowed. However, one can optimize the prior standard deviation $\sigma_p$ over a pre-defined set and use a union bound to get the bound for best $\sigma_p$ in this set. Let the pre-defined set be $\{c\exp{((1-j)/d)}|j \in \mathbb{N}\}$. If for each $j$ one chooses $\sigma_p$ to be
\begin{align*}
    \sigma_p=c\exp{((1-j)/d)},
\end{align*}
such that the above PAC-Bayes bound holds with probability $1-\delta_j$ where $\delta_j=\frac{6\delta}{\pi^2j^2}$, then all \bluetwo{bounds} can be combined according to the union bound and hold with probability
\begin{align*}
   1-\sum_{j=1}^\infty\delta_j=1-\sum_{j=1}^\infty\frac{6\delta}{\pi^2j^2}=1-\delta. 
\end{align*}

As the right hand side of Equation \ref{eq_theorem3_1} is lower bounded by $\sqrt{\frac{d \log \left(1+\frac{\|\bm{\xi}\|^2_2}{\rho^2}\right)}{4K}}$ since $\bm{\xi}$ is restricted in the ball of radius $\gamma$.  When $\sqrt{\frac{d \log \left(1+\frac{\|\bm{\xi}\|^2_2}{\rho^2}\right)}{4K}}>r$, this bound holds trivially. Thus we only consider when $\sqrt{\frac{d \log \left(1+\frac{\|\bm{\xi}\|^2_2}{\rho^2}\right)}{4K}}\leq r$, leading to
\begin{align}\label{eq_theorem3_2}
    \|\bm{\xi}\|^2_2 \leq \rho^2\left(\exp{\left(\frac{4r^2K}{d}\right)-1}\right).
\end{align}

Therefore, by Equation \ref{eq_theorem3_2}, we have
\begin{align}
    \sigma^2+{\|\bm{\xi}\|}_2^2/d \nonumber
    &\leq \sigma^2+\frac{\rho^2}{d}\left(\exp{\left(\frac{4r^2K}{d}\right)}-1\right)\nonumber\\
    &\leq \sigma^2+\rho^2\exp{\left(\frac{4r^2K}{d}\right)}.\label{eq_theorem3_5}
\end{align}
As $ \sigma_p=c\exp{((1-j)/d)}$, we consider the case $j=1-d\log(\sigma_p^{2}/c)=\left\lfloor 1-d \log \left(\left(\sigma^2+\|\bm{\xi}\|_2^2 / d\right) / c\right)\right\rfloor$, where $c$ can be set as $\sigma^2+\rho^2\exp{\left(\frac{4r^2K}{d}\right)}$to make sure $j \in \mathbb{N}$ according to Equation \ref{eq_theorem3_5}. 
For this $j$, we have
\begin{align*}
    1-d\log(\sigma_p^{2}/c) &\leq  1-d \log \left(\left(\sigma^2+\|\bm{\xi}\|_2^2 / d\right) / c\right)\nonumber\\
    \sigma_p^{2} &\geq \sigma^2+\|\bm{\xi}\|_2^2 / d.
\end{align*}
Similarly, we can derive an upper bound for $\sigma_p^{2} $,
\begin{align*}
    1-d\log(\sigma_p^{2}/c) &\geq  1-d\log \left(\left(\sigma^2+\|\bm{\xi}\|_2^2 / d\right) / c\right)-1\nonumber\\
    \sigma_p^{2} &\leq \exp(\frac{1}{d})(\sigma^2+\|\bm{\xi}\|_2^2 / d).
\end{align*}

With the above lower bound and upper bound, the KL term in Equation \ref{eq_theorem3_4} can be written as:
\begin{equation}\label{eq_theorem3_8}
\begin{aligned}
    & \text{KL}(\mathcal{Q} \| \mathcal{P})=\frac{1}{2}\left[\frac{d\sigma^2+\left\|\bm{\xi}\right\|_2^2}{\sigma_p^2}-d+d\log \left(\frac{\sigma_p^2}{\sigma^2}\right)\right]\\
    &\leq\frac{1}{2}\left[\frac{d\sigma^2+\left\|\bm{\xi}\right\|_2^2}{\sigma^2+\|\bm{\xi}\|_2^2 / d}-d+d\log \left(\frac{\exp(\frac{1}{d})(\sigma^2+\|\bm{\xi}\|_2^2 / d)}{\sigma^2}\right)\right]\\
    &=\frac{1}{2}\left[1+d\log \left(\frac{d\sigma^2+\|\bm{\xi}\|_2^2 }{d\sigma^2}\right)\right].
\end{aligned}
\end{equation}

We then consider the log term in Equation \ref{eq_theorem3_6}. Under the case of above discussed $j$, the bound holds with probability $1-\delta_j$ where $\delta_j=\frac{6\delta}{\pi^2j^2}$. Therefore, the \bluetwo{remaining} part could be simplified as
\begin{align}
  &\log\frac{K}{\delta_j}=\log\frac{K}{\delta}+\log\frac{\pi^2j^2}{6} \nonumber \\
  & \leq \log\frac{K}{\delta}+\log\frac{\pi^2{\left( 1+d \log \left(c / \left(\sigma^2+\|\bm{\xi}\|_2^2 / d\right)\right)\right)}^2}{6}\label{eq_theorem3_7}\\
  & \leq \log\frac{K}{\delta}+\log\frac{\pi^2{\left( 1+d \log \left(c / \sigma^2\right)\right)}^2}{6}\nonumber\\
  & = \log\frac{K}{\delta}+\log\frac{\pi^2{\left( 1+d \log \left(1+\exp{\left(\log(\frac{\rho^2}{\sigma^2})+\frac{4r^2K}{d}\right)}\right)\right)}^2}{6}\label{eq_theorem3_9}\\
    & \leq \log\frac{K}{\delta}+\log\frac{\pi^2{\left( 1+2d+4r^2K+d\log(\frac{\rho^2}{\sigma^2})\right)}^2}{6}\label{eq_theorem3_10}\\
     & \leq \log\frac{K}{\delta}+2\log (2+3d+6r^2K+2d\log(\frac{\rho^2}{\sigma^2}))\label{eq_theorem3_11},
\end{align}
where Equation \ref{eq_theorem3_7} is by the upper bound of $j$. Equation \ref{eq_theorem3_9} is by the value set for $c$. \ref{eq_theorem3_10} is by the fact that $\log \left(1+{e}^x\right) \leq 2+x$, for every $x \geq 0$.

Substituting the upper bound for KL term in Equation \ref{eq_theorem3_8} and the upper bound for log term in Equation \ref{eq_theorem3_11}, we have the following PAC-Bayes bound
\begin{small}
\begin{equation}\label{eq_theorem3_14}
\begin{aligned}
     &\mathbb{E}_{\epsilon_i \sim \mathcal{N}(0,\sigma^2)}[R_\mathcal{S}(\hat{\bm{x}}+\bm{\epsilon})] \leq \mathbb{E}_{\epsilon_i \sim \mathcal{N}(0,\sigma^2)}[R_{\hat{\mathcal{S}}}(\hat{\bm{x}}+\bm{\epsilon})] +\\
     &\sqrt{\frac{\frac{d}{2}\log \left(\frac{d\sigma^2+\|\bm{\xi}\|_2^2 }{d\sigma^2}\right)+2\log (2+3d+6r^2K+2d\log(\frac{\rho^2}{\sigma^2}))+\varepsilon^\prime}{2(K-1)}},
    \end{aligned}
\end{equation}
\end{small}
where $\varepsilon^\prime=\log\frac{K}{\delta}+\frac{1}{2}$.

As we set $\epsilon_i \sim \mathcal{N}(0,\sigma^2)$, we have 
$$
\|\boldsymbol{\epsilon}\|_2^2=\sum_{i=1}^d \epsilon_i^2=\sigma^2 \chi_d^2,
$$
where $\chi_d^2$ is a chi-square random variable with $d$ degrees of freedom (mean $d$, variance $2 d$ ).
By Lemma 1 in \cite{laurent2000adaptive}, we have 
$$
\operatorname{Pr}\left[\chi_d^2-d \geq 2 \sqrt{d t}+2 t\right] \leq e^{-t}, \quad t>0 .
$$
Multiplying both sides by $\sigma^2$ yields
\begin{align}\label{eq_theorem3_12}
   \operatorname{Pr}\left[\|\boldsymbol{\epsilon}\|_2^2-d \sigma^2 \geq 2 \sigma^2 \sqrt{d t}+2 \sigma^2 t\right] \leq e^{-t}. 
\end{align}
Set $ e^{-t}=1/\sqrt{K}$, then $t=\log{\sqrt{K}}$. Substituting $t$ into Equation \ref{eq_theorem3_12}, we have
\begin{align}
   \operatorname{Pr}\left[\|\boldsymbol{\epsilon}\|_2^2 \leq \sigma^2(d + 2 \sqrt{d \log{\sqrt{K}}}+2  \log{\sqrt{K}})\right] \geq 1-1/\sqrt{K}. 
\end{align}
As $d + 2 \sqrt{d \log{\sqrt{K}}}+2  \log{\sqrt{K}}\leq d\left(1+\sqrt{\frac{\log K}{d}}\right)^2$, we have, with probability $\geq 1-1 / \sqrt{K}$,
\begin{align}\label{eq_theorem3_13}
    \|\boldsymbol{\epsilon}\|_2^2 \leq d\sigma^2\left(1+\sqrt{\frac{\log K}{d}}\right)^2.
\end{align}
By defining $\rho^2=d\sigma^2\left(1+\sqrt{\frac{\log K}{d}}\right)^2$, Equation \ref{eq_theorem3_13} becomes
\begin{align*}
    \|\boldsymbol{\epsilon}\|_2^2 \leq \rho^2, \quad \text{with probability } 1-1/\sqrt{K}.
\end{align*}

Combining Equation \ref{eq_theorem3_14} and \ref{eq_theorem3_15} and splitting the expectation over the two events $\left\{\|\boldsymbol{\epsilon}\|_2 \leq \rho\right\}$ with probability $1-1 / \sqrt{K}$ and $\left\{\|\boldsymbol{\epsilon}\|_2 > \rho\right\}$ with probability $1 / \sqrt{K}$,
\begin{small}
\begin{align*}
    &R_{\mathcal{S}}(\hat{\bm{x}}) \leq \mathbb{E}_{\epsilon_i \sim \mathcal{N}(0,\sigma^2)}[R_{\hat{\mathcal{S}}}(\hat{\bm{x}}+\bm{\epsilon})]+ \\
     &\sqrt{\frac{\frac{d}{2}\log \left(\frac{d\sigma^2+\|\bm{\xi}\|_2^2 }{d\sigma^2}\right)+2\log (2+3d+6r^2K+2d\log(\frac{\rho^2}{\sigma^2}))+\varepsilon^\prime}{2(K-1)}}\\
     &\leq (1-1 / \sqrt{K})\max_{\|\boldsymbol{\epsilon}\|_2 \leq \rho}R_{\hat{\mathcal{S}}}(\hat{\bm{x}}+\bm{\epsilon})+1 / \sqrt{K}+\\
     &\sqrt{\frac{\frac{d}{2}\log (1+\frac{\|\bm{\xi}\|_2^2 }{\rho^2}(1+\sqrt{\frac{\log K}{d}})^2)+\varepsilon^{\prime\prime}+\varepsilon^\prime}{2(K-1)}}\\
     &\leq \max_{\|\boldsymbol{\epsilon}\|_2 \leq \rho}R_{\hat{\mathcal{S}}}(\hat{\bm{x}}+\bm{\epsilon})+\\
     &\sqrt{\frac{\frac{d}{2}\log (1+\frac{\gamma^2}{\rho^2}(1+\sqrt{\frac{\log K}{d}})^2)+\varepsilon^{\prime\prime}+\varepsilon^\prime}{2(K-1)}},
\end{align*}
\end{small}
where $\varepsilon^{\prime\prime}=2\log (2+3d+6r^2K+4d\log(\sqrt{d}+\sqrt{logK}))$. This completes the proof.
\end{proof}

\begin{table*}[t]
\centering
\caption{{Untargeted attack success rates (\%,$\uparrow$) on CIFAR-10 dataset.} The AEs are crafted from three surrogate models (ResNet-50, ViT, and WideResNet-70-16(AT)), against 20 target models falling into three prototypes (normally and adversarially trained convnets and normally trained metaformers). \textbf{Bold} denotes the best results and \underline{underlined} denotes the second best results.}
\label{ut_cifar}
\resizebox{\textwidth}{!}{%
\begin{threeparttable}
\begin{tabular}{@{}cc|ccccccccccccccc@{}}
\toprule
\multicolumn{2}{c|}{\textbf{Target Model Set}} & \begin{tabular}[c]{@{}c@{}}I-FGSM\\ \cite{ifgsm}\end{tabular} & \begin{tabular}[c]{@{}c@{}}DI2-FGSM\\ \cite{difgsm}\end{tabular} & \begin{tabular}[c]{@{}c@{}}SI-FGSM\\ \cite{nisi}\end{tabular} & \begin{tabular}[c]{@{}c@{}}Admix\\ \cite{admix}\end{tabular} & \begin{tabular}[c]{@{}c@{}}TI-FGSM\\ \cite{ti}\end{tabular} & \begin{tabular}[c]{@{}c@{}}SSA\\ \cite{ssa}\end{tabular} & \begin{tabular}[c]{@{}c@{}}SIA\\ \cite{sia}\end{tabular} & \begin{tabular}[c]{@{}c@{}}MI-FGSM\\ \cite{mifgsm}\end{tabular} & \begin{tabular}[c]{@{}c@{}}PI-FGSM\\ \cite{pifgsm}\end{tabular} & \begin{tabular}[c]{@{}c@{}}VT-FGSM\\ \cite{vt}\end{tabular} & \begin{tabular}[c]{@{}c@{}}PGN\\ \cite{pgn}\end{tabular} & \begin{tabular}[c]{@{}c@{}}CWA\\ \cite{cwa}\end{tabular} & \begin{tabular}[c]{@{}c@{}}SVRE\\ \cite{svre}\end{tabular} & \begin{tabular}[c]{@{}c@{}}RAP\\ \cite{rap}\end{tabular} & \textbf{DRAP} \\
\midrule
\multicolumn{1}{c|}{} & AlexNet\cite{alexnet} & 28.4 & 31.0 & 29.0 & 29.5 & 28.6 & 33.4 & 32.5 & 29.4 & 30.2 & 30.3 & 34.0 & \underline{35.2} & 28.2 & 30.9 & \textbf{41.5} \\
\multicolumn{1}{c|}{} & DenseNet\cite{densenet} & 88.1 & 94.8 & 91.7 & 95.2 & 88.2 & 96.3 & \textbf{97.9} & 90.4 & 92.3 & 95.8 & 96.5 & 85.7 & 88.5 & 96.5 & \underline{96.7} \\
\multicolumn{1}{c|}{} & ResNeXt\cite{resnext} & 93.3 & 97.3 & 95.8 & 97.8 & 93.3 & 97.7 & \textbf{99.0} & 94.4 & 94.8 & 97.3 & 98.9 & 88.7 & 92.3 & 97.7 & \underline{97.9} \\
\multicolumn{1}{c|}{} & WRN-28-10-drop\cite{wrn} & 87.6 & 94.4 & 91.3 & 95.1 & 87.9 & 96.7 & \textbf{97.7} & 90.7 & 92.2 & 96.0 & 96.8 & 85.0 & 87.7 & 96.5 & \underline{96.9} \\
\multicolumn{1}{c|}{} & GoogleNet\cite{googlenet} & 89.7 & 95.3 & 92.8 & 95.7 & 89.9 & 96.8 & \textbf{98.5} & 92.2 & 93.5 & 96.3 & 96.4 & 86.8 & 89.3 & 97.0 & \underline{96.9} \\
\multicolumn{1}{c|}{} & MobileNetv2\cite{mobilenetv2} & 89.2 & 96.0 & 92.4 & 96.0 & 89.1 & 97.2 & \textbf{98.2} & 92.3 & 93.1 & 96.2 & 97.3 & 86.0 & 89.1 & 96.7 & \underline{97.7} \\
\multicolumn{1}{c|}{} & PreResNet\cite{preresnet} & 96.7 & 98.7 & 98.0 & 99.0 & 96.5 & 99.0 & 99.1 & 97.4 & 96.8 & 98.3 & 99.5 & 91.8 & 95.8 & 98.5 & \underline{99.2} \\
\multicolumn{1}{c|}{\multirow{-8}{*}{\textbf{ConvNet Set}}} & \cellcolor[HTML]{D9D9D9}\textit{\textbf{Average}} & \cellcolor[HTML]{D9D9D9}81.9 & \cellcolor[HTML]{D9D9D9}86.8 & \cellcolor[HTML]{D9D9D9}84.4 & \cellcolor[HTML]{D9D9D9}86.9 & \cellcolor[HTML]{D9D9D9}81.9 & \cellcolor[HTML]{D9D9D9}88.2 & \cellcolor[HTML]{D9D9D9}\underline{89.0} & \cellcolor[HTML]{D9D9D9}83.8 & \cellcolor[HTML]{D9D9D9}84.7 & \cellcolor[HTML]{D9D9D9}87.2 & \cellcolor[HTML]{D9D9D9}88.5 & \cellcolor[HTML]{D9D9D9}79.9 & \cellcolor[HTML]{D9D9D9}81.6 & \cellcolor[HTML]{D9D9D9}87.7 & \cellcolor[HTML]{D9D9D9}\textbf{89.5} \\
\midrule
\multicolumn{1}{c|}{} & ViT-T\cite{vit} & 73.1 & 87.6 & 77.6 & 85.8 & 72.9 & 84.0 & \textbf{93.9} & 76.3 & 79.9 & 88.1 & \underline{88.7} & 75.9 & 72.4 & 84.8 & 80.4 \\
\multicolumn{1}{c|}{} & Swin-S\cite{swin} & 76.1 & 90.0 & 89.6 & 89.5 & 89.2 & 90.1 & \textbf{95.9} & 89.4 & 82.1 & 89.9 & \underline{92.6} & 77.3 & 76.3 & 90.1 & 83.7 \\
\multicolumn{1}{c|}{} & Swin-B\cite{swin} & 71.4 & 89.9 & 89.6 & 89.4 & 88.9 & 87.0 & \textbf{93.4} & 88.9 & 78.9 & 89.6 & \underline{90.3} & 74.9 & 73.2 & 89.9 & 79.7 \\
\multicolumn{1}{c|}{} & DeiT-T\cite{deit} & 74.9 & 88.8 & 79.6 & 87.9 & 75.2 & 87.5 & \textbf{95.0} & 78.6 & 80.9 & 89.4 & \underline{90.5} & 77.5 & 74.2 & 85.8 & 81.0 \\
\multicolumn{1}{c|}{} & DeiT-B\cite{deit} & 72.0 & 87.7 & 76.5 & 85.1 & 72.1 & 83.1 & \textbf{92.8} & 74.1 & 79.2 & 88.5 & \underline{88.4} & 74.0 & 72.7 & 83.9 & 78.0 \\
\multicolumn{1}{c|}{\multirow{-6}{*}{\textbf{Metaformer Set}}} & \cellcolor[HTML]{D9D9D9}\textit{\textbf{Average}} & \cellcolor[HTML]{D9D9D9}73.5 & \cellcolor[HTML]{D9D9D9}88.8 & \cellcolor[HTML]{D9D9D9}82.6 & \cellcolor[HTML]{D9D9D9}87.5 & \cellcolor[HTML]{D9D9D9}79.7 & \cellcolor[HTML]{D9D9D9}86.3 & \cellcolor[HTML]{D9D9D9}\textbf{94.2} & \cellcolor[HTML]{D9D9D9}81.4 & \cellcolor[HTML]{D9D9D9}80.2 & \cellcolor[HTML]{D9D9D9}89.1 & \cellcolor[HTML]{D9D9D9}\underline{90.1} & \cellcolor[HTML]{D9D9D9}75.9 & \cellcolor[HTML]{D9D9D9}73.8 & \cellcolor[HTML]{D9D9D9}86.9 & \cellcolor[HTML]{D9D9D9}80.6 \\
\midrule
\multicolumn{1}{c|}{} & PreActResNet-18\cite{andriushchenko2020understanding} & 23.6 & 24.4 & 23.6 & 24.1 & 23.7 & 26.1 & 24.7 & 24.3 & 23.9 & 24.4 & 25.9 & \underline{29.2} & 23.3 & 24.7 & \textbf{34.0} \\
\multicolumn{1}{c|}{} & WideResNet-28-10\cite{carmon2019unlabeled} & 14.4 & 15.0 & 14.0 & 14.7 & 14.4 & 15.9 & 14.4 & 15.1 & 14.3 & 15.0 & 16.0 & \underline{21.4} & 14.0 & 15.0 & \textbf{26.2} \\
\multicolumn{1}{c|}{} & WideResNet-28-10\cite{sehwag2020hydra} & 15.2 & 16.0 & 14.9 & 15.5 & 15.2 & 16.7 & 15.0 & 15.8 & 15.2 & 15.8 & 16.8 & \underline{22.0} & 14.6 & 15.8 & \textbf{27.0} \\
\multicolumn{1}{c|}{} & WideResNet-28-10\cite{wang2019improving} & 16.4 & 17.0 & 16.2 & 16.6 & 16.4 & 18.2 & 16.7 & 17.1 & 16.7 & 16.9 & 18.4 & \underline{22.6} & 16.0 & 17.3 & \textbf{27.4} \\
\multicolumn{1}{c|}{} & WideResNet-34-20\cite{rice2020overfitting} & 18.3 & 18.9 & 18.4 & 18.8 & 18.5 & 20.2 & 18.9 & 19.0 & 18.4 & 18.8 & 20.3 & \underline{24.7} & 18.0 & 18.9 & \textbf{29.6} \\
\multicolumn{1}{c|}{} & WideResNet-34-10\cite{zhang2019theoretically} & 18.8 & 19.6 & 18.8 & 19.2 & 18.9 & 20.6 & 19.4 & 19.6 & 19.1 & 19.4 & 20.7 & \underline{25.5} & 18.5 & 19.6 & \textbf{30.7} \\
\multicolumn{1}{c|}{} & ResNet-50\cite{engstrom2019adversarial} & 16.9 & 17.7 & 16.6 & 17.4 & 17.0 & 18.7 & 17.3 & 17.6 & 17.2 & 17.5 & 18.8 & \underline{23.7} & 16.7 & 17.8 & \textbf{29.1} \\
\multicolumn{1}{c|}{} & WideResNet-34-10\cite{huang2020self} & 19.9 & 20.6 & 19.8 & 20.3 & 20.1 & 21.2 & 20.1 & 20.5 & 20.0 & 20.4 & 21.3 & \underline{26.1} & 19.6 & 20.4 & \textbf{30.8} \\
\multicolumn{1}{c|}{\multirow{-9}{*}{\textbf{ConvNet(AT) Set}}} & \cellcolor[HTML]{D9D9D9}\textit{\textbf{Average}} & \cellcolor[HTML]{D9D9D9}17.9 & \cellcolor[HTML]{D9D9D9}18.6 & \cellcolor[HTML]{D9D9D9}17.8 & \cellcolor[HTML]{D9D9D9}18.3 & \cellcolor[HTML]{D9D9D9}18.0 & \cellcolor[HTML]{D9D9D9}19.7 & \cellcolor[HTML]{D9D9D9}18.3 & \cellcolor[HTML]{D9D9D9}18.6 & \cellcolor[HTML]{D9D9D9}18.1 & \cellcolor[HTML]{D9D9D9}18.5 & \cellcolor[HTML]{D9D9D9}19.8 & \cellcolor[HTML]{D9D9D9}\underline{24.4} & \cellcolor[HTML]{D9D9D9}17.6 & \cellcolor[HTML]{D9D9D9}18.7 & \cellcolor[HTML]{D9D9D9}\textbf{29.4} \\
\midrule
\multicolumn{2}{c|}{\cellcolor[HTML]{D9D9D9}\textit{\textbf{Overall Average}}} & \cellcolor[HTML]{D9D9D9}54.2 & \cellcolor[HTML]{D9D9D9}60.0 & \cellcolor[HTML]{D9D9D9}57.3 & \cellcolor[HTML]{D9D9D9}59.6 & \cellcolor[HTML]{D9D9D9}55.8 & \cellcolor[HTML]{D9D9D9}60.3 & \cellcolor[HTML]{D9D9D9}\underline{62.0} & \cellcolor[HTML]{D9D9D9}57.2 & \cellcolor[HTML]{D9D9D9}56.9 & \cellcolor[HTML]{D9D9D9}60.2 & \cellcolor[HTML]{D9D9D9}61.4 & \cellcolor[HTML]{D9D9D9}56.7 & \cellcolor[HTML]{D9D9D9}54.0 & \cellcolor[HTML]{D9D9D9}59.9 & \cellcolor[HTML]{D9D9D9}\textbf{63.2}
 \\
\bottomrule
\end{tabular}%
\end{threeparttable}}
\end{table*}

\section{Experiments on CIFAR-10}\label{app: cifar}

\subsection{Main Results}\label{app: cifar_main}
\textbf{Baselines} For the experiments conducted on CIFAR-10, we adopt the same set of baseline methods used in the ImageNet experiments.

\textbf{Models} For surrogate models on CIFAR-10, we adopt three architectures ($I=3$) from three prototypical model categories: ResNet-50\cite{resnet} from normally trained convnets, ViT\cite{vit} from normally trained metaformers, and WideResNet-70-16 (AT)\cite{wideresnetat} from adversarially trained convnets. Due to the limited availability of adversarially trained metaformer models on CIFAR-10, we exclude them from the surrogate models. 
For DRAP, model samples are gathered as proposals at each epoch during the fine-tuning of the three pretrained models, which are optimized using their respective training strategies over $n=40$ additional epochs. In order to get more diverse samples, we fine-tune the three pretrained models with larger constant learning rates of 0.05, 0.03, and 0.02, respectively, for ResNet-50, ViT, and WideResNet-70-16 (AT), without significantly degrading their clean accuracy.
All compared methods also generate AEs using the logits of the three surrogate models, leveraging the fusion strategy outlined in \cite{mifgsm}. To assess the transferability of these AEs, we evaluate them across a diverse set of target models. These target models are grouped into three categories: ConvNet Set, ConvNet(AT) Set, and Metaformer Set, as shown in Table \ref{ut_cifar}.

\textbf{Implementation Details} For both untargeted and targeted attack scenarios on CIFAR-10, the adversarial perturbation magnitude is constrained to $\gamma=8/255$, with a step size of $\beta_{\hat{\bm{x}}}=8/255/10$ for all methods. The RAP method retains its iteration number of 400 to ensure \bluetwo{a} comparable evaluation. For all other methods, a standard iteration count of 120 is used to maintain consistency across experiments.
For DRAP, the key hyper-parameters are configured as follows: the number of samples per model distribution $n=40$, the inner optimization iteration count $T=5$, the late start iteration number $n_{LS}=5$, the inner step size $\beta_\epsilon=0.1/255$, and the decay factor $\mu=1$. The total number of iterations for generating adversarial examples in DRAP is $n \times I = 120$, aligning with the iteration count of other methods to ensure fairness in comparison. 
For all baseline methods, we adhere to the protocols established in BlackboxBench to maintain consistency across experimental setups.

\textbf{Results of Untargeted Attacks} 
Table \ref{ut_cifar} presents the untargeted attack results on CIFAR-10 across normally trained and adversarially trained target model sets. Overall, our proposed method achieves the highest average attack success rate, surpassing both input-transformation\bluetwo{-}based and other optimization-based methods. This outcome aligns with the ImageNet findings, indicating that our strategy generalizes effectively to a dataset with different characteristics.
Taking a closer look, we observe that SIA remains highly competitive against normally trained target models but underperforms significantly on adversarially trained models compared to DRAP. In contrast, our approach maintains strong performance on adversarially trained models without sacrificing success rates on normally trained ones, striking a well-rounded balance across all target sets. This reinforces the idea that simultaneously pursuing flatness among diverse surrogate models can secure robust adversarial transferability, regardless of the target’s defense mechanisms.
Similar to ImageNet results, CWA continues to exhibit promising effectiveness against adversarially trained targets. However, its overall success rate remains lower, likely due to its finding a suboptimal universal perturbation.

\textbf{Results of Targeted Attacks}  
Table \ref{t_cifar} presents the targeted attack results on CIFAR-10. As with ImageNet, targeted attacks remain particularly challenging, especially on adversarially trained models in the ConvNet(AT) Set. Among all evaluated methods, DRAP demonstrates a clear advantage on the ConvNet(AT) Set, achieving significantly higher success rates compared to other methods, showcasing its effectiveness against robust adversarially trained targets.
Beyond the ConvNet(AT) Set, DRAP also performs strongly on the ConvNet Set and Metaformer Set, achieving competitive success rates that match or surpass other leading methods in many cases. This highlights the adaptability of DRAP across diverse target model categories, maintaining a strong balance between attacking normally trained and adversarially trained models. Overall, these results reinforce the robustness and competitiveness of DRAP in targeted attack scenarios, particularly when faced with defense models.

\begin{table*}[t]
\centering
\caption{{Targeted attack success rates (\%,$\uparrow$) on CIFAR-10 dataset.} The AEs are crafted from three surrogate models (ResNet-50, ViT, and WideResNet-70-16(AT)), against 20 target models falling into three prototypes (normally and adversarially trained convnets and normally trained metaformers). \textbf{Bold} denotes the best results and \underline{underlined} denotes the second best results.}
\label{t_cifar}
\resizebox{\textwidth}{!}{%
\begin{threeparttable}
\begin{tabular}{@{}cc|ccccccccccccccc@{}}
\toprule
\multicolumn{2}{c|}{\textbf{Target Model Set}} & \begin{tabular}[c]{@{}c@{}}I-FGSM\\ \cite{ifgsm}\end{tabular} & \begin{tabular}[c]{@{}c@{}}DI2-FGSM\\ \cite{difgsm}\end{tabular} & \begin{tabular}[c]{@{}c@{}}SI-FGSM\\ \cite{nisi}\end{tabular} & \begin{tabular}[c]{@{}c@{}}Admix\\ \cite{admix}\end{tabular} & \begin{tabular}[c]{@{}c@{}}TI-FGSM\\ \cite{ti}\end{tabular} & \begin{tabular}[c]{@{}c@{}}SSA\\ \cite{ssa}\end{tabular} & \begin{tabular}[c]{@{}c@{}}SIA\\ \cite{sia}\end{tabular} & \begin{tabular}[c]{@{}c@{}}MI-FGSM\\ \cite{mifgsm}\end{tabular} & \begin{tabular}[c]{@{}c@{}}PI-FGSM\\ \cite{pifgsm}\end{tabular} & \begin{tabular}[c]{@{}c@{}}VT-FGSM\\ \cite{vt}\end{tabular} & \begin{tabular}[c]{@{}c@{}}PGN\\ \cite{pgn}\end{tabular} & \begin{tabular}[c]{@{}c@{}}CWA\\ \cite{cwa}\end{tabular} & \begin{tabular}[c]{@{}c@{}}SVRE\\ \cite{svre}\end{tabular} & \begin{tabular}[c]{@{}c@{}}RAP\\ \cite{rap}\end{tabular} & \textbf{DRAP} \\
\midrule
\multicolumn{1}{c|}{} & AlexNet\cite{alexnet} & 3.4 & 4.3 & 3.5 & 3.8 & 3.3 & 5.1 & 4.9 & 3.8 & 4.0 & 3.8 & 5.4 & \underline{6.2} & 3.2 & 4.2 & \textbf{8.6} \\
\multicolumn{1}{c|}{} & DenseNet\cite{densenet} & 53.3 & 74.3 & 64.1 & 79.7 & 53.0 & 83.0 & \textbf{86.0} & 56.6 & 64.3 & 82.3 & 78.2 & 59.8 & 54.6 & 72.5 & \underline{84.8} \\
\multicolumn{1}{c|}{} & ResNeXt\cite{resnext} & 66.1 & 83.4 & 75.4 & 85.7 & 65.9 & 84.8 & \textbf{91.9} & 68.4 & 68.8 & 87.6 & 80.6 & 63.2 & 63.2 & 76.4 & \underline{86.0} \\
\multicolumn{1}{c|}{} & WRN-28-10-drop\cite{wrn} & 55.8 & 75.5 & 67.3 & 81.0 & 55.3 & 81.9 & \textbf{84.6} & 60.9 & 67.3 & 84.3 & 79.8 & 60.5 & 55.6 & 74.3 & \underline{82.1} \\
\multicolumn{1}{c|}{} & GoogleNet\cite{googlenet} & 55.5 & 76.0 & 65.9 & 80.8 & 55.2 & 85.2 & \textbf{91.4} & 60.3 & 66.5 & 84.0 & 80.6 & 60.0 & 55.1 & 73.5 & \underline{86.2} \\
\multicolumn{1}{c|}{} & MobileNetv2\cite{mobilenetv2} & 52.4 & 74.6 & 63.2 & 79.1 & 52.4 & 81.0 & \textbf{84.6} & 57.0 & 65.0 & 81.7 & 80.0 & 59.5 & 53.9 & 72.1 & \underline{82.0} \\
\multicolumn{1}{c|}{} & PreResNet\cite{preresnet} & 73.2 & 87.3 & 83.1 & 92.5 & 73.0 & 89.9 & \textbf{92.6} & 74.1 & 73.6 & 92.1 & 85.4 & 68.0 & 72.3 & 79.4 & \underline{90.6} \\
\multicolumn{1}{c|}{\multirow{-8}{*}{\textbf{ConvNet Set}}} & \cellcolor[HTML]{D9D9D9}\textit{\textbf{Average}} & \cellcolor[HTML]{D9D9D9}51.4 & \cellcolor[HTML]{D9D9D9}67.9 & \cellcolor[HTML]{D9D9D9}60.4 & \cellcolor[HTML]{D9D9D9}71.8 & \cellcolor[HTML]{D9D9D9}51.2 & \cellcolor[HTML]{D9D9D9}73.0 & \cellcolor[HTML]{D9D9D9}\textbf{76.6} & \cellcolor[HTML]{D9D9D9}54.4 & \cellcolor[HTML]{D9D9D9}58.5 & \cellcolor[HTML]{D9D9D9}73.7 & \cellcolor[HTML]{D9D9D9}70.0 & \cellcolor[HTML]{D9D9D9}53.9 & \cellcolor[HTML]{D9D9D9}51.1 & \cellcolor[HTML]{D9D9D9}64.6 & \cellcolor[HTML]{D9D9D9}\underline{74.3} \\
\midrule
\multicolumn{1}{c|}{} & ViT-T\cite{vit} & 33.3 & \underline{57.1} & 40.3 & 56.1 & 33.1 & 47.6 & \textbf{58.2} & 36.7 & 44.9 & 56.4 & 42.1 & 44.4 & 33.1 & 51.9 & 52.9 \\
\multicolumn{1}{c|}{} & Swin-S\cite{swin} & 36.9 & 14.7 & 44.2 & \underline{63.2} & 36.9 & 57.1 & \textbf{66.9} & 12.7 & 46.5 & 14.9 & 51.1 & 47.3 & 38.0 & 14.1 & 58.3 \\
\multicolumn{1}{c|}{} & Swin-B\cite{swin} & 33.0 & 14.3 & 39.3 & \underline{58.2} & 33.0 & 54.3 & \textbf{59.6} & 12.2 & 45.1 & 14.2 & 46.7 & 45.7 & 35.6 & 13.8 & 54.8 \\
\multicolumn{1}{c|}{} & DeiT-T\cite{deit} & 34.6 & 59.5 & 41.0 & 54.5 & 34.7 & 51.7 & \textbf{61.3} & 38.7 & 44.5 & \underline{59.7} & 47.1 & 45.4 & 33.6 & 53.2 & 53.1 \\
\multicolumn{1}{c|}{} & DeiT-B\cite{deit} & 33.8 & 62.2 & 40.7 & \underline{59.2} & 34.5 & 48.6 & \textbf{60.8} & 36.1 & 45.4 & 58.9 & 42.1 & 43.9 & 34.4 & 53.0 & 52.2 \\
\multicolumn{1}{c|}{\multirow{-6}{*}{\textbf{Metaformer Set}}} & \cellcolor[HTML]{D9D9D9}\textit{\textbf{Average}} & \cellcolor[HTML]{D9D9D9}34.3 & \cellcolor[HTML]{D9D9D9}41.6 & \cellcolor[HTML]{D9D9D9}41.1 & \cellcolor[HTML]{D9D9D9}\underline{58.2} & \cellcolor[HTML]{D9D9D9}34.4 & \cellcolor[HTML]{D9D9D9}51.9 & \cellcolor[HTML]{D9D9D9}\textbf{61.4} & \cellcolor[HTML]{D9D9D9}27.3 & \cellcolor[HTML]{D9D9D9}45.3 & \cellcolor[HTML]{D9D9D9}40.8 & \cellcolor[HTML]{D9D9D9}45.8 & \cellcolor[HTML]{D9D9D9}45.3 & \cellcolor[HTML]{D9D9D9}34.9 & \cellcolor[HTML]{D9D9D9}37.2 & \cellcolor[HTML]{D9D9D9}54.3 \\
\midrule
\multicolumn{1}{c|}{} & PreActResNet-18\cite{andriushchenko2020understanding} & 3.0 & 3.2 & 2.9 & 3.1 & 3.0 & 3.6 & 3.0 & 3.2 & 3.1 & 2.9 & 3.6 & \underline{4.8} & 2.6 & 3.3 & \textbf{6.4} \\
\multicolumn{1}{c|}{} & WideResNet-28-10\cite{carmon2019unlabeled} & 2.2 & 2.4 & 2.0 & 2.2 & 2.2 & 2.6 & 2.1 & 2.4 & 2.1 & 2.0 & 2.4 & \underline{4.6} & 1.7 & 2.5 & \textbf{6.1} \\
\multicolumn{1}{c|}{} & WideResNet-28-10\cite{sehwag2020hydra}  & 2.4 & 2.5 & 2.2 & 2.5 & 2.4 & 2.9 & 2.2 & 2.5 & 2.4 & 2.3 & 2.8 & \underline{4.7} & 1.8 & 2.7 & \textbf{6.4} \\
\multicolumn{1}{c|}{} & WideResNet-28-10\cite{wang2019improving} & 2.4 & 2.5 & 2.2 & 2.3 & 2.3 & 2.7 & 2.2 & 2.5 & 2.3 & 2.2 & 2.7 & \underline{4.5} & 1.8 & 2.5 & \textbf{6.0} \\
\multicolumn{1}{c|}{} & WideResNet-34-20\cite{rice2020overfitting} & 3.0 & 3.2 & 2.9 & 3.1 & 3.0 & 3.4 & 3.0 & 3.2 & 2.9 & 2.9 & 3.4 & \underline{5.1} & 2.5 & 3.3 & \textbf{6.5} \\
\multicolumn{1}{c|}{} & WideResNet-34-10\cite{zhang2019theoretically} & 2.7 & 2.9 & 2.6 & 2.8 & 2.7 & 3.2 & 2.8 & 2.9 & 2.7 & 2.6 & 3.1 & \underline{4.9} & 2.3 & 3.0 & \textbf{6.5} \\
\multicolumn{1}{c|}{} & ResNet-50\cite{engstrom2019adversarial} & 2.4 & 2.6 & 2.2 & 2.5 & 2.4 & 2.9 & 2.4 & 2.7 & 2.4 & 2.3 & 2.9 & \underline{4.9} & 2.0 & 2.7 & \textbf{7.0} \\
\multicolumn{1}{c|}{} & WideResNet-34-10\cite{huang2020self} & 2.7 & 3.0 & 2.6 & 2.8 & 2.8 & 3.1 & 2.7 & 3.0 & 2.7 & 2.7 & 3.1 & \underline{4.8} & 2.4 & 3.0 & \textbf{6.3} \\
\multicolumn{1}{c|}{\multirow{-9}{*}{\textbf{ConvNet(AT) Set}}} & \cellcolor[HTML]{D9D9D9}\textit{\textbf{Average}} & \cellcolor[HTML]{D9D9D9}2.6 & \cellcolor[HTML]{D9D9D9}2.8 & \cellcolor[HTML]{D9D9D9}2.4 & \cellcolor[HTML]{D9D9D9}2.6 & \cellcolor[HTML]{D9D9D9}2.6 & \cellcolor[HTML]{D9D9D9}3.0 & \cellcolor[HTML]{D9D9D9}2.5 & \cellcolor[HTML]{D9D9D9}2.8 & \cellcolor[HTML]{D9D9D9}2.5 & \cellcolor[HTML]{D9D9D9}2.4 & \cellcolor[HTML]{D9D9D9}3.0 & \cellcolor[HTML]{D9D9D9}\underline{4.8} & \cellcolor[HTML]{D9D9D9}2.1 & \cellcolor[HTML]{D9D9D9}2.9 & \cellcolor[HTML]{D9D9D9}\textbf{6.4} \\
\rowcolor[HTML]{D9D9D9}
\midrule
\multicolumn{2}{c|}{\cellcolor[HTML]{D9D9D9}\textit{\textbf{Overall Average}}} & \cellcolor[HTML]{D9D9D9}27.6 & \cellcolor[HTML]{D9D9D9}35.3 & \cellcolor[HTML]{D9D9D9}32.4 & \cellcolor[HTML]{D9D9D9}40.8 & \cellcolor[HTML]{D9D9D9}27.6 & \cellcolor[HTML]{D9D9D9}39.7 & \cellcolor[HTML]{D9D9D9}\textbf{43.2} & \cellcolor[HTML]{D9D9D9}27.0 & \cellcolor[HTML]{D9D9D9}32.8 & \cellcolor[HTML]{D9D9D9}37.0 & \cellcolor[HTML]{D9D9D9}37.2 & \cellcolor[HTML]{D9D9D9}32.1 & \cellcolor[HTML]{D9D9D9}27.5 & \cellcolor[HTML]{D9D9D9}33.1 & \cellcolor[HTML]{D9D9D9}\underline{42.1}
\\
\bottomrule
\end{tabular}%
\end{threeparttable}}
\end{table*}

\subsection{Composition with Input-Transformation\bluetwo{-}Based Attacks}
Following the combination experiments on ImageNet, we conduct similar tests on CIFAR-10 by integrating DRAP with several well-known input-transformation\bluetwo{-}based attacks and optimization-based attacks. Due to computational constraints, we omit the combination of PGN and SIA. The experimental setup follows the untargeted protocol as described earlier.
Table \ref{cifar_composite} summarizes the results on CIFAR-10. Consistent with the findings on ImageNet, integrating DRAP with input-transformation\bluetwo{-}based attacks significantly enhances their base performance, leading to superior attack success rates across both normal and adversarially trained models. Notably, our combined method achieves the highest attack performance in every combination. It also outperforms all other methods on the most challenging ConvNet(AT) Set. Furthermore, it demonstrates competitive performance on the ConvNet Set and Metaformer Set, maintaining a strong balance across different target model categories.
These results validate the scalability and adaptability of DRAP on CIFAR-10. Not only does it excel as a standalone approach, but it also amplifies the effectiveness of input-\bluetwo{transformation-based} methods, confirming its potential for achieving state-of-the-art adversarial transferability.

\begin{table}[h]
\centering
\caption{Attack success rates on CIFAR-10 dataset (\%, $\uparrow$) of MI, PI, VT, RAP, PGN, CWA, and DRAP when integrated with DI, TI, Admix, and SIA, respectively. The indentation denotes combination. The results are averaged on each model set.}
\label{cifar_composite}
\scalebox{0.85}{
\resizebox{\columnwidth}{!}{%
\begin{tabular}{@{}l|ccc|c@{}}
\toprule
\multicolumn{1}{c|}{\textbf{Attack}} & \textbf{\begin{tabular}[c]{@{}c@{}}ConvNet\\ Set\end{tabular}} & \textbf{\begin{tabular}[c]{@{}c@{}}Metaformer\\ Set\end{tabular}} & \textbf{\begin{tabular}[c]{@{}c@{}}CNN (AT)\\ Set\end{tabular}} & \textit{\textbf{\begin{tabular}[c]{@{}c@{}}Overall\\ Average\end{tabular}}} \\ \midrule 
DI2-FGSM & 86.7 & 88.8 & 18.6 & 60.0 \\
\quad + MI & 88.7 & 91.7 & 19.8 & 61.9 \\
\quad + PI & 89.0 & 91.9 & 19.9 & 62.1 \\
\quad + VT & 89.2 & 94.3 & 19.6 & 62.6 \\
\quad + RAP & 86.3 & 82.0 & 19.3 & 58.4 \\
\quad + PGN & \textbf{89.8} & \textbf{92.3} & 20.1 & 62.6 \\
\quad + CWA & 88.0 & 86.1 & 26.8 & 63.0 \\
\textbf{\quad + DRAP} & 89.0 & 83.2 & \textbf{30.0} & \textbf{63.9} \\ \midrule
TI-FGSM & 81.9 & 79.6 & 18.0 & 55.8 \\
\quad + MI & 83.7 & 81.7 & 18.7 & 57.2 \\
\quad + PI & 84.6 & 78.9 & 18.8 & 56.9 \\
\quad + VT & 87.1 & 89.2 & 18.6 & 60.2 \\
\quad + RAP & 85.7 & 74.6 & 18.5 & 56.1 \\
\quad + PGN & \textbf{89.9} & \textbf{93.0} & 20.1 & 62.7 \\
\quad + CWA & 82.7 & 70.2 & 24.8 & 56.4 \\
\textbf{\quad + DRAP} & 89.5 & 81.5 & \textbf{29.4} & \textbf{63.5} \\ \midrule
Admix & 86.9 & 87.5 & 18.3 & 59.6 \\
\quad + MI & 87.8 & 88.8 & 19.0 & 60.5 \\
\quad + PI & 87.2 & 87.2 & 19.0 & 59.9 \\
\quad + VT & 88.9 & \textbf{93.2} & 18.9 & 62.0 \\
\quad + RAP & 89.3 & 81.7 & 18.7 & 59.1 \\
\quad + PGN & \textbf{89.8} & 92.3 & 20.0 & 62.5 \\
\quad + CWA & 87.6 & 81.9 & 25.9 & 61.5 \\
\textbf{\quad + DRAP} & 90.0 & 83.1 & \textbf{29.0} & \textbf{63.9} \\ \midrule
SIA & 89.0 & 94.2 & 18.3 & 62.0 \\
\quad + MI & 90.5 & \textbf{96.5} & 19.5 & 63.6 \\
\quad + PI & 90.4 & 96.3 & 19.5 & 63.5 \\
\quad + VT & 89.9 & 95.2 & 19.3 & 63.0 \\
\quad + RAP & 87.7 & 88.4 & 19.0 & 60.4 \\
\quad + CWA & 89.6 & 88.7 & 26.2 & 64.0 \\
\textbf{\quad + DRAP} & \textbf{90.5} & 84.5 & \textbf{28.2} & \textbf{64.2} \\ \bottomrule
\end{tabular}%
}
}
\end{table}

\section{Implementation Details}

\begin{algorithm}[t]\label{algo:flat_rap}
	\caption{Flat-RAP algorithm} 
            \begin{algorithmic}[1]
            \STATE \textbf{Require}: benign data $(\bm{x}, y)$, perturbation budget $\gamma$, surrogate model distributions $\left\{P_{\mathcal{S}_i}\right\}_{i=1}^{I}$, number of samples within one distribution $n$, late start iteration number $n_{LS}$, inner iteration number $T$, step size $\beta,\alpha$, decay factor $\mu$.
        \STATE Initialize $\hat{\bm{x}}_0 \leftarrow \bm{x},\bm{m}\leftarrow 0$;
		\FOR {$j=0,...,K-1$}
                \FOR {$i=0,...,I-1$}
                \STATE Sample a surrogate model $\w_i$ from $P_{\mathcal{S}_i}$;
                \ENDFOR
                
                \IF{$j \geq n_{LS}$}
                \STATE \# Inner maximization
                \STATE Initialize $\bm{\epsilon} \leftarrow 0$;
                \FOR {$t=0,...,T-1$}
                \STATE Calculate $\bm{g}=\nabla_{\bm{\epsilon}}\ell\left(\frac{1}{I}\sum_{i=0}^{I-1}f\left(\hat{\bm{x}}_j+\bm{\epsilon}, \w_i\right), y\right)$;
                \STATE Update $\bm{\epsilon}=\bm{\epsilon}+\beta \cdot \operatorname{sign}(\bm{g})$;
                \ENDFOR
                \ENDIF
                
                \STATE \# Outer minimization
                \STATE Calculate $\bm{g}=\nabla_{\hat{\bm{x}}}\ell\left(\frac{1}{I}\sum_{i=0}^{I-1}f\left(\hat{\bm{x}}_j+\bm{\epsilon}, \w_i\right), y\right) ;$
                \STATE Update momentum by ${\bm{m}}=\mu \cdot {\bm{m}}+\frac{\bm{g}}{\|\bm{g}\|_1}$;
                \STATE Update $\hat{\bm{x}}_{j+1}$ by $\hat{\bm{x}}_{j+1}=\Pi_{\gamma}\left(\hat{\bm{x}}_j-\alpha \cdot {\operatorname{sign}{(\bm{m})}}\right)$;
		\ENDFOR
            \STATE \textbf{Return} $\hat{\bm{x}}_{K}$.
	\end{algorithmic} 
\end{algorithm}

\begin{algorithm}[t]\label{algo:flat_cwa}
	\caption{Flat-CWA algorithm} 
            \begin{algorithmic}[1]
            \STATE \textbf{Require}: benign data $(x, y)$, perturbation budget $\gamma$, surrogate model distributions $\left\{P_{\mathcal{S}_i}\right\}_{i=1}^{I}$, number of samples within one distribution $n$, inner iteration number $T$, step size $r, \beta,\alpha$, decay factor $\mu$.
        \STATE Initialize $\hat{\bm{x}}_0 \leftarrow \bm{x},\bm{m}\leftarrow 0$;
		\FOR {$j=0,...,K-1$}
                \FOR {$i=0,...,I-1$}
                \STATE Sample a surrogate model $\w_i$ from $P_{\mathcal{S}_i}$;
                \ENDFOR
                
                \STATE \# Inner maximization
                \STATE Calculate $\bm{g}=\nabla_{\bm{x}} \ell\left( \frac{1}{I}\sum_{i=0}^{I-1} f\left(\hat{\bm{x}}_j, \bm{w}_i\right), y \right)$;
                \STATE Update $\hat{\bm{x}}_j$ by $\hat{\bm{x}}_j^0=\Pi_{\gamma}\left(\hat{\bm{x}}_j+r \cdot \operatorname{sign}(\bm{g})\right)$;
                
                \STATE \# Outer minimization
                \FOR {$i=0,...,I-1$}
				    \STATE Calculate ${g}=\nabla_{\bm{x}}\ell\left(f\left(\hat{\bm{x}}_j^{i}, \w_i\right), y\right) ;$
                    \STATE Update momentum by ${\bm{m}}=\mu \cdot {\bm{m}}+\frac{\bm{g}}{\|\bm{g}\|_2}$;
                    \STATE Update $\hat{\bm{x}}_j^{i+1}$ by $\hat{\bm{x}}_j^{i+1}=\Pi_{\gamma}\left(\hat{\bm{x}}_j^{i}-\beta \cdot {\bm{m}}\right)$;
			    \ENDFOR
                
			    \STATE Calculate the update $\bm{g}=\hat{\bm{x}}_j^{I}-\hat{\bm{x}}_j$;
                \STATE update $\hat{\bm{x}}_{j+1}$ by $\hat{\bm{x}}_{j+1}=\Pi_{\gamma}\left(\hat{\bm{x}}_j+\alpha \cdot \operatorname{sign}(\bm{g})\right)$;
		\ENDFOR
            \STATE \textbf{Return} $\hat{\bm{x}}_{K}$.
	\end{algorithmic} 
\end{algorithm}

\subsection{\blue{Runtime and Memory Measurement Setup}}\label{app: runtime_setting}
\blue{All runtime and memory experiments reported in Figure \ref{fig:iter}(b) and Table \ref{table:memory} are conducted on a server equipped with 8 NVIDIA GeForce RTX 3090, each with 24 GB of memory. When running attack algorithms to generate AEs, the mini-batch size is fixed to 200 for all methods, except for PGN, for which we use a smaller batch size of 80 to avoid out-of-memory errors. All other hyper-parameters strictly follow the untargeted ImageNet protocol described in \bluetwo{Section} \ref{Sec: Main Results}.}

\subsection{Flatness Optimization Algorithms}\label{app: flat}
The details of using strategies from RAP and CWA to optimize the loss sharpness over diverse surrogate models are shown in Algorithm \ref{algo:flat_rap} and Algorithm \ref{algo:flat_cwa}, respectively. We stick to their original algorithms but substitute the surrogate models seen at each iteration $j$ with a batch of models from the diverse model set. Specifically, we sample one model weight from every model architecture to compose one batch. When fusing the gradients of multiple models, we use the logit\bluetwo{-}ensemble strategy as suggested in RAP and CWA. For hyper-parameters, we set perturbation budget $\gamma=4/255$ with step size $\alpha=2/255$, $n=40$ and follow the optimal settings reported in their papers to set their own hyper-parameters. In Flat-RAP, we set step size $\beta=\alpha$ and $n_{LS}=5$, which is \bluetwo{the} same as in DRAP. In Flat-CWA, we set decay factor $\mu=1$, step sizes $\beta=50$, $r = \gamma/15$.

\section{Additional Experiments on ImageNet}

\begin{table}[t]
\centering
\caption{\blue{Untargeted attack success rates (\%,$\uparrow$) on ImageNet dataset (continuation of Table \ref{tab: main}). $^*$ denotes surrogate models are ResNet-50 and ResNet-50(AT). $^\ddagger$ denotes surrogate models are ResNet-50 and ViT-B.}}
\bluetable
\label{tab: main_continued}
\resizebox{0.5\textwidth}{!}{%
\begin{threeparttable}
\begin{tabular}{@{}cc|llll|ll@{}}
\toprule
\multicolumn{2}{c|}{\textbf{Target Model Set}} & \begin{tabular}[c]{@{}c@{}}FIA$^{*}$\\ \cite{fia}\end{tabular} & \begin{tabular}[c]{@{}c@{}}NAA$^{*}$\\ \cite{naa}\end{tabular} & \begin{tabular}[c]{@{}c@{}}GhostNet$^{*}$\\ \cite{ghostnet}\end{tabular} & \multicolumn{1}{c|}{\textbf{DRAP$^{*}$}} & \begin{tabular}[c]{@{}c@{}}Bayesian$^{\ddagger}$\\ \cite{bayesian}\end{tabular} & \multicolumn{1}{c}{\textbf{DRAP$^{\ddagger}$}} \\ \midrule
\multicolumn{1}{c|}{} & AlexNet\cite{alexnet} & 46.7 & {\ul 47.8} & 42.1 & \textbf{66.0} & 55.3 & \textbf{63.6} \\
\multicolumn{1}{c|}{} & VGG-16-BN\cite{vgg} & 42.3 & {\ul 60.9} & 40.3 & \textbf{76.2} & 73.0 & \textbf{82.3} \\
\multicolumn{1}{c|}{} & DenseNet-201\cite{densenet} & 35.3 & {\ul 57.9} & 32.1 & \textbf{78.5} & 72.0 & \textbf{82.1} \\
\multicolumn{1}{c|}{} & GoogLeNet\cite{googlenet} & 31.0 & {\ul 47.3} & 27.3 & \textbf{75.6} & 63.3 & \textbf{75.2} \\
\multicolumn{1}{c|}{} & ShuffleNetV2\cite{shufflenet} & 42.2 & {\ul 51.9} & 37.0 & \textbf{82.4} & 74.8 & \textbf{84.1} \\
\multicolumn{1}{c|}{} & MobileNetV2\cite{mobilenetv2} & 42.5 & {\ul 59.5} & 36.6 & \textbf{84.4} & 78.5 & \textbf{87.2} \\
\multicolumn{1}{c|}{} & MobileNetV3-L\cite{mobilenetv3} & 31.4 & {\ul 43.5} & 24.8 & \textbf{74.8} & 69.0 & \textbf{79.5} \\
\multicolumn{1}{c|}{} & MNASNet\cite{mnasnet} & 39.0 & {\ul 55.6} & 33.6 & \textbf{84.0} & 76.4 & \textbf{86.3} \\
\multicolumn{1}{c|}{} & EfficientNet\cite{efficientnet} & 29.6 & {\ul 41.7} & 25.0 & \textbf{60.4} & 53.8 & \textbf{64.7} \\
\multicolumn{1}{c|}{} & ConvNeXt-L\cite{convnet} & 8.0 & {\ul 18.7} & 7.3 & \textbf{17.0} & 18.6 & \textbf{26.2} \\
\multicolumn{1}{c|}{\multirow{-11}{*}{\textbf{\begin{tabular}[c]{@{}c@{}}ConvNet\\ Set\end{tabular}}}} & \cellcolor[HTML]{D9D9D9}\textit{\textbf{Average}} & \cellcolor[HTML]{D9D9D9}34.8 & \cellcolor[HTML]{D9D9D9}{\ul 48.5} & \cellcolor[HTML]{D9D9D9}30.6 & \cellcolor[HTML]{D9D9D9}\textbf{69.9} & \cellcolor[HTML]{D9D9D9}63.5 & \cellcolor[HTML]{D9D9D9}\textbf{73.1} \\ \midrule
\multicolumn{1}{c|}{} & ViT-S\cite{vit} & 8.5 & {\ul 12.7} & 7.0 & \textbf{18.0} & 27.0 & \textbf{40.1} \\
\multicolumn{1}{c|}{} & DeiT-S\cite{deit} & 9.8 & {\ul 14.7} & 9.0 & \textbf{23.5} & 43.8 & \textbf{55.8} \\
\multicolumn{1}{c|}{} & PoolFormer-S\cite{poolformer} & 17.4 & {\ul 33.7} & 16.4 & \textbf{48.6} & 48.1 & \textbf{60.1} \\
\multicolumn{1}{c|}{} & TNT-S\cite{tnt} & 8.0 & {\ul 16.0} & 7.8 & \textbf{21.1} & 33.5 & \textbf{44.2} \\
\multicolumn{1}{c|}{} & Swin-S\cite{swin} & 6.6 & {\ul 12.0} & 5.4 & \textbf{12.5} & 13.6 & \textbf{19.3} \\
\multicolumn{1}{c|}{} & XCiT-S\cite{xcit} & 9.2 & {\ul 13.2} & 8.6 & \textbf{15.6} & 15.3 & \textbf{21.8} \\
\multicolumn{1}{c|}{} & CaiT-S\cite{cait} & 4.1 & {\ul 7.0} & 3.6 & \textbf{10.5} & 11.1 & \textbf{17.7} \\
\multicolumn{1}{c|}{\multirow{-8}{*}{\textbf{\begin{tabular}[c]{@{}c@{}}Metaformer\\ Set\end{tabular}}}} & \cellcolor[HTML]{D9D9D9}\textit{\textbf{Average}} & \cellcolor[HTML]{D9D9D9}9.1 & \cellcolor[HTML]{D9D9D9}{\ul 15.6} & \cellcolor[HTML]{D9D9D9}8.3 & \cellcolor[HTML]{D9D9D9}\textbf{21.4} & \cellcolor[HTML]{D9D9D9}27.5 & \cellcolor[HTML]{D9D9D9}\textbf{37.0} \\ \midrule
\multicolumn{1}{c|}{} & RaWideResNet-101-2\cite{peng2023robust} & {\ul 17.3} & 17.0 & 15.5 & \textbf{24.6} & 16.9 & \textbf{17.3} \\
\multicolumn{1}{c|}{} & WideResNet-50-2\cite{salman2020adversarially} & {\ul 22.9} & 21.7 & 20.5 & \textbf{33.0} & 21.4 & \textbf{22.9} \\
\multicolumn{1}{c|}{} & ResNet-50\cite{wong2020fast} & {\ul 40.8} & 40.4 & 39.1 & \textbf{51.3} & 39.9 & \textbf{41.2} \\
\multicolumn{1}{c|}{} & ConvNeXt-L\cite{liu2024comprehensive} & {\ul 10.6} & 10.5 & 9.9 & \textbf{15.6} & 10.4 & \textbf{11.1} \\
\multicolumn{1}{c|}{} & ConvNeXt-B\cite{liu2024comprehensive} & {\ul 11.0} & 10.4 & 9.5 & \textbf{16.3} & 10.1 & \textbf{11.2} \\
\multicolumn{1}{c|}{} & ConvNeXt-L-ConvStem\cite{singh2024revisiting} & {\ul 10.0} & 9.9 & 9.7 & \textbf{14.4} & 10.2 & \textbf{10.8} \\
\multicolumn{1}{c|}{} & ConvNeXt-B-ConvStem\cite{singh2024revisiting} & {\ul 11.2} & 11.1 & 10.6 & \textbf{17.9} & 11.3 & \textbf{11.8} \\
\multicolumn{1}{c|}{} & Inc-v3$_{\textit{ens3}}$\cite{tf} & 8.5 & {\ul 11.2} & 7.9 & \textbf{18.4} & 11.7 & \textbf{17.2} \\
\multicolumn{1}{c|}{} & Inc-v3$_{\textit{ens4}}$\cite{tf} & 10.9 & {\ul 14.0} & 9.6 & \textbf{21.3} & 15.1 & \textbf{20.4} \\
\multicolumn{1}{c|}{} & IncRes-v2$_{\textit{ens}}$\cite{tf} & 3.8 & {\ul 5.8} & 3.0 & \textbf{10.0} & 6.0 & \textbf{9.2} \\
\multicolumn{1}{c|}{\multirow{-11}{*}{\textbf{\begin{tabular}[c]{@{}c@{}}ConvNet(AT)\\ Set\end{tabular}}}} & \cellcolor[HTML]{D9D9D9}\textit{\textbf{Average}} & \cellcolor[HTML]{D9D9D9}14.7 & \cellcolor[HTML]{D9D9D9}{\ul 15.2} & \cellcolor[HTML]{D9D9D9}13.5 & \cellcolor[HTML]{D9D9D9}\textbf{22.3} & \cellcolor[HTML]{D9D9D9}15.3 & \cellcolor[HTML]{D9D9D9}\textbf{17.3} \\ \midrule
\multicolumn{1}{c|}{} & Swin-B\cite{liu2024comprehensive} & 11.3 & {\ul 11.6} & 10.8 & \textbf{15.6} & 11.6 & \textbf{12.3} \\
\multicolumn{1}{c|}{} & Swin-L\cite{liu2024comprehensive} & {\ul 9.8} & 9.6 & 9.3 & \textbf{13.4} & 9.9 & \textbf{10.5} \\
\multicolumn{1}{c|}{} & XCiT-L\cite{debenedetti2023light} & {\ul 15.2} & 14.3 & 13.6 & \textbf{23.0} & 14.8 & \textbf{15.7} \\
\multicolumn{1}{c|}{} & ViT-B-ConvStem\cite{singh2024revisiting} & {\ul 11.4} & 11.2 & 10.7 & \textbf{15.9} & 11.2 & \textbf{11.7} \\
\multicolumn{1}{c|}{\multirow{-5}{*}{\textbf{\begin{tabular}[c]{@{}c@{}}Metaformer(AT)\\ Set\end{tabular}}}} & \cellcolor[HTML]{D9D9D9}\textit{\textbf{Average}} & \cellcolor[HTML]{D9D9D9}{\ul 11.9} & \cellcolor[HTML]{D9D9D9}11.7 & \cellcolor[HTML]{D9D9D9}11.1 & \cellcolor[HTML]{D9D9D9}\textbf{17.0} & \cellcolor[HTML]{D9D9D9}11.9 & \cellcolor[HTML]{D9D9D9}\textbf{12.6} \\ \midrule
\rowcolor[HTML]{D9D9D9} 
\multicolumn{2}{c|}{\cellcolor[HTML]{D9D9D9}\textit{\textbf{Overall Average}}} & 19.6 & {\ul 25.6} & 17.5 & \textbf{36.8} & 33.1 & \textbf{39.1} \\ \bottomrule
\end{tabular}
\end{threeparttable}}
\end{table}

\subsection{\blue{Comparison to Feature-based and Model-tuning-based Attacks}}\label{app: compare_f_m}

\blue{Table \ref{tab: main_continued} is a continuation of Table \ref{tab: main} and further compares DRAP with feature-based attacks FIA \cite{fia} and NAA \cite{naa}) and model-tuning-based attacks GhostNet \cite{ghostnet} and Bayesian attack \cite{bayesian}). As shown in Table \ref{tab: main_continued}, DRAP outperforms these additional baselines.}

\blue{In addition, we specify the experimental protocols of these comparisons in this subsection. Different from input-transformation-based, optimization-based and model-fusion-based attacks, feature-based and model-tuning-based attacks are typically tightly coupled to the surrogate architecture. Reproducing these methods under our main surrogate set of five architectures used in Table \ref{tab: main} would require substantial architecture-specific choices of hyper-parameters that are not specified in original papers, as well as non-trivial additional training cost from fine-tuning surrogate models. Therefore, to ensure faithful reproduction while keeping the overhead tractable, we restrict the surrogate architectures to two models when comparing with these baselines.}

\blue{For FIA, NAA and GhostNet, we use ResNet-50 \cite{resnet} and ResNet-50(AT) \cite{resnetat} as surrogates, indicated by the superscript $^{*}$ in Table \ref{tab: main_continued}. For FIA and NAA, the architecture-specific hyper-parameter is the feature layer. Following the selection principle mentioned in the original papers ``\textit{By contrast, middle layers have well-separated class representations and they are not highly correlated to the model architecture, thus middle layers are the best choice to be attacked for better transferability. \dots Based on this conclusion, we first select a few middle layers for each source model and determine the final attacked layer according to the empirical results.}''\cite{fia}, we set the feature layer as Conv4\_6 for FIA and Conv3\_4 for NAA on both ResNet-50 and ResNet-50(AT). For GhostNet, we follow the original setting and set the magnitude of erosion $\Lambda=0.22$ for both ResNet-50 and ResNet-50(AT).}

\blue{For the Bayesian attack, we use ResNet-50 \cite{resnet} and ViT-B \cite{vit} as surrogates, indicated by the superscript $^{\ddagger}$ in Table \ref{tab: main_continued}. We directly use the fine-tuned ResNet-50 SWAG model provided in the official repository, and fine-tune a ViT-B SWAG model with a learning rate of 0.001, weight decay of $1{e}{-4}$, and set $\gamma = 0.05 /\left\|\boldsymbol{\Delta} \mathbf{w}^*\right\|_2$.}

\blue{For DRAP, we set the number of surrogate components $I=2$ to match the two-surrogate setting above. Accordingly, we set $n_{\text{iter}} = I \times n = 80$. For a fair comparison, we set the number of iterations of all compared methods in Table \ref{tab: main_continued} to 80, while keeping the rest of the settings consistent with Table \ref{tab: main}.}

\subsection{Full ImageNet Targeted Results}\label{app: full}
In Table \ref{tab:main_t_full}, we present targeted attack results on ImageNet dataset, broken down to each target \bluetwo{model,} for the methods evaluated in Table \ref{tab:main_t}.

\begin{table*}[t]
\centering
\caption{{Targeted attack success rates (\%,$\uparrow$) on ImageNet dataset.} \textbf{Bold} denotes the best results and \underline{underlined} denotes the second best results.}
\label{tab:main_t_full}
\resizebox{\textwidth}{!}{%
\begin{tabular}{@{}cc|cccccccccccccccccccc@{}}
\toprule
\multicolumn{2}{c|}{\textbf{Target Model Set}} & \begin{tabular}[c]{@{}c@{}}I-FGSM\\ \cite{ifgsm}\end{tabular} & \begin{tabular}[c]{@{}c@{}}DI2-FGSM\\ \cite{difgsm}\end{tabular} & \begin{tabular}[c]{@{}c@{}}SI-FGSM\\ \cite{nisi}\end{tabular} & \begin{tabular}[c]{@{}c@{}}Admix\\ \cite{admix}\end{tabular} & \begin{tabular}[c]{@{}c@{}}TI-FGSM\\ \cite{ti}\end{tabular} & \begin{tabular}[c]{@{}c@{}}SSA\\ \cite{ssa}\end{tabular} & \begin{tabular}[c]{@{}c@{}}SIA\\ \cite{sia}\end{tabular} & \begin{tabular}[c]{@{}c@{}}MI-FGSM\\ \cite{mifgsm}\end{tabular} & \begin{tabular}[c]{@{}c@{}}PI-FGSM\\ \cite{pifgsm}\end{tabular} & \begin{tabular}[c]{@{}c@{}}VT-FGSM\\ \cite{vt}\end{tabular} & \begin{tabular}[c]{@{}c@{}}PGN\\ \cite{pgn}\end{tabular} & \blue{\begin{tabular}[c]{@{}c@{}}TPA\\ \cite{tpa}\end{tabular}} & \blue{\begin{tabular}[c]{@{}c@{}}FEM\\ \cite{fem}\end{tabular}} & \blue{\begin{tabular}[c]{@{}c@{}}MEF\\ \cite{mef}\end{tabular}} & \begin{tabular}[c]{@{}c@{}}SVRE\\ \cite{svre}\end{tabular} & \blue{\begin{tabular}[c]{@{}c@{}}AdaEA\\ \cite{adaea}\end{tabular}} & \blue{\begin{tabular}[c]{@{}c@{}}SMER\\ \cite{smer}\end{tabular}} & \begin{tabular}[c]{@{}c@{}}CWA\\ \cite{cwa}\end{tabular} & \begin{tabular}[c]{@{}c@{}}RAP\\ \cite{rap}\end{tabular} & \textbf{DRAP} \\ \midrule
\multicolumn{1}{c|}{} & AlexNet\cite{alexnet} & 0.0 & 0.1 & 0.0 & 0.1 & 0.1 & 2.1 & 1.7 & 0.0 & 0.2 & 0.0 & 3.2 & 0.1 & 1.1 & 3.0 & 0.6 & {\ul 18.9} & 5.0 & 12.9 & 0.3 & \textbf{30.4} \\
\multicolumn{1}{c|}{} & VGG-16-BN\cite{vgg} & 9.0 & 38.1 & 13.7 & 20.1 & 13.2 & 32.9 & \textbf{93.3} & 7.0 & 13.4 & 10.9 & 50.7 & 2.1 & 6.9 & 52.5 & 16.3 & 21.5 & 18.9 & 21.9 & 29.0 & {\ul 80.0} \\
\multicolumn{1}{c|}{} & DenseNet-201\cite{densenet} & 10.4 & 48.6 & 34.1 & 24.0 & 18.4 & 50.5 & \textbf{94.7} & 13.6 & 24.8 & 14.8 & 62.7 & 3.2 & 9.2 & 60.9 & 26.4 & 47.5 & 28.8 & 53.4 & 28.4 & {\ul 90.7} \\
\multicolumn{1}{c|}{} & GoogLeNet\cite{googlenet} & 1.1 & 13.5 & 9.0 & 4.1 & 3.1 & 19.6 & {\ul 61.4} & 2.2 & 5.0 & 3.1 & 31.9 & 0.6 & 5.6 & 39.3 & 7.3 & 40.1 & 17.1 & 36.9 & 9.4 & \textbf{79.3} \\
\multicolumn{1}{c|}{} & ShuffleNetV2\cite{shufflenet} & 0.3 & 3.9 & 2.9 & 1.5 & 0.7 & 7.9 & 27.7 & 1.7 & 2.2 & 0.8 & 21.3 & 0.6 & 4.5 & 24.9 & 3.0 & {\ul 38.2} & 14.0 & 33.3 & 4.5 & \textbf{75.7} \\
\multicolumn{1}{c|}{} & MobileNetV2\cite{mobilenetv2} & 3.6 & 19.8 & 9.0 & 10.5 & 6.2 & 27.2 & {\ul 77.3} & 3.5 & 7.2 & 4.3 & 41.6 & 1.4 & 8.9 & 44.2 & 12.5 & 37.3 & 21.4 & 37.4 & 16.4 & \textbf{86.3} \\
\multicolumn{1}{c|}{} & MobileNetV3-L\cite{mobilenetv3} & 0.9 & 8.0 & 4.6 & 3.6 & 2.2 & 36.3 & {\ul 47.2} & 2.1 & 3.8 & 2.3 & 29.7 & 0.7 & 5.8 & 35.0 & 19.7 & 40.3 & 19.6 & 40.1 & 6.4 & \textbf{81.7} \\
\multicolumn{1}{c|}{} & MNASNet\cite{mnasnet} & 3.0 & 16.2 & 7.5 & 7.7 & 5.2 & 20.9 & {\ul 73.8} & 3.8 & 6.8 & 4.8 & 40.5 & 1.9 & 7.3 & 44.4 & 9.4 & 39.0 & 20.1 & 38.5 & 14.4 & \textbf{87.1} \\
\multicolumn{1}{c|}{} & EfficientNet\cite{efficientnet} & 0.9 & 7.3 & 4.3 & 1.6 & 1.1 & 12.7 & {\ul 44.5} & 0.9 & 1.8 & 1.5 & 21.5 & 0.2 & 4.2 & 31.3 & 5.7 & 34.7 & 16.2 & 31.9 & 3.3 & \textbf{72.3} \\
\multicolumn{1}{c|}{} & ConvNeXt-L\cite{convnet} & 24.6 & 64.6 & 31.3 & 39.4 & 31.8 & 87.9 & \textbf{98.7} & 17.9 & 32.1 & 26.7 & 70.7 & 4.6 & 9.7 & 74.6 & 84.4 & 28.0 & 51.4 & 53.2 & 33.1 & {\ul 92.0} \\
\multicolumn{1}{c|}{\multirow{-11}{*}{\textbf{\begin{tabular}[c]{@{}c@{}}ConvNet\\ Set\end{tabular}}}} & \cellcolor[HTML]{D9D9D9}\textit{\textbf{Average}} & \cellcolor[HTML]{D9D9D9}5.4 & \cellcolor[HTML]{D9D9D9}22.0 & \cellcolor[HTML]{D9D9D9}11.6 & \cellcolor[HTML]{D9D9D9}11.3 & \cellcolor[HTML]{D9D9D9}8.2 & \cellcolor[HTML]{D9D9D9}29.8 & \cellcolor[HTML]{D9D9D9}{\ul 62.0} & \cellcolor[HTML]{D9D9D9}5.3 & \cellcolor[HTML]{D9D9D9}9.7 & \cellcolor[HTML]{D9D9D9}6.9 & \cellcolor[HTML]{D9D9D9}37.4 & \cellcolor[HTML]{D9D9D9}1.5 & \cellcolor[HTML]{D9D9D9}6.3 & \cellcolor[HTML]{D9D9D9}41.0 & \cellcolor[HTML]{D9D9D9}18.5 & \cellcolor[HTML]{D9D9D9}34.6 & \cellcolor[HTML]{D9D9D9}21.3 & \cellcolor[HTML]{D9D9D9}36.0 & \cellcolor[HTML]{D9D9D9}14.5 & \cellcolor[HTML]{D9D9D9}\textbf{77.6} \\ \midrule
\multicolumn{1}{c|}{} & ViT-S\cite{vit} & 0.1 & 4.5 & 0.3 & 0.4 & 0.2 & 10.4 & {\ul 35.2} & 0.1 & 0.4 & 0.1 & 7.3 & 0.0 & 0.8 & 18.8 & 6.2 & 20.8 & 14.4 & 24.6 & 1.1 & \textbf{59.4} \\
\multicolumn{1}{c|}{} & DeiT-S\cite{deit} & 0.4 & 6.6 & 1.2 & 0.8 & 0.4 & 14.9 & {\ul 37.8} & 0.7 & 1.1 & 0.6 & 13.5 & 0.0 & 1.3 & 21.5 & 8.4 & 26.7 & 16.5 & 29.9 & 1.5 & \textbf{71.6} \\
\multicolumn{1}{c|}{} & PoolFormer-S\cite{poolformer} & 2.4 & 41.6 & 8.3 & 8.7 & 6.4 & 49.0 & \textbf{89.5} & 5.3 & 8.9 & 6.1 & 45.8 & 1.4 & 6.0 & 53.0 & 29.9 & 30.9 & 23.3 & 34.8 & 15.5 & {\ul 82.4} \\
\multicolumn{1}{c|}{} & TNT-S\cite{tnt} & 0.2 & 8.4 & 1.2 & 0.5 & 0.3 & 16.3 & {\ul 47.1} & 0.2 & 1.1 & 0.5 & 14.9 & 0.1 & 1.1 & 26.1 & 11.8 & 23.6 & 15.8 & 27.5 & 2.3 & \textbf{72.0} \\
\multicolumn{1}{c|}{} & Swin-S\cite{swin} & 0.0 & 4.5 & 0.3 & 0.3 & 0.2 & 6.5 & {\ul 31.7} & 0.3 & 0.3 & 0.5 & 4.0 & 0.0 & 0.3 & 12.3 & 7.8 & 8.4 & 6.9 & 12.9 & 1.0 & \textbf{37.4} \\
\multicolumn{1}{c|}{} & XCiT-S\cite{xcit} & 0.0 & 8.8 & 0.0 & 0.4 & 0.1 & 10.1 & {\ul 26.4} & 0.2 & 0.2 & 0.1 & 3.8 & 0.0 & 0.3 & 5.3 & 4.9 & 8.8 & 4.7 & 9.5 & 0.5 & \textbf{30.2} \\
\multicolumn{1}{c|}{} & CaiT-S\cite{cait} & 0.1 & 7.0 & 0.3 & 0.4 & 0.1 & 8.1 & {\ul 22.6} & 0.0 & 0.5 & 0.0 & 3.8 & 0.0 & 0.3 & 7.2 & 4.3 & 13.2 & 7.3 & 13.8 & 0.1 & \textbf{41.5} \\
\multicolumn{1}{c|}{\multirow{-8}{*}{\textbf{\begin{tabular}[c]{@{}c@{}}Metaformer\\ Set\end{tabular}}}} & \cellcolor[HTML]{D9D9D9}\textit{\textbf{Average}} & \cellcolor[HTML]{D9D9D9}0.5 & \cellcolor[HTML]{D9D9D9}11.6 & \cellcolor[HTML]{D9D9D9}1.7 & \cellcolor[HTML]{D9D9D9}1.6 & \cellcolor[HTML]{D9D9D9}1.1 & \cellcolor[HTML]{D9D9D9}16.5 & \cellcolor[HTML]{D9D9D9}{\ul 41.5} & \cellcolor[HTML]{D9D9D9}1.0 & \cellcolor[HTML]{D9D9D9}1.8 & \cellcolor[HTML]{D9D9D9}1.1 & \cellcolor[HTML]{D9D9D9}13.3 & \cellcolor[HTML]{D9D9D9}0.2 & \cellcolor[HTML]{D9D9D9}1.4 & \cellcolor[HTML]{D9D9D9}20.6 & \cellcolor[HTML]{D9D9D9}10.5 & \cellcolor[HTML]{D9D9D9}18.9 & \cellcolor[HTML]{D9D9D9}12.7 & \cellcolor[HTML]{D9D9D9}21.9 & \cellcolor[HTML]{D9D9D9}3.1 & \cellcolor[HTML]{D9D9D9}\textbf{56.4} \\ \midrule
\multicolumn{1}{c|}{} & RaWideResNet-101-2\cite{peng2023robust} & 0.0 & 0.0 & 0.0 & 0.0 & 0.0 & 0.0 & 0.0 & 0.0 & 0.0 & 0.0 & 0.0 & 0.0 & 0.0 & 0.0 & 0.0 & {\ul 5.8} & 2.2 & 3.3 & 0.1 & \textbf{7.5} \\
\multicolumn{1}{c|}{} & WideResNet-50-2\cite{salman2020adversarially} & 0.0 & 0.0 & 0.0 & 0.0 & 0.0 & 0.2 & 0.0 & 0.0 & 0.1 & 0.0 & 0.1 & 0.0 & 0.0 & 0.4 & 0.2 & {\ul 12.6} & 3.8 & 5.9 & 0.2 & \textbf{12.3} \\
\multicolumn{1}{c|}{} & ResNet-50\cite{wong2020fast} & 0.0 & 0.0 & 0.0 & 0.0 & 0.0 & 0.1 & 0.0 & 0.0 & 0.0 & 0.0 & 0.1 & 0.0 & 0.0 & 0.1 & 0.0 & {\ul 3.2} & 1.2 & 1.3 & 0.0 & \textbf{4.2} \\
\multicolumn{1}{c|}{} & ConvNeXt-L\cite{liu2024comprehensive} & 0.0 & 0.0 & 0.0 & 0.0 & 0.0 & 0.2 & 0.0 & 0.0 & 0.1 & 0.0 & 0.0 & 0.0 & 0.1 & 0.3 & 0.1 & {\ul 9.0} & 2.3 & 4.3 & 0.2 & \textbf{8.2} \\
\multicolumn{1}{c|}{} & ConvNeXt-B\cite{liu2024comprehensive} & 0.0 & 0.0 & 0.0 & 0.0 & 0.0 & 0.1 & 0.0 & 0.0 & 0.2 & 0.0 & 0.0 & 0.0 & 0.1 & 0.1 & 0.0 & {\ul 8.2} & 3.0 & 4.4 & 0.1 & \textbf{9.7} \\
\multicolumn{1}{c|}{} & ConvNeXt-L-ConvStem\cite{singh2024revisiting} & 0.0 & 0.0 & 0.0 & 0.0 & 0.0 & 0.2 & 0.0 & 0.0 & 0.1 & 0.0 & 0.0 & 0.0 & 0.1 & 0.3 & 0.1 & {\ul 8.1} & 2.0 & 3.7 & 0.2 & \textbf{8.0} \\
\multicolumn{1}{c|}{} & ConvNeXt-B-ConvStem\cite{singh2024revisiting} & 0.0 & 0.0 & 0.0 & 0.0 & 0.0 & 0.3 & 0.0 & 0.0 & 0.1 & 0.0 & 0.1 & 0.0 & 0.0 & 0.2 & 0.1 & {\ul 9.3} & 2.9 & 5.0 & 0.1 & \textbf{9.1} \\
\multicolumn{1}{c|}{} & Inc-v3$_{\textit{ens3}}$\cite{tf} & 0.0 & 0.1 & 0.1 & 0.0 & 0.0 & 0.8 & 0.3 & 0.0 & 0.0 & 0.0 & 0.2 & 0.0 & 0.1 & 3.4 & 0.4 & {\ul 25.5} & 6.2 & 13.1 & 0.0 & \textbf{26.4} \\
\multicolumn{1}{c|}{} & Inc-v3$_{\textit{ens4}}$\cite{tf} & 0.0 & 0.0 & 0.0 & 0.0 & 0.0 & 0.9 & 0.2 & 0.0 & 0.0 & 0.0 & 0.2 & 0.0 & 0.0 & 2.3 & 0.4 & {\ul 27.2} & 6.0 & 13.3 & 0.1 & \textbf{27.0} \\
\multicolumn{1}{c|}{} & IncRes-v2$_{\textit{ens}}$\cite{tf} & 0.0 & 0.0 & 0.0 & 0.0 & 0.0 & 0.3 & 0.2 & 0.0 & 0.0 & 0.0 & 0.3 & 0.0 & 0.0 & 0.5 & 0.1 & {\ul 23.7} & 5.5 & 9.2 & 0.0 & \textbf{21.6} \\
\multicolumn{1}{c|}{\multirow{-11}{*}{\textbf{\begin{tabular}[c]{@{}c@{}}ConvNet(AT)\\ Set\end{tabular}}}} & \cellcolor[HTML]{D9D9D9}\textit{\textbf{Average}} & \cellcolor[HTML]{D9D9D9}0.0 & \cellcolor[HTML]{D9D9D9}0.0 & \cellcolor[HTML]{D9D9D9}0.0 & \cellcolor[HTML]{D9D9D9}0.0 & \cellcolor[HTML]{D9D9D9}0.0 & \cellcolor[HTML]{D9D9D9}0.3 & \cellcolor[HTML]{D9D9D9}0.1 & \cellcolor[HTML]{D9D9D9}0.0 & \cellcolor[HTML]{D9D9D9}0.1 & \cellcolor[HTML]{D9D9D9}0.0 & \cellcolor[HTML]{D9D9D9}0.1 & \cellcolor[HTML]{D9D9D9}0.0 & \cellcolor[HTML]{D9D9D9}0.0 & \cellcolor[HTML]{D9D9D9}0.8 & \cellcolor[HTML]{D9D9D9}0.1 & \cellcolor[HTML]{D9D9D9}{\ul 13.3} & \cellcolor[HTML]{D9D9D9}3.5 & \cellcolor[HTML]{D9D9D9}6.4 & \cellcolor[HTML]{D9D9D9}0.1 & \cellcolor[HTML]{D9D9D9}\textbf{13.4} \\ \midrule
\multicolumn{1}{c|}{} & Swin-B\cite{liu2024comprehensive} & 0.0 & 0.0 & 0.0 & 0.0 & 0.0 & 0.1 & 0.0 & 0.0 & 0.0 & 0.0 & 0.0 & 0.0 & 0.0 & 0.1 & 0.1 & {\ul 6.9} & 2.4 & 3.5 & 0.2 & \textbf{7.0} \\
\multicolumn{1}{c|}{} & Swin-L\cite{liu2024comprehensive} & 0.0 & 0.0 & 0.0 & 0.0 & 0.0 & 0.0 & 0.0 & 0.0 & 0.0 & 0.0 & 0.0 & 0.0 & 0.0 & 0.1 & 0.0 & {\ul 8.1} & 2.5 & 4.2 & 0.2 & \textbf{8.9} \\
\multicolumn{1}{c|}{} & XCiT-L\cite{debenedetti2023light} & 0.0 & 0.0 & 0.0 & 0.0 & 0.0 & 0.2 & 0.0 & 0.1 & 0.1 & 0.0 & 0.1 & 0.0 & 0.0 & 0.3 & 0.0 & \textbf{19.1} & 4.9 & 13.1 & 0.2 & {\ul 17.6} \\
\multicolumn{1}{c|}{} & ViT-B-ConvStem\cite{singh2024revisiting} & 0.0 & 0.0 & 0.0 & 0.0 & 0.0 & 0.0 & 0.0 & 0.0 & 0.0 & 0.0 & 0.1 & 0.0 & 0.0 & 0.0 & 0.0 & {\ul 11.6} & 3.6 & 9.1 & 0.1 & \textbf{12.1} \\
\multicolumn{1}{c|}{\multirow{-5}{*}{\textbf{\begin{tabular}[c]{@{}c@{}}Metaformer(AT)\\ Set\end{tabular}}}} & \cellcolor[HTML]{D9D9D9}\textit{\textbf{Average}} & \cellcolor[HTML]{D9D9D9}0.0 & \cellcolor[HTML]{D9D9D9}0.0 & \cellcolor[HTML]{D9D9D9}0.0 & \cellcolor[HTML]{D9D9D9}0.0 & \cellcolor[HTML]{D9D9D9}0.0 & \cellcolor[HTML]{D9D9D9}0.1 & \cellcolor[HTML]{D9D9D9}0.0 & \cellcolor[HTML]{D9D9D9}0.0 & \cellcolor[HTML]{D9D9D9}0.0 & \cellcolor[HTML]{D9D9D9}0.0 & \cellcolor[HTML]{D9D9D9}0.1 & \cellcolor[HTML]{D9D9D9}0.0 & \cellcolor[HTML]{D9D9D9}0.0 & \cellcolor[HTML]{D9D9D9}0.1 & \cellcolor[HTML]{D9D9D9}0.0 & \cellcolor[HTML]{D9D9D9}\textbf{11.4} & \cellcolor[HTML]{D9D9D9}3.4 & \cellcolor[HTML]{D9D9D9}{\ul 7.5} & \cellcolor[HTML]{D9D9D9}0.2 & \cellcolor[HTML]{D9D9D9}\textbf{11.4} \\ \midrule
\rowcolor[HTML]{D9D9D9} 
\multicolumn{2}{c|}{\cellcolor[HTML]{D9D9D9}\textit{\textbf{Overall Average}}} & 1.8 & 9.7 & 4.1 & 4.0 & 2.9 & 13.4 & {\ul 29.4} & 1.9 & 3.6 & 2.5 & 15.1 & 0.5 & 2.4 & 18.1 & 8.4 & 21.2 & 11.3 & 19.5 & 5.4 & \textbf{43.5} \\ \bottomrule
\end{tabular}
}
\end{table*}

\begin{table}[!t]
\centering
\bluetable
\ContinuedFloat
\caption{Targeted attack success rates (\%,$\uparrow$) on ImageNet dataset (continued). FIA and NAA are not applicable to targeted attacks. $^*$ denotes surrogate models are ResNet-50 and ResNet-50(AT). $^\ddagger$ denotes surrogate models are ResNet-50 and ViT-B.}
\resizebox{0.5\textwidth}{!}{%
\begin{tabular}{@{}cc|cccc|cc@{}}
\toprule
\multicolumn{2}{c|}{\textbf{Target Model Set}} & \begin{tabular}[c]{@{}c@{}}FIA$^{*}$\\ \cite{fia}\end{tabular} &\begin{tabular}[c]{@{}c@{}}NAA$^{*}$\\ \cite{naa}\end{tabular} & \begin{tabular}[c]{@{}c@{}}GhostNet$^{*}$\\ \cite{ghostnet}\end{tabular} & \multicolumn{1}{c|}{\textbf{DRAP$^{*}$}} & \begin{tabular}[c]{@{}c@{}}Bayesian$^{\ddagger}$\\ \cite{bayesian}\end{tabular} & \multicolumn{1}{c}{\textbf{DRAP$^{\ddagger}$}} \\ \midrule
\multicolumn{1}{c|}{} & AlexNet\cite{alexnet} &  &  & 0.0 & 26.5 & 1.1 & 2.9 \\
\multicolumn{1}{c|}{} & VGG-16-BN\cite{vgg} &  &  & 2.4 & 71.5 & 58.9 & 56.2 \\
\multicolumn{1}{c|}{} & DenseNet-201\cite{densenet} &  &  & 4.3 & 91.1 & 73.1 & 72.5 \\
\multicolumn{1}{c|}{} & GoogLeNet\cite{googlenet} &  &  & 0.5 & 81.1 & 48.8 & 51.6 \\
\multicolumn{1}{c|}{} & ShuffleNetV2\cite{shufflenet} &  &  & 0.2 & 74.0 & 32.1 & 36.6 \\
\multicolumn{1}{c|}{} & MobileNetV2\cite{mobilenetv2} &  &  & 1.6 & 83.6 & 54.9 & 62.1 \\
\multicolumn{1}{c|}{} & MobileNetV3-L\cite{mobilenetv3} &  &  & 0.3 & 70.2 & 26.8 & 38.5 \\
\multicolumn{1}{c|}{} & MNASNet\cite{mnasnet} &  &  & 1.2 & 81.6 & 50.4 & 58.9 \\
\multicolumn{1}{c|}{} & EfficientNet\cite{efficientnet} &  &  & 0.0 & 59.0 & 23.1 & 21.7 \\
\multicolumn{1}{c|}{} & ConvNeXt-L\cite{convnet} &  &  & 0.4 & 49.6 & 28.5 & 23.4 \\
\multicolumn{1}{c|}{\multirow{-11}{*}{\textbf{\begin{tabular}[c]{@{}c@{}}ConvNet\\ Set\end{tabular}}}} & \cellcolor[HTML]{D9D9D9}\textit{\textbf{Average}} &  &  & \cellcolor[HTML]{D9D9D9}1.1 & \cellcolor[HTML]{D9D9D9}68.8 & \cellcolor[HTML]{D9D9D9}39.8 & \cellcolor[HTML]{D9D9D9}42.4 \\ \cmidrule(r){1-2} \cmidrule(l){5-8} 
\multicolumn{1}{c|}{} & ViT-S\cite{vit} &  &  & 0.0 & 25.9 & 5.9 & 15.1 \\
\multicolumn{1}{c|}{} & DeiT-S\cite{deit} &  &  & 0.0 & 34.8 & 18.2 & 30.7 \\
\multicolumn{1}{c|}{} & PoolFormer-S\cite{poolformer} &  &  & 0.3 & 65.2 & 37.3 & 40.3 \\
\multicolumn{1}{c|}{} & TNT-S\cite{tnt} &  &  & 0.0 & 33.5 & 15.0 & 21.4 \\
\multicolumn{1}{c|}{} & Swin-S\cite{swin} &  &  & 0.0 & 15.3 & 4.3 & 3.9 \\
\multicolumn{1}{c|}{} & XCiT-S\cite{xcit} &  &  & 0.0 & 13.4 & 1.9 & 2.2 \\
\multicolumn{1}{c|}{} & CaiT-S\cite{cait} &  &  & 0.0 & 20.3 & 3.4 & 3.5 \\
\multicolumn{1}{c|}{\multirow{-8}{*}{\textbf{\begin{tabular}[c]{@{}c@{}}Metaformer\\ Set\end{tabular}}}} & \cellcolor[HTML]{D9D9D9}\textit{\textbf{Average}} &  &  & \cellcolor[HTML]{D9D9D9}0.0 & \cellcolor[HTML]{D9D9D9}29.8 & \cellcolor[HTML]{D9D9D9}12.3 & \cellcolor[HTML]{D9D9D9}16.7 \\ \cmidrule(r){1-2} \cmidrule(l){5-8} 
\multicolumn{1}{c|}{} & RaWideResNet-101-2\cite{peng2023robust} &  &  & 0.0 & 4.1 & 0.0 & 0.0 \\
\multicolumn{1}{c|}{} & WideResNet-50-2\cite{salman2020adversarially} &  &  & 0.0 & 10.5 & 0.0 & 0.0 \\
\multicolumn{1}{c|}{} & ResNet-50\cite{wong2020fast} &  &  & 0.0 & 3.7 & 0.0 & 0.0 \\
\multicolumn{1}{c|}{} & ConvNeXt-L\cite{liu2024comprehensive} &  &  & 0.0 & 5.0 & 0.0 & 0.0 \\
\multicolumn{1}{c|}{} & ConvNeXt-B\cite{liu2024comprehensive} &  &  & 0.0 & 5.6 & 0.0 & 0.0 \\
\multicolumn{1}{c|}{} & ConvNeXt-L-ConvStem\cite{singh2024revisiting} &  &  & 0.0 & 4.3 & 0.0 & 0.0 \\
\multicolumn{1}{c|}{} & ConvNeXt-B-ConvStem\cite{singh2024revisiting} &  &  & 0.0 & 5.6 & 0.0 & 0.0 \\
\multicolumn{1}{c|}{} & Inc-v3$_{\textit{ens3}}$\cite{tf} &  &  & 0.0 & 20.5 & 0.0 & 0.2 \\
\multicolumn{1}{c|}{} & Inc-v3$_{\textit{ens4}}$\cite{tf} &  &  & 0.0 & 20.9 & 0.0 & 0.0 \\
\multicolumn{1}{c|}{} & IncRes-v2$_{\textit{ens}}$\cite{tf} &  &  & 0.0 & 14.2 & 0.1 & 0.1 \\
\multicolumn{1}{c|}{\multirow{-11}{*}{\textbf{\begin{tabular}[c]{@{}c@{}}ConvNet(AT)\\ Set\end{tabular}}}} & \cellcolor[HTML]{D9D9D9}\textit{\textbf{Average}} &  &  & \cellcolor[HTML]{D9D9D9}0.0 & \cellcolor[HTML]{D9D9D9}9.4 & \cellcolor[HTML]{D9D9D9}0.0 & \cellcolor[HTML]{D9D9D9}0.0 \\ \cmidrule(r){1-2} \cmidrule(l){5-8} 
\multicolumn{1}{c|}{} & Swin-B\cite{liu2024comprehensive} &  &  & 0.0 & 3.5 & 0.0 & 0.0 \\
\multicolumn{1}{c|}{} & Swin-L\cite{liu2024comprehensive} &  &  & 0.0 & 4.6 & 0.0 & 0.0 \\
\multicolumn{1}{c|}{} & XCiT-L\cite{debenedetti2023light} &  &  & 0.0 & 7.2 & 0.0 & 0.0 \\
\multicolumn{1}{c|}{} & ViT-B-ConvStem\cite{singh2024revisiting} &  &  & 0.0 & 5.4 & 0.0 & 0.0 \\
\multicolumn{1}{c|}{\multirow{-5}{*}{\textbf{\begin{tabular}[c]{@{}c@{}}Metaformer(AT)\\ Set\end{tabular}}}} & \cellcolor[HTML]{D9D9D9}\textit{\textbf{Average}} &  &  & \cellcolor[HTML]{D9D9D9}0.0 & \cellcolor[HTML]{D9D9D9}5.2 & \cellcolor[HTML]{D9D9D9}0.0 & \cellcolor[HTML]{D9D9D9}0.0 \\ \cmidrule(r){1-2} \cmidrule(l){5-8} 
\multicolumn{2}{c|}{\cellcolor[HTML]{D9D9D9}\textit{\textbf{Overall Average}}} & \multirow{-36}{*}{$\backslash$} & \multirow{-36}{*}{$\backslash$} & \cellcolor[HTML]{D9D9D9}0.4 & \cellcolor[HTML]{D9D9D9}32.6 & \cellcolor[HTML]{D9D9D9}15.6 & \cellcolor[HTML]{D9D9D9}17.5 \\ \bottomrule
\end{tabular}
}
\end{table}

\subsection{\blue{Comparison under More Settings in Untargeted Scenario}}\label{app:large_budget}

\blue{As stated in the Implementation Details of Section \ref{Sec: Main Results}, we set the default untargeted experimental protocol as $\gamma=4/255$ and $n_{iter}=200$ for non-momentum-based attacks as well as DRAP and $n_{iter}=10$ for momentum-based attacks to avoid the performance degradation with more iterations. While conventional settings typically employ a larger budget $\gamma=16/255$ and \bluetwo{fewer} number of iterations $n_{iter}=10$, our default setting could evaluate attacks under a more challenging threat model and avoid under-estimating the best achievable performance of many transfer-based attacks as well as DRAP by stopping too early. Additionally, we provide comparison results under some conventional settings to further validate the generality of DRAP.}

\begin{table*}[h]
\centering
\caption{\blue{Detection results in terms of AUROC obtained by EPS-AD \cite{zhang2023detecting} on AEs generated by different transfer-based attacks with three perturbation budgets $\gamma\in\{4/255,8/255,16/255\}$. Keeping hyper-parameters consistent with the configuration in its original paper, we train its deep‑kernel MMD on ResNet‑50 and test it on AEs crafted on an adversarially trained XCiT \cite{debenedetti2023light}. The last column reports the average AUROC across all attacks.}}
\label{tab:auroc}
\resizebox{0.8\textwidth}{!}{%
\bluetable
\begin{tabular}{@{}cl|cccccccccc|c@{}}
\toprule
\multicolumn{2}{c|}{\textbf{AUROC}} & \begin{tabular}[c]{@{}c@{}}I-FGSM\\ \cite{ifgsm}\end{tabular} & \begin{tabular}[c]{@{}c@{}}DI2-FGSM\\ \cite{difgsm}\end{tabular} & \begin{tabular}[c]{@{}c@{}}SI-FGSM\\ \cite{nisi}\end{tabular} & \begin{tabular}[c]{@{}c@{}}TI-FGSM\\ \cite{ti}\end{tabular} & \begin{tabular}[c]{@{}c@{}}SIA\\ \cite{sia}\end{tabular} & \begin{tabular}[c]{@{}c@{}}MI-FGSM\\ \cite{mifgsm}\end{tabular} & \begin{tabular}[c]{@{}c@{}}PI-FGSM\\ \cite{pifgsm}\end{tabular} &  \begin{tabular}[c]{@{}c@{}}VT-FGSM\\ \cite{vt}\end{tabular} & \begin{tabular}[c]{@{}c@{}}PGN\\ \cite{pgn}\end{tabular} & DRAP & \textit{\textbf{AVERAGE}} \\ \midrule
\multicolumn{2}{c|}{$\gamma=4/255$} & 0.9099 & 0.8478 & 0.8668 & 0.8398 & 0.7891 & 0.8731 & 0.8138 & 0.9009 & 0.8944 & 0.8842 & 0.8620 \\ \midrule
\multicolumn{2}{c|}{$\gamma=8/255$} & 0.9530 & 0.9503 & 0.9331 & 0.9227 & 0.9441 & 0.9341 & 0.9384 & 0.9500 & 0.9477 & 0.9416 & 0.9415 \\ \midrule
\multicolumn{2}{c|}{$\gamma=16/255$} & 0.9455 & 0.9730 & 0.9233 & 0.9391 & 0.9724 & 0.9355 & 0.9517 & 0.9679 & 0.9340 & 0.9060 & 0.9448 \\ \bottomrule
\end{tabular}
}
\end{table*}

\subsubsection{\blue{Justification of the default setting}}

\blue{Earlier transfer-based attacks were commonly evaluated at the perturbation norm budget $\gamma=16/255$ \cite{ti, nisi, mifgsm}, but the community has rapidly moved to more challenging budgets. Recent transfer-based attacks and benchmarks shift to $\gamma=8/255$ \cite{blackboxbench, tabench, bd, taig, lrs} and even $4/255$ \cite{bd, taig}. This shift reflects that $16/255$ has become relatively easy for strong attacks, while $4/255$ offers a more challenging threat model where improving transferability is non‑trivial. Thus, we adopt $\gamma=4/255$ as our main setting.}

\blue{On the defense side, adversarial defenses \cite{fat,song2024regional} have been extensively studied at $\gamma=8/255$ for many years and have already achieved strong performance \cite{hu2019new, moayeri2021sample,defensesurvey}. Subsequent detection methods primarily focus on $\gamma=4/255$ as the main threat model \cite{dong2022random,zhang2023detecting}. To stay consistent with this line of work, we therefore use $\gamma=4/255$ as our main budget \cite{cwa}.}

\blue{As supplementary evidence from the defense side, to validate the above statement that current detection methods have already achieved very strong performance on threat models with $\gamma=16/255$ and $8/255$, we additionally conducted a detection experiment using a recent strong detector EPS‑AD \cite{zhang2023detecting} under its most challenging setting of \textit{Detecting on Unseen and Transferable Attacks}. Table \ref{tab:auroc} reports AUROC under $\gamma \in \{4/255, 8/255, 16/255\}$. We find that, even in this challenging scenario, the detector achieves an average AUROC of 0.94 for $\gamma=16/255$ and $8/255$, whereas the AUROC is notably lower for $\gamma=4/255$. This indicates that AEs with a large budget are highly detectable with the current detector, while AEs with a small budget remain more challenging. This observation motivates our focus on $\gamma=4/255$.}

\begin{table*}[t]
\centering
\caption{\blue{Untargeted attack success rates (\%, $\uparrow$) on ImageNet dataset under perturbation budget $\gamma = 16/255$. The AEs are crafted from the five surrogate models (ResNet-50, ConvNeXt-T, ViT-B, ResNet-50(AT), and XCiT-S(AT)), against 31 target models from four prototypes (normally and adversarially trained convnets and metaformers). All baseline attacks use 10 iterations, while DRAP is additionally evaluated with 50, 100, 150, and 200 iterations. \textbf{Bold} denotes the best results and \underline{underlined} denotes the second best results.}}
\label{tab:main_16_255_1}
\resizebox{\textwidth}{!}{%
\bluetable
\begin{tabular}{cc|ccccccccccccccccccc|ccccc}
\toprule
\multicolumn{2}{c|}{\textbf{Target Model Set}} & \begin{tabular}[c]{@{}c@{}}I-FGSM\\ \cite{ifgsm}\end{tabular} & \begin{tabular}[c]{@{}c@{}}DI2\\ \cite{difgsm}\end{tabular} & \begin{tabular}[c]{@{}c@{}}SI\\ \cite{nisi}\end{tabular} & \begin{tabular}[c]{@{}c@{}}Admix\\ \cite{admix}\end{tabular} & \begin{tabular}[c]{@{}c@{}}TI\\ \cite{ti}\end{tabular} & \begin{tabular}[c]{@{}c@{}}SSA\\ \cite{ssa}\end{tabular} & \begin{tabular}[c]{@{}c@{}}SIA\\ \cite{sia}\end{tabular} & \begin{tabular}[c]{@{}c@{}}MI\\ \cite{mifgsm}\end{tabular} & \begin{tabular}[c]{@{}c@{}}PI\\ \cite{pifgsm}\end{tabular} & \begin{tabular}[c]{@{}c@{}}VT\\ \cite{vt}\end{tabular} & \begin{tabular}[c]{@{}c@{}}PGN\\ \cite{pgn}\end{tabular} & \begin{tabular}[c]{@{}c@{}}CWA\\ \cite{cwa}\end{tabular} & \begin{tabular}[c]{@{}c@{}}SVRE\\ \cite{svre}\end{tabular} & \begin{tabular}[c]{@{}c@{}}RAP\\ \cite{rap}\end{tabular} & \begin{tabular}[c]{@{}c@{}}TPA\\ \cite{tpa}\end{tabular} & \begin{tabular}[c]{@{}c@{}}FEM\\ \cite{fem}\end{tabular} & \begin{tabular}[c]{@{}c@{}}MEF\\ \cite{mef}\end{tabular} & \begin{tabular}[c]{@{}c@{}}AdaEA\\ \cite{adaea}\end{tabular} & \begin{tabular}[c]{@{}c@{}}SMER\\ \cite{smer}\end{tabular} & \textbf{\begin{tabular}[c]{@{}c@{}}DRAP\\ (10 iter)\end{tabular}} & \textbf{\begin{tabular}[c]{@{}c@{}}DRAP\\ (50 iter)\end{tabular}} & \textbf{\begin{tabular}[c]{@{}c@{}}DRAP\\ (100 iter)\end{tabular}} & \textbf{\begin{tabular}[c]{@{}c@{}}DRAP\\ (150 iter)\end{tabular}} & \textbf{\begin{tabular}[c]{@{}c@{}}DRAP\\ (200 iter)\end{tabular}} \\ \midrule
\multicolumn{1}{c|}{} & AlexNet & 58.3 & 62.3 & 64.8 & 66.8 & 61.7 & 84.9 & 83.1 & 78.4 & 80.6 & 57.6 & 87.7 & 92.4 & 84.5 & 52.9 & 82.1 & 89.4 & 91.0 & 80.9 & 93.9 & 89.5 & 97.3 & 98.1 & {\ul 98.2} & \textbf{98.4} \\
\multicolumn{1}{c|}{} & VGG-16-BN & 88.4 & 93.6 & 94.4 & 95.2 & 90.3 & 98.6 & 99.0 & 95.6 & 97.1 & 82.9 & 97.8 & 98.0 & 98.0 & 67.1 & 95.5 & 99.2 & 97.8 & 95.0 & 96.2 & 87.6 & 98.9 & {\ul 99.4} & \textbf{99.6} & \textbf{99.6} \\
\multicolumn{1}{c|}{} & DenseNet-201 & 81.8 & 92.6 & 93.8 & 94.8 & 85.7 & 98.6 & 99.1 & 94.2 & 97.4 & 77.8 & 98.3 & 97.7 & 98.3 & 62.1 & 95.5 & 99.2 & 98.6 & 94.4 & 96.7 & 90.1 & 99.6 & {\ul 99.9} & \textbf{100.0} & \textbf{100.0} \\
\multicolumn{1}{c|}{} & GoogLeNet & 64.5 & 80.2 & 83.2 & 83.5 & 69.8 & 96.7 & 97.2 & 86.9 & 91.3 & 67.3 & 95.2 & 95.1 & 93.0 & 52.1 & 88.9 & 97.7 & 96.0 & 88.6 & 94.5 & 89.7 & 98.9 & {\ul 99.3} & {\ul 99.3} & \textbf{99.5} \\
\multicolumn{1}{c|}{} & ShuffleNetV2 & 72.1 & 82.4 & 84.5 & 84.7 & 75.7 & 96.2 & 97.0 & 87.2 & 89.9 & 70.6 & 96.0 & 96.3 & 93.5 & 59.1 & 90.2 & 95.8 & 97.0 & 90.0 & 96.4 & 93.7 & 99.4 & {\ul 99.8} & {\ul 99.8} & \textbf{99.9} \\
\multicolumn{1}{c|}{} & MobileNetV2 & 82.5 & 92.3 & 91.8 & 94.6 & 86.0 & 98.3 & 98.9 & 93.7 & 96.1 & 80.1 & 98.2 & 97.6 & 97.6 & 65.8 & 94.1 & 98.7 & 98.3 & 94.4 & 96.7 & 92.0 & 99.7 & {\ul 99.9} & \textbf{100.0} & {\ul 99.9} \\
\multicolumn{1}{c|}{} & MobileNetV3-L & 68.7 & 85.6 & 84.9 & 85.8 & 74.4 & 97.6 & 97.9 & 88.4 & 92.3 & 68.8 & 96.7 & 97.4 & 94.6 & 55.3 & 90.0 & 97.5 & 97.7 & 91.5 & 96.8 & 91.1 & 99.6 & 99.6 & {\ul 99.7} & \textbf{99.8} \\
\multicolumn{1}{c|}{} & MNASNet & 80.9 & 91.0 & 92.4 & 92.7 & 83.7 & 98.2 & 98.9 & 93.3 & 95.8 & 77.3 & 98.2 & 97.9 & 97.5 & 62.9 & 94.1 & 98.4 & 98.4 & 93.2 & 96.6 & 92.3 & {\ul 99.8} & \textbf{99.9} & \textbf{99.9} & \textbf{99.9} \\
\multicolumn{1}{c|}{} & EfficientNet & 60.9 & 80.5 & 78.6 & 76.2 & 67.8 & 95.3 & 96.5 & 82.2 & 87.9 & 64.2 & 95.4 & 95.2 & 91.4 & 50.1 & 84.6 & 96.1 & 95.4 & 84.6 & 95.2 & 87.0 & 98.7 & {\ul 99.4} & \textbf{99.6} & \textbf{99.6} \\
\multicolumn{1}{c|}{} & ConvNeXt-L & 85.8 & 92.6 & 92.4 & 95.1 & 84.6 & 96.4 & 98.6 & 94.6 & 96.7 & 80.5 & 95.2 & 98.2 & 98.2 & 59.4 & 94.4 & 98.6 & 94.8 & 92.9 & 97.4 & 76.2 & 98.7 & {\ul 99.3} & {\ul 99.3} & \textbf{99.6} \\
\multicolumn{1}{c|}{\multirow{-11}{*}{\textbf{\begin{tabular}[c]{@{}c@{}}ConvNet\\ Set\end{tabular}}}} & \cellcolor[HTML]{D9D9D9}\textit{\textbf{Average}} & \cellcolor[HTML]{D9D9D9}74.4 & \cellcolor[HTML]{D9D9D9}85.3 & \cellcolor[HTML]{D9D9D9}86.1 & \cellcolor[HTML]{D9D9D9}86.9 & \cellcolor[HTML]{D9D9D9}78.0 & \cellcolor[HTML]{D9D9D9}96.1 & \cellcolor[HTML]{D9D9D9}96.6 & \cellcolor[HTML]{D9D9D9}89.5 & \cellcolor[HTML]{D9D9D9}92.5 & \cellcolor[HTML]{D9D9D9}72.7 & \cellcolor[HTML]{D9D9D9}95.9 & \cellcolor[HTML]{D9D9D9}96.6 & \cellcolor[HTML]{D9D9D9}94.7 & \cellcolor[HTML]{D9D9D9}58.7 & \cellcolor[HTML]{D9D9D9}90.9 & \cellcolor[HTML]{D9D9D9}97.1 & \cellcolor[HTML]{D9D9D9}96.5 & \cellcolor[HTML]{D9D9D9}90.6 & \cellcolor[HTML]{D9D9D9}96.0 & \cellcolor[HTML]{D9D9D9}88.9 & \cellcolor[HTML]{D9D9D9}99.1 & \cellcolor[HTML]{D9D9D9}{\ul 99.5} & \cellcolor[HTML]{D9D9D9}{\ul 99.5} & \cellcolor[HTML]{D9D9D9}\textbf{99.6} \\ \midrule
\multicolumn{1}{c|}{} & ViT-S & 34.7 & 59.5 & 45.3 & 49.3 & 43.6 & 78.6 & 89.4 & 66.2 & 71.2 & 37.1 & 83.8 & 89.7 & 77.5 & 28.1 & 67.3 & 85.6 & 89.3 & 70.1 & 87.7 & 69.2 & 94.8 & 95.8 & {\ul 96.8} & \textbf{97.4} \\
\multicolumn{1}{c|}{} & DeiT-S & 41.2 & 63.8 & 52.6 & 57.4 & 49.3 & 83.0 & 89.4 & 72.0 & 74.8 & 45.5 & 87.8 & 92.1 & 80.6 & 36.4 & 72.2 & 86.9 & 90.8 & 76.3 & 91.0 & 79.3 & 96.3 & 97.3 & {\ul 97.5} & \textbf{97.9} \\
\multicolumn{1}{c|}{} & PoolFormer-S & 69.3 & 87.1 & 82.0 & 84.7 & 74.7 & 95.5 & 97.5 & 88.3 & 90.6 & 68.8 & 94.4 & 95.5 & 93.7 & 51.1 & 87.7 & 97.2 & 94.9 & 87.2 & 94.1 & 81.7 & 98.1 & 98.8 & {\ul 98.8} & \textbf{99.0} \\
\multicolumn{1}{c|}{} & TNT-S & 44.1 & 73.9 & 58.2 & 64.9 & 53.9 & 88.2 & 94.9 & 74.3 & 78.5 & 49.2 & 90.0 & 92.1 & 85.4 & 34.8 & 74.3 & 90.4 & 90.6 & 77.8 & 89.6 & 76.4 & 95.5 & 96.8 & {\ul 97.3} & \textbf{97.6} \\
\multicolumn{1}{c|}{} & Swin-S & 29.8 & 60.6 & 41.6 & 49.1 & 34.7 & 71.8 & \textbf{90.5} & 60.0 & 65.7 & 34.4 & 75.2 & 82.6 & 72.9 & 25.2 & 58.3 & 80.3 & 83.3 & 64.0 & 78.7 & 56.5 & 84.8 & 86.9 & 88.2 & {\ul 89.5} \\
\multicolumn{1}{c|}{} & XCiT-S & 27.6 & 55.5 & 33.6 & 37.6 & 32.8 & 61.5 & \textbf{86.0} & 45.9 & 50.7 & 32.1 & 67.3 & 71.8 & 58.0 & 23.3 & 47.6 & 68.7 & 74.0 & 53.4 & 71.3 & 53.2 & 77.6 & 79.2 & 78.0 & {\ul 80.7} \\
\multicolumn{1}{c|}{} & CaiT-S & 17.2 & 41.8 & 23.0 & 25.1 & 23.5 & 49.3 & 78.4 & 38.7 & 43.4 & 22.0 & 58.5 & 68.5 & 49.9 & 15.5 & 40.6 & 61.3 & 65.8 & 45.7 & 68.8 & 49.4 & 76.6 & 79.3 & {\ul 79.8} & \textbf{80.7} \\
\multicolumn{1}{c|}{\multirow{-8}{*}{\textbf{\begin{tabular}[c]{@{}c@{}}Metaformer\\ Set\end{tabular}}}} & \cellcolor[HTML]{D9D9D9}\textit{\textbf{Average}} & \cellcolor[HTML]{D9D9D9}37.7 & \cellcolor[HTML]{D9D9D9}63.2 & \cellcolor[HTML]{D9D9D9}48.0 & \cellcolor[HTML]{D9D9D9}52.6 & \cellcolor[HTML]{D9D9D9}44.6 & \cellcolor[HTML]{D9D9D9}75.4 & \cellcolor[HTML]{D9D9D9}89.4 & \cellcolor[HTML]{D9D9D9}63.6 & \cellcolor[HTML]{D9D9D9}67.8 & \cellcolor[HTML]{D9D9D9}41.3 & \cellcolor[HTML]{D9D9D9}79.6 & \cellcolor[HTML]{D9D9D9}84.6 & \cellcolor[HTML]{D9D9D9}74.0 & \cellcolor[HTML]{D9D9D9}30.6 & \cellcolor[HTML]{D9D9D9}64.0 & \cellcolor[HTML]{D9D9D9}81.5 & \cellcolor[HTML]{D9D9D9}84.1 & \cellcolor[HTML]{D9D9D9}67.8 & \cellcolor[HTML]{D9D9D9}83.0 & \cellcolor[HTML]{D9D9D9}66.5 & \cellcolor[HTML]{D9D9D9}89.1 & \cellcolor[HTML]{D9D9D9}90.6 & \cellcolor[HTML]{D9D9D9}{\ul 90.9} & \cellcolor[HTML]{D9D9D9}\textbf{91.8} \\ \midrule
\multicolumn{1}{c|}{} & RaWideResNet-101-2 & 23.2 & 23.4 & 23.4 & 25.7 & 25.4 & 34.8 & 28.7 & 31.3 & 32.2 & 23.0 & 42.1 & 62.7 & 42.9 & 19.7 & 36.8 & 36.9 & 41.7 & 39.8 & 68.2 & 56.3 & 69.8 & 72.8 & {\ul 73.6} & \textbf{74.2} \\
\multicolumn{1}{c|}{} & WideResNet-50-2 & 29.4 & 29.3 & 30.4 & 32.5 & 32.9 & 42.4 & 35.0 & 39.0 & 40.3 & 28.7 & 50.2 & 72.2 & 51.8 & 24.8 & 45.6 & 45.7 & 50.1 & 48.4 & 78.6 & 69.1 & 82.4 & 84.5 & {\ul 85.5} & \textbf{85.8} \\
\multicolumn{1}{c|}{} & ResNet-50 & 44.7 & 45.2 & 46.5 & 47.3 & 47.7 & 55.6 & 50.2 & 50.2 & 51.2 & 44.4 & 59.6 & 73.8 & 57.9 & 41.4 & 57.3 & 54.2 & 58.7 & 56.0 & 78.9 & 75.2 & 84.0 & 85.5 & {\ul 86.3} & \textbf{86.8} \\
\multicolumn{1}{c|}{} & ConvNeXt-L & 14.9 & 15.6 & 16.0 & 17.6 & 17.0 & 23.3 & 19.7 & 24.0 & 24.7 & 15.1 & 30.9 & 52.8 & 28.9 & 13.4 & 26.6 & 25.9 & 32.5 & 29.3 & 58.4 & 45.8 & 58.4 & 61.1 & {\ul 62.4} & \textbf{63.1} \\
\multicolumn{1}{c|}{} & ConvNeXt-B & 15.9 & 16.4 & 15.8 & 19.0 & 17.7 & 25.2 & 20.5 & 23.3 & 24.5 & 16.0 & 33.0 & 57.5 & 31.6 & 12.8 & 28.0 & 28.4 & 35.0 & 31.5 & 62.1 & 48.9 & 63.7 & 65.7 & {\ul 66.1} & \textbf{66.9} \\
\multicolumn{1}{c|}{} & ConvNeXt-L-ConvStem & 14.6 & 14.8 & 14.6 & 16.9 & 16.2 & 22.2 & 18.5 & 22.3 & 22.8 & 13.8 & 30.5 & 53.2 & 28.0 & 11.9 & 25.8 & 26.8 & 32.6 & 29.7 & 59.1 & 47.4 & 59.6 & 60.7 & {\ul 62.4} & \textbf{62.8} \\
\multicolumn{1}{c|}{} & ConvNeXt-B-ConvStem & 16.6 & 17.2 & 16.9 & 18.9 & 18.3 & 26.0 & 21.1 & 25.2 & 25.9 & 16.4 & 34.3 & 56.7 & 32.3 & 13.8 & 28.8 & 29.6 & 35.2 & 32.7 & 61.6 & 49.4 & 62.5 & 65.2 & {\ul 67.0} & \textbf{67.1} \\
\multicolumn{1}{c|}{} & Inc-v3$_{\textit{ens3}}$ & 19.6 & 34.7 & 27.4 & 29.6 & 25.0 & 65.0 & 63.6 & 40.1 & 45.4 & 27.5 & 72.7 & 77.2 & 54.1 & 17.1 & 37.7 & 63.9 & 80.3 & 48.9 & 80.2 & 65.5 & 86.2 & 87.6 & {\ul 89.1} & \textbf{90.6} \\
\multicolumn{1}{c|}{} & Inc-v3$_{\textit{ens4}}$ & 22.9 & 37.4 & 31.6 & 32.2 & 26.3 & 62.8 & 62.9 & 43.1 & 46.7 & 29.3 & 72.9 & 76.0 & 53.7 & 20.4 & 39.7 & 62.0 & 81.0 & 50.6 & 79.2 & 66.2 & 84.1 & {\ul 88.6} & {\ul 88.6} & \textbf{90.0} \\
\multicolumn{1}{c|}{} & IncRes-v2$_{\textit{ens}}$ & 12.1 & 24.6 & 18.7 & 20.7 & 17.8 & 50.6 & 47.9 & 29.8 & 33.3 & 17.1 & 64.7 & 65.3 & 40.9 & 11.1 & 26.8 & 48.1 & 72.7 & 37.9 & 74.6 & 55.8 & 75.7 & 79.4 & {\ul 81.0} & \textbf{82.5} \\
\multicolumn{1}{c|}{\multirow{-11}{*}{\textbf{\begin{tabular}[c]{@{}c@{}}ConvNet\\ (AT) Set\end{tabular}}}} & \cellcolor[HTML]{D9D9D9}\textit{\textbf{Average}} & \cellcolor[HTML]{D9D9D9}21.4 & \cellcolor[HTML]{D9D9D9}25.9 & \cellcolor[HTML]{D9D9D9}24.1 & \cellcolor[HTML]{D9D9D9}26.0 & \cellcolor[HTML]{D9D9D9}24.4 & \cellcolor[HTML]{D9D9D9}40.8 & \cellcolor[HTML]{D9D9D9}36.8 & \cellcolor[HTML]{D9D9D9}32.8 & \cellcolor[HTML]{D9D9D9}34.7 & \cellcolor[HTML]{D9D9D9}23.1 & \cellcolor[HTML]{D9D9D9}49.1 & \cellcolor[HTML]{D9D9D9}64.7 & \cellcolor[HTML]{D9D9D9}42.2 & \cellcolor[HTML]{D9D9D9}18.6 & \cellcolor[HTML]{D9D9D9}35.3 & \cellcolor[HTML]{D9D9D9}42.2 & \cellcolor[HTML]{D9D9D9}52.0 & \cellcolor[HTML]{D9D9D9}40.5 & \cellcolor[HTML]{D9D9D9}70.1 & \cellcolor[HTML]{D9D9D9}58.0 & \cellcolor[HTML]{D9D9D9}72.6 & \cellcolor[HTML]{D9D9D9}75.1 & \cellcolor[HTML]{D9D9D9}{\ul 76.2} & \cellcolor[HTML]{D9D9D9}\textbf{77.0} \\ \midrule
\multicolumn{1}{c|}{} & Swin-B & 16.3 & 16.4 & 15.9 & 18.7 & 17.9 & 24.0 & 20.3 & 23.3 & 23.9 & 16.0 & 31.4 & 55.5 & 30.3 & 13.9 & 26.3 & 27.7 & 32.9 & 31.1 & 60.2 & 48.4 & 61.0 & 63.2 & {\ul 64.5} & \textbf{64.6} \\
\multicolumn{1}{c|}{} & Swin-L & 14.1 & 14.6 & 14.3 & 15.9 & 15.6 & 20.9 & 18.0 & 20.7 & 21.4 & 14.1 & 28.8 & 50.5 & 26.6 & 11.5 & 24.0 & 24.3 & 30.7 & 29.0 & 57.2 & 44.5 & 56.1 & 59.1 & {\ul 60.5} & \textbf{61.1} \\
\multicolumn{1}{c|}{} & XCiT-L & 23.4 & 23.5 & 24.0 & 26.7 & 25.9 & 35.6 & 28.4 & 34.1 & 34.4 & 22.5 & 45.0 & 72.9 & 45.7 & 18.6 & 38.8 & 38.9 & 44.2 & 43.8 & 76.3 & 62.4 & 77.9 & 81.6 & {\ul 83.2} & \textbf{83.8} \\
\multicolumn{1}{c|}{} & ViT-B-ConvStem & 15.9 & 16.0 & 15.8 & 18.5 & 17.7 & 25.0 & 20.2 & 23.7 & 23.8 & 15.2 & 33.7 & 60.9 & 32.8 & 12.6 & 29.6 & 29.4 & 34.5 & 33.8 & 65.2 & 52.5 & 66.5 & 70.1 & {\ul 71.0} & \textbf{72.1} \\
\multicolumn{1}{c|}{\multirow{-5}{*}{\textbf{\begin{tabular}[c]{@{}c@{}}Metaformer\\ (AT) Set\end{tabular}}}} & \cellcolor[HTML]{D9D9D9}\textit{\textbf{Average}} & \cellcolor[HTML]{D9D9D9}17.4 & \cellcolor[HTML]{D9D9D9}17.6 & \cellcolor[HTML]{D9D9D9}17.5 & \cellcolor[HTML]{D9D9D9}20.0 & \cellcolor[HTML]{D9D9D9}19.3 & \cellcolor[HTML]{D9D9D9}26.4 & \cellcolor[HTML]{D9D9D9}21.7 & \cellcolor[HTML]{D9D9D9}25.5 & \cellcolor[HTML]{D9D9D9}25.9 & \cellcolor[HTML]{D9D9D9}17.0 & \cellcolor[HTML]{D9D9D9}34.7 & \cellcolor[HTML]{D9D9D9}60.0 & \cellcolor[HTML]{D9D9D9}33.9 & \cellcolor[HTML]{D9D9D9}14.2 & \cellcolor[HTML]{D9D9D9}29.7 & \cellcolor[HTML]{D9D9D9}30.1 & \cellcolor[HTML]{D9D9D9}35.6 & \cellcolor[HTML]{D9D9D9}34.4 & \cellcolor[HTML]{D9D9D9}64.7 & \cellcolor[HTML]{D9D9D9}52.0 & \cellcolor[HTML]{D9D9D9}65.4 & \cellcolor[HTML]{D9D9D9}68.5 & \cellcolor[HTML]{D9D9D9}{\ul 69.8} & \cellcolor[HTML]{D9D9D9}\textbf{70.4} \\ \midrule
\rowcolor[HTML]{D9D9D9} 
\multicolumn{2}{c|}{\cellcolor[HTML]{D9D9D9}\textit{\textbf{Overall Average}}} & 41.7 & 52.4 & 48.7 & 50.9 & 45.6 & 64.6 & 66.0 & 57.1 & 59.7 & 42.4 & 69.2 & 78.9 & 65.2 & 33.7 & 59.0 & 67.2 & 71.5 & 62.0 & 80.7 & 69.1 & 83.9 & 85.6 & {\ul 86.2} & \textbf{86.8} \\ \bottomrule
\end{tabular}
}
\end{table*}

\begin{table*}[!t]
\ContinuedFloat
\centering
\caption{\blue{Untargeted attack success rates (\%, $\uparrow$) on ImageNet dataset under perturbation budget $\gamma = 16/255$ (continued). $^*$ denotes surrogate models are ResNet-50 and ResNet-50(AT). $^\ddagger$ denotes surrogate models are ResNet-50 and ViT-B.}}
\resizebox{0.7\textwidth}{!}{%
\bluetable
\begin{tabular}{@{}cc|ccccc|ccc@{}}
\toprule
\multicolumn{2}{c|}{\textbf{Target Model Set}} & \begin{tabular}[c]{@{}c@{}}FIA$^{*}$\\ \cite{fia}\end{tabular} & \begin{tabular}[c]{@{}c@{}}NAA$^{*}$\\ \cite{naa}\end{tabular} & \begin{tabular}[c]{@{}c@{}}GhostNet$^{*}$\\ \cite{ghostnet}\end{tabular} & \textbf{\begin{tabular}[c]{@{}c@{}}DRAP$^{*}$\\ (10 iter)\end{tabular}} & \textbf{\begin{tabular}[c]{@{}c@{}}DRAP$^{*}$\\ (50 iter)\end{tabular}} & \begin{tabular}[c]{@{}c@{}}Bayesian$^{\ddagger}$\\ \cite{bayesian}\end{tabular} & \textbf{\begin{tabular}[c]{@{}c@{}}DRAP$^{\ddagger}$\\ (10 iter)\end{tabular}} & \textbf{\begin{tabular}[c]{@{}c@{}}DRAP$^{\ddagger}$\\ (50 iter)\end{tabular}} \\ \midrule
\multicolumn{1}{c|}{} & AlexNet & 66.9 & 70.3 & 53.6 & {\ul 92.7} & \textbf{97.4} & 82.1 & {\ul 85.2} & \textbf{92.7} \\
\multicolumn{1}{c|}{} & VGG-16-BN & 84.1 & 93.1 & 77.8 & {\ul 94.4} & \textbf{99.4} & {\ul 98.1} & 92.3 & \textbf{99.2} \\
\multicolumn{1}{c|}{} & DenseNet-201 & 82.6 & 94.2 & 75.0 & {\ul 95.6} & \textbf{99.8} & {\ul 98.3} & 91.0 & \textbf{99.6} \\
\multicolumn{1}{c|}{} & GoogLeNet & 72.7 & 88.0 & 56.3 & {\ul 95.2} & \textbf{99.2} & {\ul 96.3} & 89.9 & \textbf{98.2} \\
\multicolumn{1}{c|}{} & ShuffleNetV2 & 79.0 & 87.9 & 63.1 & {\ul 96.6} & \textbf{99.6} & {\ul 98.2} & 91.9 & \textbf{99.0} \\
\multicolumn{1}{c|}{} & MobileNetV2 & 83.6 & 93.3 & 73.9 & {\ul 96.9} & \textbf{99.7} & {\ul 98.8} & 93.0 & \textbf{99.4} \\
\multicolumn{1}{c|}{} & MobileNetV3-L & 70.1 & 85.6 & 52.7 & {\ul 94.3} & \textbf{99.3} & {\ul 96.6} & 88.0 & \textbf{98.2} \\
\multicolumn{1}{c|}{} & MNASNet & 83.0 & 90.5 & 69.0 & {\ul 96.5} & \textbf{99.8} & {\ul 98.3} & 92.9 & \textbf{99.6} \\
\multicolumn{1}{c|}{} & EfficientNet & 60.2 & 80.3 & 44.3 & {\ul 90.0} & \textbf{98.1} & {\ul 91.5} & 80.4 & \textbf{93.3} \\
\multicolumn{1}{c|}{} & ConvNeXt-L & 29.5 & {\ul 63.4} & 23.7 & 57.7 & \textbf{80.6} & {\ul 68.9} & 52.1 & \textbf{69.7} \\
\multicolumn{1}{c|}{\multirow{-11}{*}{\textbf{\begin{tabular}[c]{@{}c@{}}ConvNet\\ Set\end{tabular}}}} & \cellcolor[HTML]{D9D9D9}\textit{\textbf{Average}} & \cellcolor[HTML]{D9D9D9}71.2 & \cellcolor[HTML]{D9D9D9}84.7 & \cellcolor[HTML]{D9D9D9}58.9 & \cellcolor[HTML]{D9D9D9}{\ul 91.0} & \cellcolor[HTML]{D9D9D9}\textbf{97.3} & \cellcolor[HTML]{D9D9D9}{\ul 92.7} & \cellcolor[HTML]{D9D9D9}85.7 & \cellcolor[HTML]{D9D9D9}\textbf{94.9} \\ \midrule
\multicolumn{1}{c|}{} & ViT-S & 23.6 & 43.8 & 14.2 & {\ul 63.0} & \textbf{83.5} & {\ul 80.0} & 67.3 & \textbf{84.4} \\
\multicolumn{1}{c|}{} & DeiT-S & 26.2 & 47.3 & 17.0 & {\ul 71.4} & \textbf{86.4} & \textbf{90.3} & 75.0 & {\ul 89.5} \\
\multicolumn{1}{c|}{} & PoolFormer-S & 49.4 & 78.5 & 41.0 & {\ul 83.4} & \textbf{95.6} & {\ul 90.1} & 74.9 & \textbf{91.0} \\
\multicolumn{1}{c|}{} & TNT-S & 24.2 & 53.5 & 16.2 & {\ul 67.0} & \textbf{83.1} & \textbf{87.2} & 66.8 & {\ul 82.8} \\
\multicolumn{1}{c|}{} & Swin-S & 15.3 & 42.3 & 11.3 & {\ul 44.1} & \textbf{64.8} & {\ul 54.4} & 44.0 & \textbf{56.8} \\
\multicolumn{1}{c|}{} & XCiT-S & 17.3 & 33.9 & 12.8 & {\ul 40.5} & \textbf{57.7} & \textbf{44.6} & 32.4 & {\ul 41.2} \\
\multicolumn{1}{c|}{} & CaiT-S & 10.3 & 25.2 & 6.8 & {\ul 39.3} & \textbf{57.5} & {\ul 43.4} & 32.4 & \textbf{44.7} \\
\multicolumn{1}{c|}{\multirow{-8}{*}{\textbf{\begin{tabular}[c]{@{}c@{}}Metaformer\\ Set\end{tabular}}}} & \cellcolor[HTML]{D9D9D9}\textit{\textbf{Average}} & \cellcolor[HTML]{D9D9D9}23.8 & \cellcolor[HTML]{D9D9D9}46.4 & \cellcolor[HTML]{D9D9D9}17.0 & \cellcolor[HTML]{D9D9D9}{\ul 58.4} & \cellcolor[HTML]{D9D9D9}\textbf{75.5} & \cellcolor[HTML]{D9D9D9}{\ul 70.0} & \cellcolor[HTML]{D9D9D9}56.1 & \cellcolor[HTML]{D9D9D9}\textbf{70.1} \\ \midrule
\multicolumn{1}{c|}{} & RaWideResNet-101-2 & 26.1 & 23.5 & 18.2 & {\ul 51.5} & \textbf{63.9} & 20.4 & {\ul 23.7} & \textbf{25.7} \\
\multicolumn{1}{c|}{} & WideResNet-50-2 & 35.2 & 32.2 & 23.9 & {\ul 67.4} & \textbf{82.2} & 27.1 & {\ul 30.1} & \textbf{33.4} \\
\multicolumn{1}{c|}{} & ResNet-50 & 49.1 & 48.1 & 41.6 & {\ul 78.0} & \textbf{84.6} & 44.1 & {\ul 46.6} & \textbf{50.0} \\
\multicolumn{1}{c|}{} & ConvNeXt-L & 16.5 & 14.5 & 10.9 & {\ul 38.7} & \textbf{48.6} & 13.9 & {\ul 15.4} & \textbf{16.2} \\
\multicolumn{1}{c|}{} & ConvNeXt-B & 17.9 & 15.6 & 11.7 & {\ul 41.8} & \textbf{54.6} & 14.2 & {\ul 15.5} & \textbf{16.4} \\
\multicolumn{1}{c|}{} & ConvNeXt-L-ConvStem & 15.6 & 14.0 & 10.8 & {\ul 39.6} & \textbf{49.8} & 13.1 & {\ul 14.9} & \textbf{16.3} \\
\multicolumn{1}{c|}{} & ConvNeXt-B-ConvStem & 18.4 & 15.7 & 12.2 & {\ul 43.1} & \textbf{53.2} & 15.0 & {\ul 16.2} & \textbf{17.9} \\
\multicolumn{1}{c|}{} & Inc-v3$_{\textit{ens3}}$ & 17.0 & 37.7 & 11.5 & {\ul 65.5} & \textbf{81.9} & \textbf{43.6} & 37.5 & {\ul 39.1} \\
\multicolumn{1}{c|}{} & Inc-v3$_{\textit{ens4}}$ & 22.3 & 41.0 & 14.1 & {\ul 63.7} & \textbf{79.9} & \textbf{43.2} & 38.8 & {\ul 37.6} \\
\multicolumn{1}{c|}{} & IncRes-v2$_{\textit{ens}}$ & 13.2 & 27.1 & 7.0 & {\ul 53.3} & \textbf{68.1} & \textbf{27.5} & 26.3 & {\ul 24.6} \\
\multicolumn{1}{c|}{\multirow{-11}{*}{\textbf{\begin{tabular}[c]{@{}c@{}}ConvNet\\ (AT) Set\end{tabular}}}} & \cellcolor[HTML]{D9D9D9}\textit{\textbf{Average}} & \cellcolor[HTML]{D9D9D9}23.1 & \cellcolor[HTML]{D9D9D9}26.9 & \cellcolor[HTML]{D9D9D9}16.2 & \cellcolor[HTML]{D9D9D9}{\ul 54.3} & \cellcolor[HTML]{D9D9D9}\textbf{66.7} & \cellcolor[HTML]{D9D9D9}26.2 & \cellcolor[HTML]{D9D9D9}{\ul 26.5} & \cellcolor[HTML]{D9D9D9}\textbf{27.7} \\ \midrule
\multicolumn{1}{c|}{} & Swin-B & 17.4 & 15.4 & 12.3 & {\ul 40.9} & \textbf{50.0} & 14.5 & {\ul 16.2} & \textbf{17.0} \\
\multicolumn{1}{c|}{} & Swin-L & 14.8 & 13.8 & 10.1 & {\ul 36.3} & \textbf{46.1} & 12.4 & {\ul 13.9} & \textbf{14.9} \\
\multicolumn{1}{c|}{} & XCiT-L & 24.0 & 21.3 & 15.5 & {\ul 53.5} & \textbf{67.3} & 19.5 & {\ul 22.7} & \textbf{25.0} \\
\multicolumn{1}{c|}{} & ViT-B-ConvStem & 15.9 & 15.2 & 11.7 & {\ul 42.5} & \textbf{55.8} & 13.8 & {\ul 16.1} & \textbf{18.7} \\
\multicolumn{1}{c|}{\multirow{-5}{*}{\textbf{\begin{tabular}[c]{@{}c@{}}Metaformer\\ (AT) Set\end{tabular}}}} & \cellcolor[HTML]{D9D9D9}\textit{\textbf{Average}} & \cellcolor[HTML]{D9D9D9}18.0 & \cellcolor[HTML]{D9D9D9}16.4 & \cellcolor[HTML]{D9D9D9}12.4 & \cellcolor[HTML]{D9D9D9}{\ul 43.3} & \cellcolor[HTML]{D9D9D9}\textbf{54.8} & \cellcolor[HTML]{D9D9D9}15.1 & \cellcolor[HTML]{D9D9D9}{\ul 17.2} & \cellcolor[HTML]{D9D9D9}\textbf{18.9} \\ \midrule
\rowcolor[HTML]{D9D9D9} 
\multicolumn{2}{c|}{\cellcolor[HTML]{D9D9D9}\textit{\textbf{Overall Average}}} & 38.1 & 48.6 & 29.7 & {\ul 65.6} & \textbf{77.0} & {\ul 56.1} & 51.1 & \textbf{57.8} \\ \bottomrule
\end{tabular}
}
\end{table*}

\subsubsection{\blue{Additional experimental results under $\gamma=16/255$}}

\blue{To validate the generality of DRAP to conventional settings, we conduct comparisons under $\gamma=16/255$ with iterations of 10 and step size of $1.6/255$, while keeping the rest of the protocol consistent with the untargeted setting in Section \ref{Sec: Main Results}. The results are reported in Table \ref{tab:main_16_255_1}. We find that, when all methods use 10 iterations, DRAP achieves competitive performance compared with baselines. Moreover, DRAP relies on surrogate model diversity, which grows with the number of iterations as discussed in Section \ref{alb_sec: dis} Q2-1. With only 10 iterations, DRAP uses only 2 samples per surrogate component (our default is $n=40$), which severely limits the within‑distribution diversity that DRAP is designed to exploit. When we increase the number of iterations for DRAP to 50 (\ie, $n=10$) and beyond, DRAP becomes the best‑performing attack, showing that our method remains effective at $\gamma=16/255$.}

\subsubsection{\blue{Additional experimental results under $\gamma=8/255$}}

\begin{figure*}[t]
 \setlength{\abovecaptionskip}{-0.1cm}
\setlength{\belowcaptionskip}{-0.1cm}
 \centering
\includegraphics[width=0.85\linewidth]{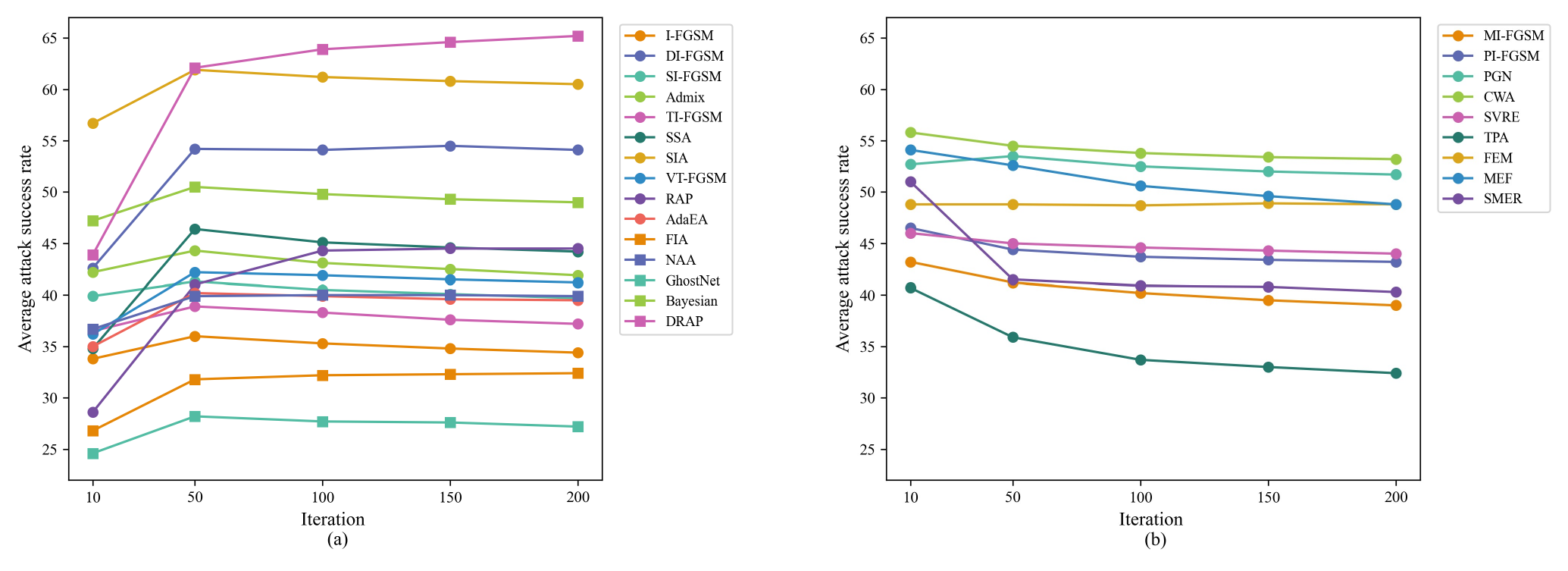}
 \caption{\blue{Average untargeted attack success rate over all 31 target models \wrt~different \bluetwo{numbers} of iterations under perturbation budget $\gamma=8/255$. \textbf{(a)} Non-momentum-based attacks except DRAP. Their default number of iterations $n_{iter}$ is set to 200, since they benefit from more than 10 iterations. \textbf{(b)} Momentum-based attacks which accumulate gradients across iterations. Their default number of iterations $n_{iter}$ is set to 10 as suggested in their original papers, since their performance generally deteriorates for additional rounds.}}\label{fig:8_255_iter}
 \vspace{-5pt}
\end{figure*}

\blue{We further conduct comparisons under $\gamma=8/255$ with step size of $0.8/255$ and varying iterations for all methods. In this experiment, in addition to validating the effectiveness of DRAP at $\gamma=8/255$, we further show that we use more iterations to ensure convergence of non-momentum-based attacks including DRAP, while carefully ensuring fairness for momentum-based attacks. The results are shown in Figure \ref{fig:8_255_iter}. The key observations are:
\begin{itemize}[leftmargin=10pt,itemsep=1pt,topsep=0pt]
    \item In Figure \ref{fig:8_255_iter}, we first observe a similar trend as in $\gamma=16/255$. DRAP achieves competitive attack performance at 10 iterations. Moreover, with more iterations for all methods, DRAP consistently outperforms others at every iteration number, showing that our method is also effective at $\gamma=8/255$.
    \item In Figure \ref{fig:8_255_iter}(a), besides DRAP, for data‑based, feature‑based, and part of optimization‑based and model-based methods without momentum strategy, the attack success rate increases with more iterations. The optimal performance is typically obtained with 50 iterations and plateaus with further iterations. This indicates that these attacks also benefit from larger iterations, and evaluating only at 10 iterations would underestimate their best achievable performance.
    \item However, in \bluetwo{Figure} \ref{fig:8_255_iter}(b), for these momentum‑based methods, increasing iterations beyond 10 generally degrades performance, as discussed in prior benchmarks \cite{eval, blackboxbench}. Thus, as already stated in the Implementation Details of Section \ref{Sec: Main Results}, we follow the original settings and keep iterations of these methods at 10 to avoid unfairness. 
\end{itemize}
}

\section{Additional Ablative Study}

\begin{table*}[h]
\centering
\caption{The number of gradient calculations $N_g$ required in different methods. Aside from DRAP, one update direction calculation of others requires I queries of model in surrogate ensemble. We use the same untargeted experimental protocol as in Section \ref{Sec: Main Results} ($I=5$). For RAP, the late-start $K_{LS}$ is proportionally scaled from original 100 (for 400 iterations) to 50 for our 200-iteration setting.}
\label{tab:cost}
\resizebox{0.9\textwidth}{!}{%
\begin{tabular}{@{}c|c|c|c|c|c@{}}
\toprule
\textbf{Attack} & \textbf{\# of gradient calculations ($N_g$) vs. $n_{iter}$} & \textbf{Hyper-parameters} & \textbf{$N_g(n_{iter}=25)$} & \blue{\textbf{$N_g(n_{iter}=50)$}} & \blue{\textbf{$N_g(n_{iter}=100)$}} \\ \midrule
MI-FGSM & $n_{iter}\times I$  & $\backslash$ & 125 & 250 & 500 \\ \midrule
PI-FGSM & $n_{iter}\times I$  & $\backslash$ & 125 & 250 & 500 \\ \midrule
RAP & $\left\{\begin{array}{l} n_{iter}\times I,\ n_{iter}<K_{LS}\\K_{LS}\times I+(n_{iter}-K_{LS})\times (T+1)\times I,n_{iter}\geq K_{LS}\\\end{array}\right.$ & $K_{LS}=\left\{\begin{array}{l}0,n_{iter} \leq 50\\50,n_{iter} > 50\\\end{array}\right.,\ T=10$ & 1375 & 2750 & 3000 \\ \midrule
CWA & $n_{iter}\times 2I$ & $\backslash$ & 250 & 500 & 1000 \\ \midrule
PGN & $n_{iter}\times N \times 2I$ & $N=20$ & 5000 & 10000 & 20000 \\ \midrule
DRAP & $\left\{\begin{array}{l} n_{iter},\ n_{iter}<n_{LS}\times I\\n_{LS}\times I+(n_{iter}-n_{LS}\times I)\times(T+1),n_{iter}\geq n_{LS}\times I\\\end{array}\right.$ & $n_{LS}=\left\{\begin{array}{l}0,n_{iter}/I \leq 5\\5,n_{iter}/I > 5\\\end{array}\right.,T=5$ & 150 &175 & 475 \\ \bottomrule
\end{tabular}%
}
\end{table*}

\subsection{Choice of Architecture within Prototype}\label{app:addresult_conv}

\begin{table}[t]
\centering
\caption{\blue{Average attack success rates (\%, $\uparrow$) of DRAP with ResNet-50 substituted by other convnet architectures.}}
\label{tab:change_arch}
\resizebox{8.5cm}{!}{%
\bluetable
\begin{tabular}{@{}cl|cccc@{}}
\toprule
\multicolumn{1}{l}{}                               &              & \textbf{\begin{tabular}[c]{@{}c@{}}ConvNet\\ Set\end{tabular}} & \textbf{\begin{tabular}[c]{@{}c@{}}Metaformer\\ Set\end{tabular}} & \textbf{\begin{tabular}[c]{@{}c@{}}ConvNet\\ (AT) Set\end{tabular}} & \textbf{\begin{tabular}[c]{@{}c@{}}Metaformer\\ (AT) Set\end{tabular}} \\ \midrule
\multicolumn{2}{c|}{ResNet-50}                                    & 80.3                                                  & 42.6                                                     & 24.9                                                       & 20.2                                                          \\ \midrule
\multicolumn{1}{c|}{\multirow{3}{*}{Substitution}} & VGG-19-BN    & 77.7                                                  & 40.8                                                     & 24.4                                                       & 20.3                                                          \\
\multicolumn{1}{c|}{}                              & Inception-V3 & 75.1                                                  & 42.6                                                     & 24.9                                                       & 20.3                                                          \\
\multicolumn{1}{c|}{}                              & DenseNet-121 & 80.1                                                  & 42.6                                                     & 24.9                                                       & 20.2                                                          \\ \bottomrule
\end{tabular}
}
\end{table}

In DRAP, we improve between-distribution diversity by choosing surrogate distributions on model weights of architectures from diverse prototypes. \emph{Does the choice of model architecture within a prototype \bluetwo{has} as significant an impact on improving transferability as the prototypes themselves, as validated in Q2-2?} To investigate this, we substitute the surrogate model ResNet-50 in default protocol in Section \ref{Sec: Main Results} with other ConvNet architectures such as VGG-19-BN, Inception-V3 and DenseNet-121 and report the attack success rate in Table \ref{tab:change_arch}. We can observe that DRAP is less sensitive to different convnets. This is because they give rise to a similar loss landscape to that of ResNet-50 from the adversarial perspective. The results demonstrate that the specific choice of model architecture within a prototype is less important to transferability than the between-distribution diversity among prototypes.

\blue{To further explore this question, we consider hybrid architectures, such as ConvNeXt-T, a modernized convnet which adopts design and training strategies popularized by vision transformers. 
Keeping configurations for the other two surrogate prototypes (adversarially trained convnet and metaformer) the same as in Section \ref{Sec: Main Results}, here we construct three different surrogate configurations for the normally trained prototype: (1) using ResNet-50, a traditional convnet, and ViT-B as separate convnet and metaformer prototypes, (2) using a single ConvNeXt-T to represent the whole normal prototype, (3) using ConvNeXt-T and ViT-B as separate convnet and metaformer prototypes. 
As reported in Table \ref{tab:convnext}, the first two configurations yield comparable performance on each target set. This indicates that ConvNeXt-T exhibits adversarial vulnerabilities that are simultaneously similar to those of normally trained convnets and metaformers. In resource-limited scenarios, such hybrid architectures can therefore be used to cover both the convnet and metaformer with a single surrogate component. Furthermore, combining ConvNeXt-T with an additional metaformer ViT-B leads to significantly higher attack success rates on both normal convnet and metaformer target sets than combining ResNet-50 with ViT-B, as hybrid models provide more diverse adversarial vulnerabilities than a traditional convnet.}

\begin{table}[t]
\centering
\caption{\blue{Average attack success rates (\%, $\uparrow$) of DRAP with different representations of the normally trained prototype.}}
\label{tab:convnext}
\resizebox{8.5cm}{!}{%
\bluetable
\begin{tabular}{@{}l|cccc@{}}
\toprule
\multicolumn{1}{l|}{\textbf{\begin{tabular}[c]{@{}l@{}}Components in\\ normal prototype\end{tabular}}} & \textbf{\begin{tabular}[c]{@{}c@{}}ConvNet\\ Set\end{tabular}} & \textbf{\begin{tabular}[c]{@{}c@{}}Metaformer\\ Set\end{tabular}} & \textbf{\begin{tabular}[c]{@{}c@{}}ConvNet\\ (AT) Set\end{tabular}} & \textbf{\begin{tabular}[c]{@{}c@{}}Metaformer\\ (AT) Set\end{tabular}}  \\ \midrule
ResNet-50+ViT-B   & 67.4   & 30.5   & 24.9   & 20.5      \\
ConvNeXt-T  & 67.6   & 29.2   & 23.9   & 20.5    \\
ConvNeXt-T+ViT-B  & 71.4   & 37.8   & 24.0   & 20.4    \\
\bottomrule
\end{tabular}
}
\end{table}

\subsection{Effect of Hyper-parameters}\label{app:addresult_para}


\begin{figure*}[t]
 \setlength{\abovecaptionskip}{-0.1cm}
\setlength{\belowcaptionskip}{-0.1cm}
 \centering
\includegraphics[width=\linewidth]{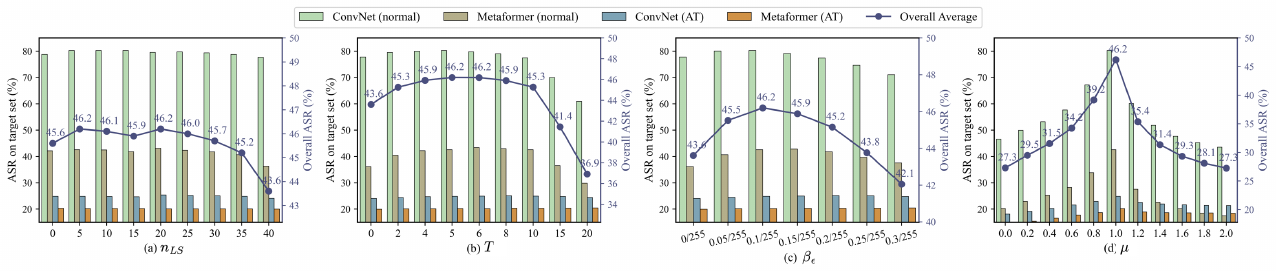}
 \caption{\blue{Ablation study on hyper-parameters of DRAP. We evaluate the effect of late start iteration number $n_{LS}$, inner iteration number $T$, inner step size $\beta_{\bm{\epsilon}}$ and decay factor $\mu$ on the attack success rates of DRAP.}}\label{fig:para_sensiti}
\end{figure*}

\blue{In this subsection, we provide an ablation study of DRAP on its hyper-parameters. In particular, we study the effect of the late start iteration number $n_{LS}$, the inner iteration number $T$, the inner step size $\beta_{\bm{\epsilon}}$, and the decay factor $\mu$. The results are shown in Figure \ref{fig:para_sensiti}. Note that the effect of the number of samples within one surrogate component $n$ has already been discussed in Section \ref{alb_sec: dis} Q2-1.}

\textbf{Late start iteration number $n_{LS}$.} As the late start iteration number $n_{LS}$ decides when the sharpness penalty begins to take effect during the optimization process, we range $n_{LS}$ from 0 to $n=40$. When $n_{LS}=0$, the sharpness penalty is active throughout the entire optimization process. As $n_{LS}$ increases, the influence of the sharpness penalty gradually weakens, eventually vanishing at $n_{LS}=n$, at which point DRAP reduces to a diverse-model-ensemble attack. As shown in Figure \ref{fig:para_sensiti}(a), for a fixed number of iterations used to update AE, attacks with more iterations penalizing sharpness can effectively improve the attack performance over $n_{LS}=n$. The peak performance is observed at approximately $n_{LS}=5$, validating the effectiveness of the late-start strategy.

\blue{\textbf{Inner iteration number $T$ and inner step size $\beta_{\bm{\epsilon}}$.} The hyper-parameters $T$ and $\beta_{\bm{\epsilon}}$ jointly control how far the reverse perturbation inner loop explores from the current AE when finding the local worst-case point. As observed in Figures \ref{fig:para_sensiti}(b) and (c), starting from $T=0$ or $\beta_{\bm{\epsilon}}=0$, \ie, \bluetwo{with} no reverse step, increasing either of them enables the inner maximization to explore a larger neighborhood around the current AE, which gradually enforces local flatness and improves transferability. However, when $T$ or $\beta_{\bm{\epsilon}}$ becomes too large, the worst-case can drift too far away from the current AE, so minimizing the loss at that point no longer effectively regularizes the sharpness around the current AE, leading to a decline in attack performance.} 

\blue{\textbf{Decay factor $\mu$.} As discussed in Section \ref{sec:Find a Flat Optimum from a Diverse Set of Surrogates}, due to the diversity in the loss landscapes of different surrogate models in DRAP, accumulating a velocity vector in the gradient across iterations helps to stabilize the optimization path. Figure \ref{fig:para_sensiti}(d) shows that the best performance is achieved at $\mu=1.0$, which corresponds to equally aggregating gradients from all previous surrogate models to perform the current update.}

}

\vfill

\end{document}